%% file: tensorTrace.tex
\documentclass[11pt]{article}

\input{headers}

\title{
    Hutchinson's Estimator is Bad at Kronecker-Trace-Estimation
}

\author{
	Raphael A. Meyer \\ California Institute of Technology\\ \texttt{ram900@caltech.edu}
	\and
	Haim Avron\\ Tel Aviv University\\ \texttt{haimav@tauex.tau.ac.il}
}

\begin{document}
\maketitle

\input{abstract}
\input{intro}

\input{prelims}
\input{partial-trace}
\input{trace-estimation}
\input{psd-trace-estimation}
\input{complex-oracle}
\input{matvec-lower-bound}
\input{kronecker-matrix-recovery}

\input{conclusion}
\input{acks}

\bibliographystyle{apalike}
\bibliography{local,morebibs}

\end{document}

%% file: headers.tex
\usepackage[margin=1in,footskip=0.25in]{geometry}
\usepackage[colorlinks=true, linkcolor=red, urlcolor=blue, citecolor=gray]{hyperref}
\usepackage[usenames,dvipsnames,svgnames,table]{xcolor}

\usepackage{algorithm}
\usepackage{algorithmic}
\usepackage{amsfonts,amssymb,amsthm,amsmath}
\usepackage{bbm}
\usepackage{booktabs}
\usepackage{caption}
\usepackage{color}
\usepackage[capitalize, nameinlink]{cleveref}
\usepackage{commath}
\usepackage{enumitem}
\usepackage{float}
\usepackage{graphicx}
\usepackage{mathtools}
\usepackage{nicefrac}       % compact symbols for 1/2, etc.
\usepackage{subcaption}
\usepackage{xspace}
\usepackage{url}
\usepackage{breakcites}

\usepackage{crossreftools}
\usepackage{contour}
\usepackage[normalem]{ulem}

\contourlength{0.8pt}

\usepackage{todonotes}

\makeatletter
\def\hlinewd#1{%
	\noalign{\ifnum0=`}\fi\hrule \@height #1 \futurelet
	\reserved@a\@xhline}
\makeatother

\newtheorem{theorem}{Theorem}

\newtheorem{corollary}[theorem]{Corollary}
\newtheorem{lemma}[theorem]{Lemma}

\newtheorem{definition}[theorem]{Definition}

\newtheorem{importedtheorem}[theorem]{Imported Theorem}

\makeatletter
\newtheorem*{rep@theorem}{\rep@title}
\newcommand{\newreptheorem}[2]{%
\newenvironment{rep#1}[1]{%
 \def\rep@title{\Cref{##1} Restated}%
 \begin{rep@theorem}}%
 {\end{rep@theorem}}}
\makeatother
\makeatletter
\newtheorem*{rep@lemma}{\rep@title}
\newcommand{\newreplemma}[2]{%
\newenvironment{rep#1}[1]{%
 \def\rep@title{\Cref{##1} Restated}%
 \begin{rep@lemma}}%
 {\end{rep@lemma}}}
\makeatother
\newreptheorem{theorem}{Theorem}
\newreptheorem{corollary}{Corollary}
\newreplemma{lemma}{Lemma}

\newcommand\numberthis{\addtocounter{equation}{1}\tag{\theequation}}

\newcommand{\defeq}[0]{\ensuremath{\;{\vcentcolon=}\;}\xspace}

\let\norm\relax
\newcommand{\norm}[1]{\enVert[0]{#1}}

\DeclareMathOperator*{\Var}{Var}
\DeclareMathOperator*{\E}{\mathbb{E}}
\DeclareMathOperator{\tr}{tr}
\DeclareMathOperator{\diag}{diag}

\colorlet{todo_background_normal}{white}
\definecolor{todo_background_dark}{RGB}{39,40,34}

\definecolor{advice_text}{RGB}{78, 12, 123}
\colorlet{advice_background}{todo_background_normal}

\definecolor{incomplete_text}{RGB}{204, 64, 84}
\colorlet{incomplete_background}{todo_background_normal}

\newcommand{\eps}[0]{\ensuremath{\varepsilon}}
\let\epsilon\eps

\newcommand{\tsfrac}[2]{{\textstyle\frac{#1}{#2}}}

\usepackage{relsize} % Relsize package needed

\newcommand{\pmrdm}[3]{\ensuremath{\cM_{#1}(#3)}}
\newcommand{\herm}{\textsf{\upshape H}} % Hermitian Transpose
\newcommand{\ptran}[1]{\intercal_{#1}}

\let\abs\relax
\newcommand{\abs}[1]{\ensuremath{\vert{#1}\vert}\xspace} % Absolute Value
\newcommand{\mat}[1]{\mathbf{#1}} % Matrix
\renewcommand{\vec}[1]{\boldsymbol{\mathrm{#1}}} % Vector
\newcommand{\vecalt}[1]{\boldsymbol{#1}} % Vector, but better for greek letters
\newcommand{\bmat}[1]{\begin{bmatrix} #1 \end{bmatrix}} % [   Bracket Matrix   ]
\newcommand{\sbmat}[1]{\left[\begin{smallmatrix} #1 \end{smallmatrix}\right]} % Small bracket matrix

\newcommand{\mA}{\ensuremath{\mat{A}}\xspace}
\newcommand{\mB}{\ensuremath{\mat{B}}\xspace}
\newcommand{\mC}{\ensuremath{\mat{C}}\xspace}
\newcommand{\mD}{\ensuremath{\mat{D}}\xspace}
\newcommand{\mE}{\ensuremath{\mat{E}}\xspace}

\newcommand{\mG}{\ensuremath{\mat{G}}\xspace}

\newcommand{\mI}{\ensuremath{\mat{I}}\xspace}

\newcommand{\mP}{\ensuremath{\mat{P}}\xspace}

\newcommand{\mV}{\ensuremath{\mat{V}}\xspace}
\newcommand{\mW}{\ensuremath{\mat{W}}\xspace}

\newcommand{\mDelta}{\ensuremath{\mat{\Delta}}\xspace}

\newcommand{\ve}{\ensuremath{\vec{e}}\xspace}

\newcommand{\vg}{\ensuremath{\vec{g}}\xspace}

\newcommand{\vm}{\ensuremath{\vec{m}}\xspace}

\newcommand{\vr}{\ensuremath{\vec{r}}\xspace}

\newcommand{\vu}{\ensuremath{\vec{u}}\xspace}
\newcommand{\vv}{\ensuremath{\vec{v}}\xspace}

\newcommand{\vx}{\ensuremath{\vec{x}}\xspace}
\newcommand{\vy}{\ensuremath{\vec{y}}\xspace}
\newcommand{\vz}{\ensuremath{\vec{z}}\xspace}

\newcommand{\vsigma}{\ensuremath{\vecalt{\sigma}}\xspace}

\newcommand{\cD}{\ensuremath{{\mathcal D}}\xspace}

\newcommand{\cM}{\ensuremath{{\mathcal M}}\xspace}
\newcommand{\cN}{\ensuremath{{\mathcal N}}\xspace}

\newcommand{\cS}{\ensuremath{{\mathcal S}}\xspace}

\newcommand{\cV}{\ensuremath{{\mathcal V}}\xspace}
\newcommand{\cW}{\ensuremath{{\mathcal W}}\xspace}

\newcommand{\bbC}{\ensuremath{{\mathbb C}}\xspace}

\newcommand{\bbN}{\ensuremath{{\mathbb N}}\xspace}

\newcommand{\bbR}{\ensuremath{{\mathbb R}}\xspace}

%% file: abstract.tex
% Clean

\begin{abstract}

We study the problem of estimating the trace of a matrix \mA that can only be accessed through Kronecker-matrix-vector products.
That is, for any Kronecker-structured vector \(\vx = \otimes_{i=1}^k \vx_i\), we can compute \(\mA\vx\).
We focus on the natural generalization of Hutchinson's Estimator to this setting, proving tight rates for the number of matrix-vector products this estimator needs to find a \((1\pm\eps)\) approximation to the trace of \mA.

We find an exact equation for the variance of the estimator when using a Kronecker of Gaussian vectors, revealing an intimate relationship between Hutchinson's Estimator, the partial trace operator, and the partial transpose operator.
Using this equation, we show that when using real vectors, in the worst case, this estimator needs \(O(\nicefrac{3^k}{\eps^2})\) products to recover a \((1\pm\eps)\) approximation of the trace of any PSD \mA, and a matching lower bound for certain PSD \mA.
However, when using complex vectors, this can be exponentially improved to \(\Theta(\nicefrac{2^k}{\eps^2})\).
Further, if the \(\vx_i\) vectors are low-dimensional and if we instead build \(\vx\) as the Kronecker product of (scaled) random unit vectors on the complex sphere, then as few as \(\nicefrac{1.33^k}{\eps^2}\) samples suffice.
We show that Hutchinson's Estimator converges slowest when \mA itself also has Kronecker structure.
We conclude with some theoretical evidence suggesting that, by combining Hutchinson's Estimator with other techniques, it may be possible to avoid the exponential dependence on \(k\).

\end{abstract} 

%% file: intro.tex
% Clean

\section{Introduction}
\label{sec:intro}

A ubiquitous problem in scientific computing is that of estimating the trace of a \(D \times D\) matrix \mA that can be only accessed via matrix-vector products.
In other words, we are given access to an oracle that can evaluate \(\mA\vx\) for any \(\vx\in\bbR^D\), and our goal is to approximate \(\tr(\mA)\) with as few oracle queries as possible.
This \emph{matrix-vector oracle model} or \emph{implicit matrix model} is a common computational framework used in much of the numerical linear algebra community \cite{sun2021querying,chen2022krylov,halikias2022matrix,bakshi2022low,persson2022improved}.
However, in many applications, we cannot efficiently evaluate \(\mA\vx\) for all vectors \(\vx\in\bbR^D\).
Instead, only some matrix vector products can be efficiently computed.

In this paper, we study trace estimation in the \emph{Kronecker-matrix-vector oracle model}, also known as the \emph{rank-one vector model}, where we can evaluate \(\mA\vx\) for any vector \(\vx\) that can can be written as a Kronecker product of smaller vectors, i.e. \(\vx = \vx_1 \otimes \vx_2 \otimes \cdots \otimes \vx_k\) for some vectors \(\vx_1,\ldots,\vx_k \in \bbR^d\).
Here, \(\otimes\) denotes the \emph{Kronecker product}, which is a common primitive in many tensor-based algorithms and in much of quantum physics.
The number of vectors (\(k\)) and the number of entries (\(d\)) in each of \(\vx_1,\dots,\vx_k\) are parameters of the model.

Our study is motivated by recent literature on entanglement estimation in quantum states~ \cite{feldman2022entanglement}.
This computational model has also been studied in other applied domains \cite{chen2022krylov, kueng2017low}.
An important setting where we can efficiently compute Kronecker-matrix-vector products  is when \mA is the sum of a small number of matrices, where each summand is a Kronecker product of \(k\) matrices.
Such matrices occur as Hamiltonians of various quantum models~\cite{wang23qubit}.
Another instance in which we can efficiently compute Kronecker-matrix-vector products is when \mA is the flattening of a tensor network.
This is the case in \cite{feldman2022entanglement}, where the oracle is efficiently implemented when \mA describes a density operator that represents certain systems of quantum particles.
Recently, there has been some effort to find algorithms that can efficiently solve linear algebra problems using this oracle, although the papers do not explicitly define such a Kronecker-matrix-vector oracle \cite{bujanovic2021norm,avron2014subspace,ahle2020oblivious,sun2021tensor,bujanovic2024subspace}.
In this paper, we focus on the fundamental problem of trace estimation.

The trace estimation problem in the Kronecker-matrix-vector oracle model can be exactly solved using \(D = d^k\) oracle queries.
Since each standard basis vector \(\ve_1,\ldots,\ve_D\) obeys Kronecker structure, we can exactly look up each diagonal entry of \(\mA\) and evaluate \(\tr(\mA) = \sum_{i=1}^D \ve_i^\intercal\mA\ve_i\).
So, our goal is to estimate the trace of \mA to \((1\pm\eps)\) relative error using \(o(d^k)\) matrix-vector products.
In the standard non-Kronecker setting, the most fundamental algorithm for trace estimation is Hutchinson's Estimator \cite{girard1987algorithme,hutchinson1989stochastic}:
\[
	H_\ell(\mA) \defeq \frac1{\ell}\sum_{j=1}^{\ell} {\vx^{(j)}}^{\intercal} \mA\vx^{(j)}
	\hspace{1cm} \text{where} \hspace{1cm}
	\vx^{(j)} \overset{iid}{\sim} \cD.
\]
If the distribution of random vectors has \(\E[\vx] = \vec 0\) and \(\E[\vx\vx^\intercal] = \mI\), then \(\E[H_{\ell}(\mA)] = \tr(\mA)\).
Further, if \(\cD\) is a Gaussian distribution, a Rademacher distribution, or the uniform distribution over the sphere of radius \(\sqrt d\), then \(\Var[H_\ell(\mA)] \leq \frac{2}{\ell}\norm{\mA}_F^2\).
In particular, under the common assumption that \mA is PSD, so that \(\norm{\mA}_F \leq \tr(\mA)\), we then find that the standard deviation of Hutchinson's estimator is
\[
	\sqrt{\Var[H_\ell(\mA)]} \leq \sqrt{\tsfrac2\ell} \norm{\mA}_F \leq \sqrt{\tsfrac2\ell} \tr(\mA) \leq \eps\tr(\mA)
\]
when we take \(\ell \geq \Omega(\frac1{\eps^2})\) queries.
For large dimension \(D\), this very simple algorithm is much more efficient than exactly computing the diagonal entries of \mA, since the estimator has no dependence on the size of \mA at all.

This estimator can be naturally generalized to the Kronecker-matrix-vector oracle model by simply changing the distribution \cD.
Instead of taking \cD to be a distribution over Gaussian or Rademacher vectors, we now take \cD to be the Kronecker of \(k\) vectors drawn from some other isotropic zero-mean distribution \(\cD'\):
\[
	\vx = \vx_1 \otimes \vx_2 \otimes \cdots \otimes \vx_k
	\hspace{1cm} \text{where} \hspace{1cm}
	\vx_{i} \overset{iid}{\sim} \cD'
\]
Notably, if \(\E[\vx_i] = \vec0\) and \(\E[\vx_i\vx_i^\intercal]=\mI\), then we also immediately have that \(\E[\vx] = \vec0\) and \(\E[\vx\vx^\intercal]=\mI\), so that \(\E[H_\ell(\mA)] = \tr(\mA)\).
This simple estimator, which we refer to as the Kronecker-Hutchinson Estimator, fits the Kronecker-matrix-vector oracle model, and we may hope that it is efficient enough to estimate the trace of \mA very quickly.
At a high level, we answer this in the negative, proving novel and very tight bounds on the variance of the Kronecker-Hutchinson Estimator, showing that the sample complexity of the estimator has an unavoidable exponential dependence on \(k\).

    However, we also show that we can significantly improve this exponential dependence by changing our base distribution \(\cD'\).
    In the non-Kronecker matrix-vector case, there are a few common choices of \(\cD'\) (i.e. Gaussian, Rademacher, uniform on the sphere), all of which only change our sample complexity by constant factors.
    However, in the Kronecker case, the difference in choice of \(\cD'\) can be the difference of a sample complexity that scales as \(3^k\) versus \(1.33^k\).

\subsection{Prior Work}
We first remark there is a great deal of literature on trace estimation without the Kronecker constraint.
The design of Hutchinson's Estimator is attributed to \cite{girard1987algorithme,hutchinson1989stochastic}, and some of the core papers that analyze the estimator are \cite{AT11,roosta2015improved,cortinovis2021randomized}.

Although relatively a new topic, trace estimation in the Kronecker-matrix-vector oracle model has recently been explored in the literature.
From the physics literature, the works of \cite{feldman2022entanglement} and \cite{huang2020predicting} analyze trace estimators when \(\mA = \otimes_{i=1}^k \mA_i\), and shows some empirical results for their performance.
From a theoretical perspective, \cite{vershynin2020concentration} provides bounds that can be used to analyze the Kronecker-Hutchinson Estimator.
However, the paper provides very coarse bounds for this problem, showing that if \(k = O(d)\) then we have \(\vx^\intercal\mA\vx \in \tr(\mA) \pm O(kd^{k-1}\norm{\mA}_2)\), where \(\norm{\mA}_2\) is the operator norm of \mA.
To obtain a relative error \((1\pm\eps)\), this suggests an overall sample complexity of \(\ell = O(\frac{kd^{k-1}}{\eps^2} \frac{\norm{\mA}_2}{\tr(\mA)}) = O(\frac{D}{\eps^2} \frac{\norm{\mA}_2}{\tr(\mA)})\), which is generally more than \(D\) queries.

The research of \cite{bujanovic2021norm} directly analyzes the Kronecker-Hutchinson Estimator when \(k=2\), showing first that \(\ell = O(\frac{1}{\eps^2})\) samples suffice to ensure that \(H_\ell(\mA)\geq(1-\eps)\tr(\mA)\), and second that \(\ell = O(\frac d{\eps^2} \frac{\norm{\mA}_2}{\tr(\mA)})\) samples suffice to ensure that \(H_\ell(\mA) \leq (1+\eps)\tr(\mA)\).
Lastly, Remark 2 of the arXiv version of \cite{ahle2020oblivious} is a result on using a Kronecker-Johnson-Lindenstrauss transform, and implies that \(O(\frac{3^k}{\eps^2})\) samples suffice for the Kronecker-Hutchinson Estimator with Rademacher vectors to achieve relative error \((1\pm\eps)\).
While this worst-case bound is tight, the paper gives no way to understand when this bound is lose, and for what matrices fewer samples are needed.
Further, all the aforementioned theoretical papers rely on heavy-duty analysis, for instance involving Gaussian and Rademacher Chaos theory.
In contrast, our contributions give tighter results using simpler techniques.

We also note that \cite{ahle2020oblivious} does also include a Johnson-Lindenstrauss style guarantee for working with Kronecker-shaped data using a sketch that has \(O(\frac k{\eps^2})\) rows.
However, this sketch does not actually operate in the Kronecker-matrix-vector oracle model.

\subsection{Our Contributions}
\label{sec:our-contributions}

Our core technical contributions provide a new understanding of the sample complexity of the Kronecker-Hutchinson Estimator when our query vectors are formed as the Kronecker product of iid vectors from some base distribution \(\cD'\).
First, we derive the exact expression of the variance of the estimator when our base distribution \(\cD'\) is Gaussian.
This also provides an exact variance when \(\cD'\) is uniformly distributed on the sphere, and this expression suffices to upper bound the variance when the estimator instead uses a Kronecker product of Rademacher vectors.
For PSD matrices, our results imply a worst-case upper bound of \(\ell = O(\frac{3^k}{\eps^2})\) on the number of samples needed to estimate the trace of \mA to relative error \((1\pm\eps)\) when \(\cD'\) is Gaussian.

We also provide a coarser result, showing how we can relate any base distribution \(\cD'\) to the worst-case sample complexity \(\ell\) needed for PSD matrix trace estimation.
When \(\cD'\) is either Rademacher or uniformly distributed on the sphere, then we converge exponentially faster than in the Gaussian case when \(d=o(k)\).
We match these upper bounds with a lower bound showing that the Kronecker-Hutchinson Estimator with iid zero-mean unit-variance vectors has to use \(\ell = \Omega(\frac{1}{\eps^2}(3-\frac2d)^k)\) samples to achieve standard deviation \(\eps\tr(\mA)\), exactly matching the upper bound when \(\cD'\) is Rademacher.
Lastly, we show that if we are allowed to use complex Kronecker-matrix-vector products (i.e. if \(\vx = \vx_1 \otimes \cdots \otimes \vx_k\) where \(\vx_i \in \bbC^{d}\)), then for reasons discussed in \Cref{sec:complex-matvec-intro} our computational fundamentally shifts.
This increase in expressive power allows the Gaussian sample complexity to improve to \(\Theta(\frac{2^k}{\eps^2})\), with even better rates when \(d = o(k)\) for complex Rademacher and uniformly distributed vectors on the complex sphere of radius \(\sqrt d\).
In the nicest case, when \(d=2\) and we are allowed to use complex vectors, by taking \(\cD'\) to be uniform on the sphere, we get a sample complexity of \(\ell = O(\frac{1.33^k}{\eps^2})\).

\begin{table}[H]
\centering
\begin{tabular}{@{}clccc@{}} \toprule
	& & \multicolumn{3}{c}{Worst-Case PSD \mA}  \\ \cmidrule(lr){3-5}
	Field
		& Base Distribution
		& Arbitrary \(d\)
		& \(d=2\)
		& \(d=\Omega(k)\)
		\\ \midrule
	\(\bbR\)
		& Gaussian
		& \(3^k\)
		& \(3^k\)
        & \(3^k\) 
		\\[0.25em]
		& Rademacher
		& \((3-\frac2d)^k\)
		& \(2^k\)
        & \(3^k\)
		\\[0.25em]
		& Uniform on Sphere
		& \((3-\frac{6}{d+2})^k\)
		& \(1.5^k\)
		& \(3^k\)
		\\[0.25em] 
	\(\bbC\)
		& Gaussian
		& \(2^k\)
		& \(2^k\)
        & \(2^k\)		
		\\[0.25em]
		& Rademacher
		& \((2-\frac1d)^k\)
		& \(1.5^k\)
        & \(2^k\)
		\\[0.25em]
		& Uniform on Sphere
		& \((2-\frac{2}{d+1})^k\)
		& \(1.33^k\)
		& \(2^k\)
		\\[0.25em]
		\bottomrule
\end{tabular}
\caption{
    Summary of sample complexities for estimating the trace of a PSD matrix \mA to relative error \((1\pm\eps)\) when \(\eps = O(1)\), under different choices of base distribution \(\cD'\).
    The sample complexities all take \(\mA\) to be a worst-case input matrix.
    The first three rows constrain us to take \(\vx\in\bbR^{d^k}\), while the last three rows takes \(\vx\in\bbC^{d^k}\).
    The first column expresses the exact worst-case sample complexity for arbitrary \(d\) and \(k\), while the second and third columns specialize to the settings where \(d=2\) or \(d = \Omega(k)\).
    We see that when \(d=2\), in both the real and complex cases, Rademacher vectors and uniformly sampled vectors from the sphere exponentially outperform Gaussian vectors.
    However, when \(d = \Omega(k)\), this advantage dissapears.
    These results are formalized in \Cref{thm:real-psd} in the real case and \Cref{thm:complex-psd} in the complex case.
}
\label{table:exp-rates}
\end{table}

We now turn our attention to more concretely stating our technical results, starting with the exact variance of the Kronecker-Hutchinson Estimator when \(\cD'\) is Gaussian.
We first recall the \emph{partial trace} operator from quantum physics.
In quantum physics, a Hermitian \emph{density matrix} \(\mA\in\bbC^{d^k \times d^k}\) is roughly speaking used to represent the joint probability distribution for a set of \(k\) particles.
Suppose these particles are partitioned into two sets, so that \([k] = \cS \cup \bar\cS\).
Further suppose that Alice can observe the particles in \cS, and that Bob can observe the particles in \(\bar\cS\), but they cannot observe each other's particles.
Then, the partial trace of \mA with respect to \cS, denoted \(\tr_{\cS}(\mA)\), describes the marginal distribution of the particles that Bob observes.
This is often called ``tracing out the subsystem \cS''.

We include a general formal definition of the partial trace in \Cref{sec:prelims}, but a rough intuition is gained by examining the \(k=2\) case.
Here, we can break down \(\mA = \bbR^{d^2 \times d^2}\) as a block matrix full of \(\mA_{ij}\in\bbR^{d \times d}\):
\[
	\mA = \bmat{
		\mA_{11} & \mA_{12} & \ldots & \mA_{1d} \\
		\mA_{21} & \mA_{22} & \ldots & \mA_{2d} \\
		\vdots & \vdots & \ddots & \vdots \\
		\mA_{d1} & \mA_{d2} & \ldots & \mA_{dd} \\
	}.
\]
Then, the partial trace of \mA with respect to \(\{1\}\) is the sum of the diagonal blocks, i.e. \(\tr_{\{1\}}(\mA) = \sum_{i=1}^d \mA_{ii}\), and the partial trace of \mA with respect to \(\{2\}\) is the matrix that contains the trace of each block matrix, i.e.
\[
	\tr_{\{2\}}(\mA) = \bmat{
		\tr(\mA_{11}) & \tr(\mA_{12}) & \ldots & \tr(\mA_{1d}) \\
		\tr(\mA_{21}) & \tr(\mA_{22}) & \ldots & \tr(\mA_{2d}) \\
		\vdots & \vdots & \ddots & \vdots \\
		\tr(\mA_{d1}) & \tr(\mA_{d2}) & \ldots & \tr(\mA_{dd}) \\
	}.
\]
We also need to recall the \emph{partial transpose} operator from quantum physics.
This operator is often used in quantum physics, but lacks an intuitive physical meaning analogous to what the partial trace enjoys.
Like the partial trace, the partial transpose can also be taken over any subset of particles \(\cV \subseteq [k]\).
We denote the partial transpose of \mA with respect to \cV as \(\mA^{\ptran{\cV}}\), and call this ``transposing the subsystem \cV''.
We have a formal definition of the partial transpose in \Cref{sec:prelims}.
When \(k=2\), the two partial transposes of \mA are:
\[
	\mA^{\ptran{\{1\}}}
	= \bmat{
		\mA_{11} & \mA_{21} & \cdots & \mA_{d1} \\
		\mA_{12} & \mA_{22} & \cdots & \mA_{d2} \\
		\vdots & \vdots & \ddots & \vdots \\
		\mA_{1d} & \mA_{2d} & \cdots & \mA_{dd}
	}
	\hspace{0.5cm}\text{and}\hspace{0.5cm}
	\mA^{\ptran{\{2\}}}
	= \bmat{
		\mA_{11}^\intercal & \mA_{12}^\intercal & \cdots & \mA_{1d}^\intercal \\
		\mA_{21}^\intercal & \mA_{22}^\intercal & \cdots & \mA_{2d}^\intercal \\
		\vdots & \vdots & \ddots & \vdots \\
		\mA_{d1}^\intercal & \mA_{d2}^\intercal & \cdots & \mA_{dd}^\intercal
	}.
\]
Having setup both the partial trace and the partial transpose, we can now state our first theorem:
\begin{theorem}
\label{thm:real-partial-trace-transpose}
	Fix \(d,k\in\bbN\), and let \(\mA\in\bbR^{d^k \times d^k}\).
	Let \(\vx = \vx_1 \otimes \cdots \otimes \vx_k\), where each \(\vx_i\) is independent and is either a Rademacher vector or a \(\cN(\vec0,\mI)\) Gaussian vector.
	Let \(\bar\mA \defeq \frac1{2^k} \sum_{\cV \subseteq [k]} \mA^{\ptran\cV}\) be the average of all partial transposes of \mA.
	Then,
	\(\Var[\vx^\intercal\mA\vx] \leq \sum_{\cS\subset[k]} 2^{k - \abs{\cS}} \norm{\tr_{\cS}(\bar\mA)}_F^2\).
	Furthermore, if all \(\vx_i\) are \(\cN(\vec0,\mI)\), then this expression is an equality.
\end{theorem}
\Cref{thm:real-partial-trace-transpose} is proven in \Cref{sec:exact-var-proof}.
The matrix \(\bar\mA\) seems odd at a glance, so we first justify why it is natural.
For the classical \(k=1\) setting with Gaussian vectors, we can write that \(\Var[\vx^\intercal\mA\vx] = 2\norm{\mA}_F^2\) for all symmetric matrices \mA \cite{AT11,ram900}.
For asymmetric matrices, this equality does not hold.
Instead, we let \(\hat\mA \defeq \frac{\mA+\mA^\intercal}{2}\), and note that \(\vx^\intercal\mA\vx = \vx^\intercal\hat\mA\vx\) for all \(\vx\in\bbR^D\).
Since \(\hat\mA\) is symmetric, we then find that \(\Var[\vx^\intercal\mA\vx] = 2\norm{\hat\mA}_F^2\).
\Cref{thm:real-partial-trace-transpose} generalizes this idea to all possible ways to partially transpose across all possible subsystems \(\cV\subseteq[k]\).

Next, by the triangle inequality, note that \(\norm{\hat\mA}_F \leq \norm{\mA}_F\), and that symmetric matrices have \(\mA=\hat\mA\).
So, we know that \(\Var[\vx^\intercal\mA\vx] = 2\norm{\hat\mA}_F^2 \leq 2\norm{\mA}_F^2\) holds for all matrices \mA, with equality for symmetric \mA.
This bound can be easier to work with, since it does not require worrying about \(\hat\mA\).
For large \(k\), we similarly bound \(\norm{\tr_\cS(\bar\mA)}_F \leq \norm{\tr_\cS(\mA)}_F\), providing a more approachable version of \Cref{thm:real-partial-trace-transpose}:
\begin{theorem}
\label{thm:real-partial-trace}
	Fix \(d,k\in\bbN\), and let \(\mA\in\bbR^{d^k \times d^k}\).
	Let \(\vx = \vx_1 \otimes \cdots \otimes \vx_k\), where each \(\vx_i\) is independent and is either a Rademacher vector or a \(\cN(\vec0,\mI)\) Gaussian vector.
	Then,
	\(\Var[\vx^\intercal\mA\vx] \leq \sum_{\cS\subset[k]} 2^{k - \abs{\cS}} \norm{\tr_{\cS}(\mA)}_F^2\).
	Furthermore, if all \(\vx_i\) are \(\cN(\vec0,\mI)\) and if \(\mA=\bar\mA\), then this expression is an equality.
\end{theorem}
\Cref{thm:real-partial-trace} is proven in \Cref{sec:proof-of-partial-trace}.
Note that many reasonable matrices will have \(\mA = \bar\mA\).
For instance, recall that the Kronecker-matrix-vector oracle model may be used when \mA is the sum of a small number of matrices, where each summand is a Kronecker product of \(k\) matrices.
If each summand is the Kronecker product of \emph{symmetric} matrices, then we have \(\mA = \bar\mA\).

We also briefly note that since a Gaussian vector can be decomposed as the product of a $\chi^2$-distributed random variable and a uniformly distributed random unit vector, we can use \Cref{thm:real-partial-trace-transpose} to exactly express the variance of the Kronecker-Hutchinson estimator when \(\cD'\) is uniform from the sphere:
\begin{corollary}
\label{corol:random-unit-vec-variance}
Fix \(d,k\in\bbN\) and let \(\mA\in\bbR^{d^k \times d^k}\).
Let \(\vu = \vu_1 \otimes \cdots \otimes \vu_k\) where each \(\vu_i\) is uniformly at random drawn from the sphere such that \(\norm{\vu_i}_2 = \sqrt d\).
Let \(\vx = \vx_1 \otimes \cdots \otimes \vx_k\) where each \(\vx_i\) is \(\cN(\vec0,\mI)\).
Then, \(\Var[\vu^\intercal\mA\vu] = \frac{\Var[\vx^\intercal\mA\vx] - ((1+\frac2d)^k-1)(\tr(\mA))^2}{(1+\frac2d)^k} \leq \frac1{(1+\frac2d)^k} \Var[\vx^\intercal\mA\vx]\).
\end{corollary}

\Cref{corol:random-unit-vec-variance} is proven in \Cref{sec:exact-var-unit-vec}.
To achieve \((1\pm\eps)\) relative error guarantees, we now add the assumption that \(\mA\) is PSD.
We start by showing a simple result that suffices to give tight worst-case bounds on the sample complexity of this trace estimation problem:
\begin{theorem}
\label{thm:coarse-psd-rate}
    Let \(\cD'\) be a distribution over vectors \(\vx_i\in\bbR^d\) such that \(\E[\vx_i]=\vec0\), \(\E[\vx_i\vx_i^\intercal] = \mI\), and \(\Var[\vx_i^\intercal\mB\vx_i] \leq C(\tr(\mB))^2\) for all PSD \mB.
    Then, letting \(\vx = \vx_1 \otimes \cdots \otimes \vx_k\), we have \(\Var[\vx^\intercal\mA\vx] \leq (1+C)^k (\tr(\mA))^2\).
    In particular, taking \(\ell = O(\frac{(1+C)^k}{\eps^2})\) suffices for \(H_\ell(\mA)\) to have standard deviation at most \(\eps\tr(\mA)\).
\end{theorem}

\Cref{thm:coarse-psd-rate} is proven in \Cref{sec:real-psd}.
By plugging in various value of \(C\) associated with different distributions \(\cD'\), we get a plethora of different results, summarized both the following theorem and in \Cref{table:exp-rates}.
\begin{theorem}
\label{thm:real-psd}
Fix \(d,k\in\bbN\), and let \(\mA\in\bbR^{d^k \times d^k}\) be a PSD matrix.
Let \(\vx = \vx_1 \otimes \cdots \otimes \vx_k\) where each \(\vx_i\) is drawn iid from a distribution \(\cD'\).
If \(\cD'\) is Gaussian, then \(\Var[\vx^\intercal\mA\vx] \leq 3^k (\tr(\mA))^2\).
If \(\cD'\) is Rademacher, then \(\Var[\vx^\intercal\mA\vx] \leq (3-\frac2d)^k (\tr(\mA))^2\).
If \(\cD'\) is uniformly distributed over vectors with \(\norm{\vx_i}_2=\sqrt{d}\), then \(\Var[\vx^\intercal\mA\vx] \leq (3-\frac6{d+2})^k (\tr(\mA))^2\).
\end{theorem}
\Cref{thm:real-psd} is also proven in \Cref{sec:real-psd}.
Notably, as mentioned in \Cref{table:exp-rates}, when we allow \(\cD'\) to be a distribution over complex vectors, our asymptotics improve significantly.
We discuss this further in \Cref{sec:complex-matvec-intro}.
For now, we take a moment to note that when \(k=2\), we recover the case of \cite{bujanovic2021norm}, where we get the following:

\begin{corollary}
	Let \(\mA \in \bbR^{d^2 \times d^2}\).
	Let \(\vx = \vx_1 \otimes \vx_2\), where each \(\vx_i\) is independent and is a Rademacher vector or a \(\cN(\vec0,\mI)\) Gaussian vector or uniformly distributed on the sphere.
	Let \(\bar\mA = \frac14(\mA + \mA^{\ptran{\{1\}}} + \mA^{\ptran{\{2\}}} + \mA^\intercal)\).
	Then, \(\Var[\vx^\intercal\mA\vx] \leq 2 \norm{\tr_{\{1\}}(\bar\mA)}_F^2 + 2\norm{\tr_{\{2\}}(\bar\mA)}_F^2 + 4 \norm{\bar\mA}_F^2\).
	If \mA is PSD, with probability \(\frac23\), we know that the Kronecker-Hutchinson Estimator returns a \((1\pm\eps)\) relative error estimate to the trace of \mA if \(\ell = O(\eps^{-2})\).
\end{corollary}
This success probability can be boosted to \(1-\delta\) for an arbitrary \(\delta\) by picking up a \(\log(\frac1\delta)\) dependence in the sample complexity by returning the median of \(O(\log(\frac1\delta))\) iid trials of \(H_\ell(\mA)\).
We further show that the expensive rate of \(O(\frac{3^k}{\eps^2})\) is essentially tight in the worst case for the Kronecker-Hutchinson Estimator:
\begin{theorem}
\label{thm:real-psd-lower-bound}
	Fix \(d,k\in\bbN\).
	Consider Hutchinson's Estimator run with vectors \(\vx = \vx_1 \otimes \cdots \otimes \vx_k\) where \(\vx_i\) are vectors of iid real zero-mean unit-variance entries.
	Then, there exists a PSD matrix \mA such that Hutchinson's Estimator needs \(\ell = \Omega(\frac{(3-\frac2d)^k}{\eps^2})\) samples in order to have standard deviation \(\leq \eps \tr(\mA)\).
\end{theorem}

\Cref{thm:real-psd-lower-bound} is proven in \Cref{sec:real-psd-lower-bound}.
This worst-case complexity is rather painful.
We note that the typical performance of this estimator may be much better though.
As a first step to truly understanding when the Kronecker-Hutchinson Estimator converges quickly, we show that estimating the trace of a random rank-1 matrix is much faster than this worst-case rate:
\begin{theorem}
\label{thm:real-rank-one-estimation}
Fix \(d,n\in\bbN\).
Let \(\mA=\vg\vg^\intercal\) where \(\vg\in\bbR^{d^k}\) is a vector of iid standard normal Gaussian.
Then if \(\cD'\) is either Rademacher or uniformly distributed on the sphere, on average across the choice \(\vg\), it suffices to take \(\ell=\frac{2}{\eps^2}\) for \(H_\ell(\mA)\) to have standard deviation at most \(\eps\tr(\mA)\).
That is, \(\E_{\vg}[\Var[H_\ell(\mA) \mid \vg]] \leq \eps^2 (\E_{\vg}[\tr(\mA)])^2\).
However, if \(\cD'\) is Gaussian, then \(\ell = \Omega(\frac1{\eps^2}(1+\frac2d)^k)\) is needed to achieve the same guarantee.
\end{theorem}
\Cref{thm:real-rank-one-estimation} is proven in \Cref{sec:real-rank-one-estimation}.
We remark two aspects of this result.
First, for all choices of \(\cD'\) here, we vastly improve upon the sample complexities reflected by the worst-case analysis in \Cref{table:exp-rates}.
Second, we remark that the Rademacher distribution and the uniform distribution on the sphere achieve sample complexity independent of \(k\), while the Gaussian complexity is independent of \(k\) only if \(d = \Omega(k)\).

\subsubsection{Complex Kronecker-Matrix-Vector Products}
\label{sec:complex-matvec-intro}

Although certainly cheaper than exactly computing the trace with \(d^k\) queries to the diagonal of \mA, the rate of \(\frac{3^k}{\eps^2}\) is prohibitively expensive for large \(k\).
Our next set of results show that the bound improves if we can issue complex queries instead of real ones, that is when \(\vx = \vx_1 \otimes \cdots\otimes \vx_k\) where \(\vx_i \in \bbC^d\).
It seems realistic to say that any software implementation of the Kronecker-matrix-vector oracle could handle complex queries instead of real ones.

Note that for \(k=1\), there is almost no distinction between the real and complex oracle models, as any complex matrix-vector product can be simulated by running two real matrix-vector products.
However, this is no longer the case for large values of \(k\) and real simulation becomes expensive.
To see this, note that the real part of the Kronecker product of complex vectors might not have a Kronecker structure.
In fact, it may take \(2^k\) real-valued Kronecker-matrix-vector products to simulate a single complex-valued Kronecker-matrix-vector product.
In a vague intuitive sense, the complex Kronecker-structured vectors are exponentially less constrained than the real Kronecker-structured vectors.
Running with this intuition in mind, we prove the following analogues to \Cref{thm:real-partial-trace-transpose,thm:real-partial-trace,thm:real-psd,thm:real-rank-one-estimation,thm:real-psd-lower-bound}, all proven in \Cref{sec:complex-matvec}:
\begin{theorem}
\label{thm:complex-upper-bounds}
	Fix \(d,k\in\bbN\), and let \(\mA\in\bbR^{d^k \times d^k}\).
	Let \(\vx = \vx_1 \otimes \cdots \otimes \vx_k\), where each \(\vx_i\) is either a complex Rademacher vector or a complex Gaussian vector.
	That is, either we have \(\vx_i = \frac1{\sqrt2}\vr_i + \frac{i}{\sqrt2}\vm_i\) where \(\vr_i\) and \(\vm_i\) are \(\cN(\vec0,\mI)\) vectors, or each entry of \(\vx_i\) is drawn uniformly iid from \(\{\pm1,\pm i\}\).
	Let \(\bar\mA \defeq \frac1{2^k} \sum_{\cV \subseteq [k]} \mA^{\ptran\cV}\) be the average of all partial transposes of \mA.
	Then,
	\(\Var[\vx^\herm\mA\vx] \leq \sum_{\cS\subset[k]} \norm{\tr_{\cS}(\bar\mA)}_F^2 \leq \sum_{\cS\subset[k]} \norm{\tr_{\cS}(\mA)}_F^2\).
	If \(\vx_i\) are complex Gaussians, then the first expression is an equality.
	If we further know that \(\mA = \bar\mA\), then both expression are equalities.
	Further, if \(\vu = \vu_1 \otimes \cdots \otimes \vu_k\) where \(\vu_i\in\bbC^d\) are uniformly random vectors with \(\norm{\vu_i}_2=\sqrt d\), then \(\Var[\vu^\herm\mA\vu] = \frac{\Var[\vx^\herm\mA\vx]-((1+\frac1d)^k-1)(\tr(\mA))^2}{(1+\frac1d)^k} \leq \frac{1}{(1+\frac1d)^k} \Var[\vx^\herm\mA\vx]\).
\end{theorem}
\begin{theorem}
\label{thm:complex-psd}
Fix \(d,k\in\bbN\), and let \(\mA\in\bbR^{d^k \times d^k}\) be a PSD matrix.
Let \(\vx = \vx_1 \otimes \cdots \otimes \vx_k\) where each \(\vx_i\) is drawn iid from a distribution \(\cD'\).
If \(\cD'\) is complex Gaussian, then \(\Var[\vx^\intercal\mA\vx] \leq 2^k (\tr(\mA))^2\).
If \(\cD'\) is complex Rademacher, then \(\Var[\vx^\intercal\mA\vx] \leq (2-\frac1d)^k (\tr(\mA))^2\).
If \(\cD'\) is uniformly distributed over complex vectors with \(\norm{\vx_i}_2=\sqrt{d}\), then \(\Var[\vx^\intercal\mA\vx] \leq (2-\frac2{d+1})^k (\tr(\mA))^2\).
\end{theorem}

\begin{theorem}
\label{thm:complex-rank-one-estimation}
	Fix \(d,n\in\bbN\).
	Let \(\mA=\vg\vg^\intercal\) where \(\vg\in\bbR^{d^k}\) is a vector of iid standard normal Gaussian.
	Then if \(\cD'\) is either complex Rademacher or uniformly distributed on the complex sphere, on average across the choice \(\vg\), it suffices to take \(\ell=\frac{2}{\eps^2}\) for \(H_\ell(\mA)\) to have standard deviation at most \(\eps\tr(\mA)\).
	That is, \(\E_{\vg}[\Var[H_\ell(\mA) \mid \vg]] \leq \eps^2 (\E_{\vg}[\tr(\mA)])^2\).
	However, if \(\cD'\) is complex Gaussian, then \(\ell = \Omega(\frac1{\eps^2}(1+\frac1d)^k)\) is needed to achieve the same guarantee.
\end{theorem}
\begin{theorem}
\label{thm:complex-psd-lower-bound}
	Fix \(d,k\in\bbN\).
	Consider Hutchinson's Estimator run with vectors \(\vx = \vx_1 \otimes \cdots \otimes \vx_k\) where \(\vx_i\) are vectors of iid complex zero-mean unit-variance entries.
	Then, there exists a PSD matrix \mA such that Hutchinson's Estimator needs \(\ell = \Omega(\frac{1}{\eps^2}(2-\frac1d)^k)\) samples in order to have standard deviation \(\leq \eps \tr(\mA)\).
\end{theorem}

By applying Chebyshev's Inequality to \Cref{thm:real-psd,thm:complex-upper-bounds}, we see an \emph{exponential speedup} from \(\ell = O(\frac{3^k}{\eps^2})\) queries in the real-valued case to \(\ell = O(\frac{2^k}{\eps^2})\) queries in the complex-valued case.
For large \(k\), this gap from \(3^k\) to \(2^k\) is a very significant speedup, and we are not aware of any other prior works which discuss this potential.

Another way to view this speedup is from comparing the Kronecker-Hutchinson Estimator to exactly summing the diagonal of \mA, which takes exactly \(d^k\) queries.
When \(d=2\), the Gaussian real-valued Kronecker-Hutchinson Estimator is exponentially more expensive than just summing the diagonal directly.
However, the complex-valued Gaussian Kronecker-Hutchinson Estimator is only \(O(\frac{1}{\eps^2})\) times more expensive in the worst case.
This is another way to frame the power of the complex-valued estimator -- the complex-valued Gaussian estimator is at most polynomially slower than summing the diagonal, while the real-valued Gaussian estimator can be exponentially slower.

\subsubsection{Comparing \Cref{thm:real-partial-trace-transpose} and \Cref{thm:real-partial-trace}}

We also want to understand the difference between \Cref{thm:real-partial-trace-transpose} and \Cref{thm:real-partial-trace} in both the real and complex cases.
That is, how much information do we lose when we upper bound \(\norm{\bar\mA}_F \leq \norm{\mA}_F\)?
The analysis of random rank-1 matrices lets us start to understand this.
\Cref{thm:real-rank-one-estimation} and \Cref{thm:complex-rank-one-estimation} show us that the true variance of the Kronecker-Hutchinson Estimator on a random rank-1 matrix with Gaussian query vectors scales as \(d^{2k}(1+\frac2d)^k\) and \(d^{2k}(1+\frac1d)^k\) in the real and complex cases, respectively.
We can also to compute the upper bound to the variance of these estimators according to \Cref{thm:real-partial-trace} and its complex counterpart:
\begin{theorem}
\label{thm:random-rank-one-real-var-overestimate}
Fix \(d,k\in\bbN\).
Let \(\mA=\vg\vg^\intercal\) where \(\vg\in\bbR^{d^k}\) is a vector of iid real standard normal Gaussian.
Let \(\vx\) be the kronecker product of \(k\) standard normal Gaussian vectors.
Then,
\[
	\E_{\vg}\left[\Var[\vx^\intercal\mA\vx \mid \vg]\right]
	\leq
	\E_{\vg}\left[\sum_{\cS\subseteq[k]} 2^{k - \abs\cS} \norm{\tr_\cS(\mA)}_F^2\right]
	\geq 2^k d^{2k} (1+\tsfrac2d)^k.
\]
The second inequality above is tight up to a multiplicative factor of 3.
\end{theorem}
\begin{theorem}
\label{thm:random-rank-one-complex-var-overestimate}
Fix \(d,k\in\bbN\).
Let \(\mA=\vg\vg^\intercal\) where \(\vg\in\bbR^{d^k}\) is a vector of iid complex standard normal Gaussian.
Let \(\vx\) be the kronecker product of \(k\) standard normal Gaussian vectors.
Then,
\[
	\E_{\vg}\left[\Var[\vx^\intercal\mA\vx \mid \vg]\right]
	\leq
	\E_{\vg}\left[\sum_{\cS\subseteq[k]} \norm{\tr_\cS(\mA)}_F^2\right]
	\leq 3 d^{2k} (1+\tsfrac1d)^k.
\]
The second inequality is tight up to a multiplicative factor of 3.
\end{theorem}

\Cref{thm:random-rank-one-real-var-overestimate,thm:random-rank-one-complex-var-overestimate} are proven in \Cref{sec:real-rank-one-estimation,sec:complex-matvec}, respectively.
In the real case, applying the bound \(\norm{\bar\mA}_F \leq \norm{\mA}_F\) causes the expression for the variance to become a massive overestimate, introducing a spurious factor of \(2^k\).
However, in the complex case, this bound gives a nearly tight result that varries from the true variance by a factor of at most 3.
This is rather unintuitive, as this seems to suggest that applying the same bound on the same matrix \(\bar\mA\) results in a variance estimate that is much more accurate when we use complex Kronecker matrix-vector products as opposed to real Kronecker matrix-vector products.
More work is needed to understand when bounding \(\norm{\bar\mA}_F \leq \norm{\mA}_F\) does not introduce spurious factors that are exponential in \(k\).

\subsubsection{Additional Results and Hope for a Faster Algorithm}

We complement the aforementioned results on the Kronecker-Hutchinsons estimator with a few informative additions.
First, we show a lower bound against all Kronecker-matrix-vector algorithms.
By adapting the techniques of \cite{braverman2020gradient} and \cite{jiang2021optimal}, we show that any Kronecker-matrix-vector algorithm that estimates the trace of \mA with standard deviation \(\eps\tr(\mA)\) must make at least \(\Omega(\frac{\sqrt k}{\eps})\) oracle queries.

Second, we note that the lower bounds in \Cref{thm:real-psd-lower-bound,thm:complex-psd-lower-bound} both use construction of \mA such that \(\mA = \mA_1 \otimes \cdots \otimes \mA_k\).
That is, the Kronecker-Hutchinson Estimator converges slowest for matrices that have a Kronecker structure.
However, if we truly are given \mA with the promise that \(\mA = \mA_1 \otimes \cdots \otimes \mA_k\) for some unknown matrices \(\mA_1,\ldots,\mA_k\), then in \Cref{sec:kron-recovery} we show that the trace of \mA can be recovered exactly using just \(kd+1\) Kronecker-matrix-vector products.
Similarly, \Cref{thm:real-rank-one-estimation,thm:complex-rank-one-estimation} suggest that the Kronecker-Hutchinson Estimator converges exponentially slowly for random rank-one matrices.
But in this rank-one case, we show that we can exactly recover the trace of \mA with a single Kronecker-matrix-vector query!

This trend, where the matrices that make the Kronecker-Hutchinson Estimator converge slowly can have their traces efficiently estimated by another algorithm, is reminiscent of the setup for the analysis in \cite{meyer21hutchpp}.
In that paper, the authors recall the analysis of the standard deviation of Hutchinson's Estimator:
\[
	\sqrt{\Var[\vx^\intercal\mA\vx]} \leq \sqrt{\frac2\ell}\norm{\mA}_F \leq \sqrt{\frac2\ell}\tr(\mA) \leq \eps\tr(\mA)
\]
where the first inequality uses that \(\norm\mA_F\leq\tr(\mA)\) for PSD \mA, and the second inequality is forced to pick the sample complexity \(\ell \geq \frac{2}{\eps^2}\).
Since the first inequality is only tight for matrices \mA that are nearly low-rank, the authors propose an algorithm which combines Hutchinson's Estimator with a low-rank approximation method that can efficiently compute the trace of nearly-low-rank matrices efficiently.
This way, their algorithm has a worst-case matrix-vector complexity of \(\Theta(\frac1\eps)\), while Hutchinson's Estimator is stuck with \(\Theta(\frac1{\eps^2})\).

Our analysis similarly characterizes two cases when the Kronecker-Hutchinson Estimator converges slowly -- when \mA is either low-rank or has Kronecker structure.
We also show that we can compute the trace such matrices very efficiently.
This suggests that a more involved algorithm can interpolate between these low query complexities and perhaps even avoid any exponential dependence on \(k\) at all.
We leave this as an important problem for future work.

Lastly, we remark that the analysis from \cite{meyer21hutchpp} crucially uses a tight bound for the variance of Hutchinson's Estimator in terms of \(\norm\mA_F\).
We believe that our core results, \Cref{thm:real-partial-trace-transpose,thm:real-partial-trace}, can serve an analogous role in constructing a faster trace-estimating algorithm.

%% file: prelims.tex
% Clean

\section{Preliminaries}
\label{sec:prelims}

We use capital bold letter to denote matrices, lowercase bold letters to denote vectors, and lowercase non-bold letters to denote scalars.
\bbR is the set of reals, and \bbC is the set of complex numbers.
\(\mathfrak{Re}[\vx]\) and \(\mathfrak{Im}[\vx]\) denote the real and imaginary parts of a vector or scalar.
\(\vx^\intercal\) denotes the transpose of a vector, while \(\vx^\herm\) denotes the conjugate transpose of a matrix or vector.
\(\mA \in \bbR^{d^k \times d^k}\) denotes out input matrix.
We use the bracket notation \([\mB]_{ij}\) to denote the \(i,j\) entry of matrix \mB.
We let \(\ve_i\) denote the \(i^{th}\) standard basis vector, whose exact dimension may depend on the context.
We let \(\mI \in \bbR^{d \times d}\) always refer to the \(d \times d\) identity matrix.
We let \(\norm{\mA}_F\) denote the Frobenius norm, \(\tr(\mA)\) the trace, \(\tr_{\cS}(\mA)\) the partial trace, and \(\mA^{\ptran\cV}\) the partial transpose.
We let \([k] = \{1,\ldots,k\}\) be the set of integers from 1 to \(k\).

We say a vector \(\vx_i\in\bbR^d\) is Gaussian if each entry is sampled iid from \(\cN(0,1)\), is Rademacher if each entry is sampled uniformly iid from \(\{\pm1\}\), and is sampled from the sphere if \(\vx_i\) is drawn uniformly at random from the set of all vectors in \(\bbR^d\) with \(\norm{\vx_i}_2=\sqrt d\).
We say a vector \(\vx_i \in \bbC^d\) is complex Gaussian if each entry is an iid copy of \(\frac{1}{\sqrt 2} r + \frac{i}{\sqrt 2} m\) where \(r,m\sim\cN(0,1)\).
We say \(\vx_i\in\bbC^d\) is complex Rademacher if each entry is sampled uniformly iid from \(\{\pm \frac1{\sqrt2},\pm \frac{i}{\sqrt2}\}\).
We say \(\vx_i\in\bbC^d\) is sampled from the complex sphere if \(\vx_i\) is sampled uniformly from the set of all vectors in \(\bbC^d\) with \(\norm{\vx_i}_2=\sqrt d\).

\subsection{Kronecker product}
The Kronecker product of two matrices \(\mA\in\bbR^{n \times m}\) and \(\mB\in\bbR^{p \times q}\) is \(\mA\otimes\mB \in \bbR^{np \times mq}\), where
\[
	\mA\otimes\mB \defeq \bmat{
		[\mA]_{11}\mB & [\mA]_{12}\mB & \ldots & [\mA]_{1m}\mB \\
		[\mA]_{21}\mB & [\mA]_{22}\mB & \ldots & [\mA]_{2m}\mB \\
		\vdots & \vdots & \ddots & \vdots \\
		[\mA]_{n1}\mB & [\mA]_{n2}\mB & \ldots & [\mA]_{nm}\mB
	}
\]
We also use the shorthands \(\otimes_{i=1}^k \vx_i \defeq \vx_1 \otimes \vx_2 \otimes \cdots \otimes \vx_k\), as well as \(\mA^{\otimes k} \defeq \otimes_{i=1}^k \mA\).
The Kronecker product has three key properties that we repeatedly use:
\begin{lemma}[Mixed-Product Property (Folklore)]
	\label{lem:mixed-product}
	\((\mA\otimes\mB)(\mC\otimes\mD) = \mA\mC \otimes \mB\mD\)
\end{lemma}
\begin{lemma}[Transposing (Folklore)]
	\label{lem:kron-transpose}
	\((\mA\otimes\mB)^\intercal = (\mA^\intercal \otimes \mB^\intercal)\)
\end{lemma}
\begin{lemma}[Fact 7.4.1 in \cite{bernstein11}]
\label{lem:kron-extract-vec}
Let \(\vx\in\bbR^{n}\) and \(\vy\in\bbR^{m}\).
Then, \(\vx \otimes \vy = (\mI_n \otimes \vy) \vx = (\vx \otimes \mI_m) \vy\).
\end{lemma}

At one point, we will also require a short lemma about the squared Frobenius norm and the Kronecker product.
This lemma is roughly equivalent to noting that \(\norm{\mA}_F^2 = \sum_{i,j=1}^d \norm{\mA_{ij}}_F^2\) in the block matrix decomposition of \mA when \(k=2\) in \Cref{sec:our-contributions}.
\begin{lemma}
\label{lem:kron-sum-of-frob}
Let \(\mC,\mE \in \bbR^{d^{k-i}\times d^{k-i}}\) and \(\mD \in \bbR^{d^i \times d^i}\) for any \(i\in[k]\).
Then, we have
\[
	\sum_{\hat i, \hat j=1}^{d^i} \norm{(\mC \otimes \ve_{\hat i})^\intercal \mD (\mE \otimes \ve_{\hat j})}_F^2 = \norm{(\mC \otimes \mI)^\intercal\mD(\mE \otimes \mI)}_F^2.
\]
\end{lemma}
\begin{proof}
First, recall there exists a permutation matrix \mP such that \(\mC \otimes \vy = \mP(\vy \otimes \mC)\mP^\intercal\) for all \(\vy\in\bbR^{d^i}\).
Now let \(\hat\mD \defeq \mP\mD\mP\).
Then, by expand \(\hat\mD\) as a block matrix of matrices \(\hat\mD_{\hat i, \hat j} \in \bbR^{d^{k-i} \times d^{k-i}}\), and noting that \((\ve_{\hat i}\otimes\mC)^\intercal = \bmat{\mat 0 & \cdots & \mat 0 & \mC^\intercal & \mat 0 & \cdots & \mat 0}\).
We then get that
\[
	(\ve_{\hat i} \otimes \mC)^\intercal\hat\mD(\ve_{\hat j} \otimes \mE)
	 = \bmat{
		 \mat 0 & \cdots & \mat 0 & \mC^\intercal & \mat 0 & \cdots & \mat 0
		 }
		 \bmat{
		 	\hat\mD_{11} & \cdots & \hat\mD_{1d} \\
		 	\vdots & \ddots & \vdots \\
		 	\hat\mD_{d1} & \cdots & \hat\mD_{dd} \\
		 }
		 \bmat{
		 \vec 0 \\ \vdots \\ \vec 0 \\ \mE \\ \vec 0 \\ \vdots \\ \vec 0
		 }
	 = \mC^\intercal\hat\mD_{\hat i \hat j}\mE.
\]
We can also notice that
\[
	(\mI \otimes \mC)^\intercal\hat\mD(\mI\otimes\mE)
	= \bmat{\mC^\intercal \\ & \ddots \\ & & \mC^\intercal} \bmat{\hat\mD_{11} & \ldots & \hat\mD_{dd} \\ \vdots & \ddots & \vdots \\ \hat\mD_{d1} & \ldots & \hat\mD_{dd}} \bmat{\mE \\ & \ddots \\ & & \mE}
	= \bmat{\mC^\intercal\hat\mD_{11}\mE & \ldots & \mC^\intercal\hat\mD_{1d}\mE \\ \vdots & \ddots & \vdots \\ \mC^\intercal\hat\mD_{d1}\mE & \ldots & \mC^\intercal\hat\mD_{dd}\mE}.
\]
Then, noting that the Frobenius norm is invariant to permutation, we find that
\begin{align*}
	\sum_{\hat i, \hat j=1}^{d^i} \norm{(\mC \otimes \ve_{\hat i})^\intercal\mD(\mE\otimes\ve_{\hat j})}_F^2
	&= \sum_{\hat i, \hat j=1}^{d^i} \norm{\mP^\intercal(\mC \otimes \ve_{\hat i})^\intercal\mP\mD\mP(\ve_{\hat j} \otimes \mE)\mP^\intercal}_F^2 \\
	&= \sum_{\hat i, \hat j=1}^{d^i} \norm{(\mC \otimes \ve_{\hat i})^\intercal\hat\mD(\ve_{\hat j} \otimes \mE)}_F^2 \\
	&= \sum_{\hat i, \hat j=1}^{d^i} \norm{\mC^\intercal\hat\mD_{\hat i \hat j}\mE}_F^2 \\
	&= \norm{(\mI\otimes\mC)^\intercal\hat\mD(\mI\otimes\mE)}_F^2 \\
	&= \norm{(\mC\otimes\mI)^\intercal\hat\mD(\mE\otimes\mI)}_F^2,
\end{align*}
where the last equality comes from again permuting the matrices inside the norm.
\end{proof}

%% file: partial-trace.tex
% Clean

\subsection{The Partial Trace Operator}
We take a moment to formalize the partial trace operator and characterize some useful properties of it.
We start by defining the partial trace with respect to a single subsystem:
\begin{definition}
\label{def:partial-trace-singular}
Let \(\mA\in\bbR^{d^k \times d^k}\).
Then, the partial trace of \mA with respect to its \(i^{th}\) subsystem is
\begin{align}
	\tr_{\{i\}}(\mA)
	\defeq \sum_{j=1}^d (\mI^{\otimes i-1} \otimes \ve_j \otimes \mI^{\otimes k-1})^\intercal \mA (\mI^{\otimes i-1} \otimes \ve_j \otimes \mI^{\otimes k-1})
	\in\bbR^{d^{k-1} \times d^{k-1}}.
	\label{eq:partial-trace-singular}
\end{align}
where these identity matrices are \(d \times d\) matrices, and where \(\ve_j\in\bbR^d\) is the \(j^{th}\) standard basis vector.
\end{definition}
Equivalently, the partial trace operator can be defined as the unique linear operator which guarantees that for all \(\mA\in\bbR^{d^{i-1} \times d^{i-1}}\), \(\mB\in\bbR^{d \times d}\), and \(\mC\in\bbR^{d^{k-i} \times d^{k-i}}\), we have that \(\tr_{\{i\}}(\mA \otimes \mB \otimes \mC) = \tr(\mB) \cdot (\mA \otimes \mC)\).

We will often consider the partial trace of \mA with respect to a set of subsystems \(\cS \subseteq[k]\).
We can think of this as either composing the operation in \Cref{eq:partial-trace-singular} \(\abs{\cS}\) many times, or as extending the summation in \Cref{eq:partial-trace-singular} to sum over all \(d^{\abs{\cS}}\) basis vectors that span the product of subspaces defined in \cS.
To see this equivalence, consider the partial trace of \mA with respect to \(\cS=\{1,2\}\):
\begin{align*}
	\tr_{\{1,2\}}(\mA)
	&= \tr_{\{1\}}(\tr_{\{2\}}(\mA)) \\
	&= \sum_{j=1}^d (\ve_j \otimes \mI^{\otimes k-1})^\intercal
		\bigg(
			\sum_{\hat j=1}^d (\mI \otimes \ve_{\hat j} \otimes \mI^{\otimes k-2})^\intercal
			\mA
			(\mI \otimes \ve_{\hat j} \otimes \mI^{\otimes k-2})
		\bigg)
		(\ve_j \otimes \mI^{\otimes k-1}) \\
	&= \sum_{j=1}^d
		\sum_{\hat j=1}^d
		(\ve_j \otimes \mI \otimes \mI^{\otimes k-2})^\intercal
		(\mI \otimes \ve_{\hat j} \otimes \mI^{\otimes k-2})^\intercal
		\mA
		(\mI \otimes \ve_{\hat j} \otimes \mI^{\otimes k-2})
		(\ve_j \otimes \mI \otimes \mI^{\otimes k-2}) \\
	&= \sum_{j=1}^d
		\sum_{\hat j=1}^d
		(\ve_j \otimes \ve_{\hat j} \otimes \mI^{\otimes k-2})^\intercal
		\mA
		(\ve_j \otimes \ve_{\hat j} \otimes \mI^{\otimes k-2}).
		\tag{\Cref{lem:mixed-product}}
\end{align*}
Next, note that \(\ve_j \otimes \ve_{\hat j} \in \bbR^{d^2}\) is a standard basis vector in \(\bbR^{d^2}\) uniquely identified by the pair \((j, \hat j)\).
That is, we can replace this double sum above with a single sum over standard basis vectors in \(\bbR^{d^2}\):
\begin{align*}
	\tr_{\{1,2\}}(\mA)
	&= \sum_{j=1}^{d^2} (\ve_j \otimes \mI^{\otimes k-2})^\intercal \mA (\ve_j \otimes \mI^{\otimes k-2}).
	\tag{\(\ve_j\in\bbR^{d^2}\) now}
\end{align*}
Which now looks remarkably similar to \Cref{def:partial-trace-singular}, but with a larger definition of ``subsystem'' that the standard basis vector is acting on.
Formally, we define the partial trace of \mA with respect to a set \cS as following:
\begin{definition}
\label{def:partial-trace}
Let \(\mA\in\bbR^{d^k \times d^k}\), and let \(\cS\subseteq[k]\).
Let \(i_1,\ldots,i_{\abs{\cS}}\) be the indices in \cS, sorted so that \(i_1 < i_2 < \ldots < i_{\abs{\cS}}\).
Then, the partial trace of \(\mA\) with respect to the subsystem \cS is
\[
	\tr_{\cS}(\mA) \defeq \tr_{\{i_1\}}(\tr_{\{i_2\}}(\cdots(\tr_{\{i_{\abs{\cS}}\}}(\mA))))
	\in\bbR^{d^{k-\abs{\cS}} \times d^{k-\abs{\cS}}}.
\]
In the special case that \(\cS = \{i,i+1,\ldots,k\}\), we denote this partial trace as \(\tr_{i:}(\mA)\).
\end{definition}
In principle, the order of subscripts does not really matter.
Tracing out the first subsystem and then tracing out the second subsystem produces the exact same matrix as tracing out the second subsystem and then tracing out the first subsystem.
However, notationally, a non-ascending order becomes messy since removing the second subsystem changes the indices.
To see this messiness, note that the above statement would be formalized as \(\tr_{\{1\}}(\tr_{\{2\}}(\mA)) = \tr_{\{1\}}(\tr_{\{1\}}(\mA))\).
To avoid this issue, we always expand the partial traces in \cS in sorted order as in \Cref{def:partial-trace}.

Moving on, a vital property of the partial trace is that it is trace-preserving:
\begin{lemma}[Folklore]
\label{lem:trace-of-partial-trace}
For all \(\mA\in\bbR^{d^k \times d^k}\) and \(\cS\subseteq[k]\), we have \(\tr(\tr_{\cS}(\mA)) = \tr(\mA)\).
\end{lemma}

\subsection{Post-Measurement Reduced Density Matrices}
Our analysis also heavily involves post-measurement reduced density matrices, which we now define, and whose long name originates from quantum physics.
As mentioned in \Cref{sec:our-contributions}, if \mA is a matrix that describes the joint probability distribution of many particles (a density matrix), then \(\tr_\cS(\mA)\) describes the joint probability distribution of the particles that Bob can view.
The Post-measurement Reduced Density Matrix (PMRDM) of \mA with respect to vector \vx describes the joint distribution of the particles that Bob can view \emph{given that} Alice measured her particles in the direction of the vector \vx.\footnote{For the purpose of understanding the paper, there is no particular need to understand this intuition. Those interested in learning more are recommended to read through Chapter 6 of \cite{aaronsonIntroduction}.}
For brevity, we refer to this matrix as the PMRDM.
\begin{definition}
Let \(\mA\in\bbR^{d^k \times d^k}\), \(i \in [k]\), and \(\vx_i \in \bbR^d\).
Then the PMRDM of \mA after measuring \(\vx\) along the \(i^{th}\) subsystem is
\[
	\pmrdm{\{i\}}{\vx_i}{\mA}
	\defeq
		(\mI^{\otimes i-1} \otimes \vx_i \otimes \mI^{\otimes k-i})^\intercal
		\mA
		(\mI^{\otimes i-1} \otimes \vx_i \otimes \mI^{\otimes k-i})
	\in\bbR^{d^{k-1} \times d^{k-1}}.
\]
Further let \(\vx_{1},\ldots,\vx_{i} \in \bbR^d\).
Then, the PMRDM of \mA after measuring according to \(\vx_{1},\ldots,\vx_{i}\) in the first \(i\) subsystems of \([k]\) is
\begin{align*}
	\pmrdm{:i}{\vx_{:i}}{\mA}
	&\defeq
		\pmrdm{\{1\}}{\vx_{i}}{
			\pmrdm{\{2\}}{\vx_{i}}{
				\cdots
				(\pmrdm{\{i\}}{\vx_{i}}{
					\mA
				})
			}
		} \\
	&= (\vx_{:i} \otimes \mI^{\otimes k-i})^\intercal
		\mA
		(\vx_{:i} \otimes \mI^{\otimes k-i})
	\in\bbR^{d^{k-i} \times d^{k-i}}.
\end{align*}
where \(\vx_{:i} = \vx_1 \otimes \cdots \otimes \vx_i \in \bbR^{d^i}\).
\end{definition}
Notice that the notation for the PMRDM \(\pmrdm{\{i\}}{\vx_i}{\mA}\) does not explicitly show the vector \(\vx_i\) that define the measurement.
This is done for visual clarity, and is safe because the meaning of the vector \(\vx_i\) will always be clear from context.\footnote{Namely, we always take \(\vx_i\) to be the \(i^{th}\) term in the expansion of \(\vx = \otimes_{i=1}^k \vx_i\), where \(\vx_i\) has either iid Gaussian or Rademacher entries.}
Further, while we could define the PMRDM with respect to an arbitrary set of subsystems \cS, as done in \Cref{def:partial-trace}, our proofs only need to examine the PMRDM when looking at the first \(i\) subsystems, so for simplicity we only define the PMRDM for that case.

We need one important property from the PMRDM, which relates it to the partial trace:
\begin{lemma}
\label{lem:pmrdm-to-partial-trace}
Let \(\mA\in\bbR^{d^k \times d^k}\), \(i\in[k]\), and \(\vx_1,\ldots,\vx_i\in\bbR^d\).
Define \(\vx_{:i} = \vx_1 \otimes \cdots \otimes \vx_i\).
Then, \(\tr(\pmrdm{:i}{\vx_{:i}}{\mA}) = \vx_{:i}^\intercal\tr_{i+1:}(\mA)\vx_{:i}\).
\end{lemma}
\begin{proof}
We directly expand and simplify:
\begin{align*}
	\tr(\pmrdm{:i}{\vx_{:i}}{\mA})
	&= \tr((\vx_{:i} \otimes \mI^{\otimes k-i})^\intercal \mA (\vx_{:i} \otimes \mI^{\otimes k-i})) \\
	&= \sum_{j=1}^{d^{k-i}} \ve_j^\intercal (\vx_{:i} \otimes \mI^{\otimes k-i})^\intercal \mA (\vx_{:i} \otimes \mI^{\otimes k-i}) \ve_j \\
	&= \sum_{j=1}^{d^{k-1}} \vx_{:i}^\intercal (\mI^{\otimes i} \otimes \ve_j)^\intercal \mA (\mI^{\otimes i} \otimes \ve_j) \vx_{:i} \tag{\Cref{lem:kron-extract-vec}}\\
	&= \vx_{:i}^\intercal \tr_{i+1:}(\mA)\, \vx_{:i}
\end{align*}
\end{proof}

\subsection{The Partial Transpose Operator}
Finally, we formalize the partial transpose operator, and characterize a useful property of it.
We start by defining the partial transpose of \mA with respect to a single subsystem.
\begin{definition}
\label{def:partial-transpose}
Let \(\mA \in \bbR^{d^k \times d^k}\) and \(i\in[k]\).
Then the partial transpose of \mA with respect to the \(i^{th}\) subsystem is
\[
	\mA^{\ptran{\{i\}}}
	\defeq \sum_{\hat i, \hat j = 1}^d
		(\mI^{\otimes i-1} \otimes \ve_{\hat i}\ve_{\hat j}^\intercal \otimes \mI^{\otimes k-i})
		\mA
		(\mI^{\otimes i-1} \otimes \ve_{\hat i}\ve_{\hat j}^\intercal \otimes \mI^{\otimes k-i}).
\]
\end{definition}
Equivalently, the partial transpose operator is the unique linear operator which guarantees that for all \(\mB\in\bbR^{d^{i-1} \times d^{i-1}}\), \(\mC\in\bbR^{d \times d}\), \(\mD \in \bbR^{d^{k-i} \times d^{k-i}}\), we have that \((\mB\otimes\mC\otimes\mD)^{\ptran{\{i\}}} = (\mB \otimes \mC^\intercal \otimes \mD)\).
To verify this equivalence, notice that we can express the transpose of a matrix as
\[
	\sum_{\hat i, \hat j=1}^d \ve_{\hat i}\ve_{\hat j}^\intercal \mC \ve_{\hat i}\ve_{\hat j}^\intercal
	= \sum_{\hat i, \hat j=1}^d [\mC]_{\hat j \hat i} \ve_{\hat i}\ve_{\hat j}^\intercal
	= \mC^\intercal.
\]
Then, we see that
\begin{align*}
	\sum_{\hat i, \hat j = 1}^d
		(\mI^{\otimes i-1} \otimes \ve_{\hat i}\ve_{\hat j}^\intercal \otimes \mI^{\otimes k-i})
		(\mB \otimes \mC \otimes \mD)
		(\mI^{\otimes i-1} \otimes \ve_{\hat i}\ve_{\hat j}^\intercal \otimes \mI^{\otimes k-i})
	&= \sum_{\hat i, \hat j = 1}^d
		(\mB \otimes [\mC^\intercal]\ve_{\hat i}\ve_{\hat j}^\intercal \otimes \mD) \\
	&= (\mB \otimes \mC^\intercal \otimes \mD).
\end{align*}
Like the partial trace, we will often be interested in taking the partial transpose of \mA with respect to a set of subsystems \(\cV \subseteq [k]\).
We extend the definition of the partial transpose directly to this case:
\begin{definition}
Let \(\mA \in \bbR^{d^k \times d^k}\) and \(\cV\subseteq[k]\).
Let \(i_{1}, \ldots, i_{\abs\cV}\) be the indices in \cV.
Then, the partial transpose of \mA with respect to the subsystem \cV is
\[
	\mA^{\ptran\cV}
	\defeq
	(((\mA^{\ptran{i_1}})^{\ptran{i_2}})\cdots)^{\ptran{i_{\abs\cV}}}.
\]
\end{definition}
Notably, the order of partial transposes does not matter in this case.
That is, we have \((\mA^{\ptran{\{i\}}})^{\ptran{\{j\}}} = (\mA^{\ptran{\{j\}}})^{\ptran{\{i\}}}\).
We next establish two useful properties of the partial transpose:
\begin{lemma}
\label{lem:partial-trace-of-partial-transpose}
Let \(\mA \in \bbR^{d^k \times d^k}\) and \(i\in[k]\).
Then, \(\tr_{\{i\}}(\mA^{\ptran{\{i\}}}) = \tr_{\{i\}}(\mA)\).
\end{lemma}
\begin{proof}
We directly compute
\begin{align*}
	\tr_{\{i\}}(\mA^{\ptran{\{i\}}})
	&= \sum_{\hat i = 1}^d (\mI^{\otimes i-1} \otimes \ve_{\hat i} \otimes \mI^{\otimes k-i})^\intercal
		\mA^{\ptran{\{i\}}}
		(\mI^{\otimes i-1} \otimes \ve_{\hat i} \otimes \mI^{\otimes k-i}) \\
	&= \sum_{\hat i, \hat j, \hat \ell = 1}^d (\mI^{\otimes i-1} \otimes \ve_{\hat i}^\intercal\ve_{\hat j}\ve_{\hat \ell}^\intercal \otimes \mI^{\otimes k-i})
		\mA
		(\mI^{\otimes i-1} \otimes \ve_{\hat i}\ve_{\hat j}\ve_{\hat \ell}^\intercal \otimes \mI^{\otimes k-i}) \\
	&= \sum_{\hat i, \hat j, \hat \ell = 1}^d \mathbbm1_{[\hat i = \hat j]} \mathbbm1_{[\hat i = \hat \ell]} (\mI^{\otimes i-1} \otimes \ve_{\hat \ell}^\intercal \otimes \mI^{\otimes k-i})
		\mA
		(\mI^{\otimes i-1} \otimes \ve_{\hat j} \otimes \mI^{\otimes k-i}) \\
	&= \sum_{\hat i = 1}^d (\mI^{\otimes i-1} \otimes \ve_{\hat i} \otimes \mI^{\otimes k-i})^\intercal
		\mA
		(\mI^{\otimes i-1} \otimes \ve_{\hat i} \otimes \mI^{\otimes k-i}) \\
	&= \tr_{\{i\}}(\mA).
\end{align*}
\end{proof}
\begin{lemma}
\label{lem:reorder-parital-transpose}
Let \(\mA \in \bbR^{d^k \times d^k}\), \(i\in[k]\), \(\mB\in\bbR^{d^{i-1} \times d^{i-1}}\), and \(\mC \in \bbR^{d^{k-i} \times d^{k-i}}\).
Then, for all \(\hat i, \hat j \in [d]\) we have
\[
	(\mB \otimes \ve_{\hat j} \otimes \mC)^\intercal \mA (\mB \otimes \ve_{\hat i} \otimes \mC)
	= (\mB \otimes \ve_{\hat i} \otimes \mC)^\intercal \mA^{\ptran{\{i\}}} (\mB \otimes \ve_{\hat j} \otimes \mC)
\]
\end{lemma}
\begin{proof}
We expand the right-hand-side:
\begin{align*}
	&
		(\mB \otimes \ve_{\hat i} \otimes \mC)^\intercal
		\mA^{\ptran{\{i\}}}
		(\mB \otimes \ve_{\hat j} \otimes \mC) \\
	=&
		(\mB \otimes \ve_{\hat i} \otimes \mC)^\intercal
		({\textstyle \sum_{\hat \ell, \hat m=1}^d}
			(\mI^{\otimes i-1} \otimes \ve_{\hat \ell}\ve_{\hat m}^\intercal \otimes \mI^{\otimes k-i})
			\mA
			(\mI^{\otimes i-1} \otimes \ve_{\hat \ell}\ve_{\hat m}^\intercal \otimes \mI^{\otimes k-i}))
		(\mB \otimes \ve_{\hat j} \otimes \mC) \\
	=&
		{\textstyle \sum_{\hat \ell, \hat m=1}^d}
			(\mB^\intercal \otimes \ve_{\hat i}^\intercal\ve_{\hat \ell}\ve_{\hat m}^\intercal \otimes \mC^\intercal)
			\mA
			(\mB \otimes \ve_{\hat \ell}\ve_{\hat m}^\intercal\ve_{\hat j} \otimes \mC) \\
	=&
		(\mB^\intercal \otimes \ve_{\hat j}^\intercal \otimes \mC^\intercal)
		\mA
		(\mB \otimes \ve_{\hat i} \otimes \mC),
\end{align*}
where we remove sum by noting that \(\ve_{\hat i}^\intercal\ve_{\hat \ell}\) and \(\ve_{\hat m}^\intercal \ve_{\hat j}\) are zero unless \(\hat \ell = \hat i\) and \(\hat m = \hat j\).
\end{proof}

We now have enough tools to exactly understand the variance of the Kronecker-Hutchinson Estimator.

%% file: trace-estimation.tex
% Clean

\section{The Exact Variance of Kronecker-Hutchinson Estimator}
\label{sec:exact-var-proof}
We begin by recalling the exact variance of the classical (\(k=1\)) Hutchinson's Estimator:
\begin{lemma}[\cite{girard1987algorithme,hutchinson1989stochastic,AT11, ram900}]
\label{lem:hutch-classical-variance}
	Let \(\mA \in \bbR^{D \times D}\) and let \(\vg \in \bbR^D\) be a vector of either iid Rademacher or standard normal Gaussian entries.
	Let \(\hat\mA = \frac12(\mA + \mA^\intercal)\).
	Then, \(\E[\vg^\intercal\mA\vg] = \tr(\mA)\) and \(\Var[\vg^\intercal\mA\vg] \leq 2\norm{\hat\mA}_F^2 \leq 2\norm{\mA}_F^2\).
	The first inequality is tight if \vg is Gaussian, and the second is tight if \mA is symmetric.
\end{lemma}
This factor of 2 is important for our analysis, and is the base of the exponent \(2^{k - \abs{\cS}}\) in \Cref{thm:real-partial-trace-transpose}.
The key to our proof technique is to use \Cref{lem:kron-extract-vec} to expand
\begin{align*}
	\vx^\intercal\mA\vx
	&= (\vx_1 \otimes \cdots \otimes \vx_k)^\intercal \mA (\vx_1 \otimes \cdots \otimes \vx_k) \\
	&= \vx_k^\intercal (\vx_1 \otimes \cdots \otimes \vx_{k-1} \otimes \mI)^\intercal \mA (\vx_1 \otimes \cdots \otimes \vx_{k-1} \otimes \mI) \vx_k \\
	&= \vx_k^\intercal \pmrdm{:k-1}{\vx_{:k-1}}{\mA} \vx_k \numberthis \label{eq:introduce-pmrdm}
\end{align*}
which then is just a classical (\(k=1\)) Hutchinson's Estimate of the trace of \(\pmrdm{:k-1}{\vx_{:k-1}}{\mA}\).
\Cref{lem:hutch-classical-variance} then gives (for Gaussians) an exact understanding of the variance of this estimator.
As a warm-up, we first use this proof technique to show that the Kronecker-Hutchinson Estimator is unbiased, though shorter proofs are certainly possible.
\begin{theorem}
	\label{thm:kron-hutch-unbiased}
	Let \(\mA \in \bbR^{d^k \times d^k}\).
	Let \(\vx_1,\ldots,\vx_k \in \bbR^d\) be iid Rademacher or standard normal Gaussian vectors, and let \(\vx = \vx_1 \otimes \cdots \otimes \vx_k\).
	Then, \(\E[\vx^\intercal\mA\vx] = \tr(\mA)\).
\end{theorem}
\begin{proof}
For any \(i\in[k]\),
\begin{align*}
	\E[\vx_{:i}^\intercal \tr_{i+1:}(\mA)\, \vx_{:i}]
	&= \E[\vx_i^\intercal \, \pmrdm{:i-1}{\vx_{:i}}{\tr_{i+1:}(\mA)}\, \vx_i] \tag{\Cref{eq:introduce-pmrdm}} \\
	&= \E_{\vx_{:i-1}}\E_{\vx_i}[\vx_i^\intercal \, \pmrdm{:i-1}{\vx_{:i}}{\tr_{i+1:}(\mA)}\, \vx_i\mid \vx_{:i-1}] \tag{Tower Rule} \\
	&= \E_{\vx_{:i-1}}[ \tr(\pmrdm{:i-1}{\vx_{:i}}{\tr_{i+1:}(\mA)})\, ] \tag{\Cref{lem:hutch-classical-variance}} \\
	&= \E[ \vx_{:i-1}^\intercal \tr_{i:}(\mA)\, \vx_{:i-1} ]. \tag{\Cref{lem:pmrdm-to-partial-trace}}
\end{align*}
Repeating this argument \(k\) times in a row, and letting \(\tr_{\emptyset}(\mA) \defeq \mA\), we conclude by \Cref{lem:trace-of-partial-trace} that
\[
	\E[\vx^\intercal \mA \vx] = \E[\vx_{:k}^\intercal \tr_{\emptyset}(\mA)\, \vx_{:k}] = \ldots = \E[\vx_{1}^\intercal \tr_{2:}(\mA)\, \vx_{1}] = \tr(\tr_{2:}(\mA)) = \tr(\mA)
\]
which completes the proof.

\end{proof}
Analyzing the variance of the estimator requires a bit more effort due to the Frobenius norm term and \(\bar\mA\) matrix that appear in \Cref{lem:hutch-classical-variance}.
More concretely, we approach the analysis of the variance of \(\vx^\intercal\mA\vx\) via the classical formula
\[
	\Var[\vx^\intercal\mA\vx] 
	= \E[(\vx^\intercal\mA\vx)^2] - (\E[\vx^\intercal\mA\vx])^2.
\]
Since we already know that \((\E[\vx^\intercal\mA\vx])^2 = (\tr(\mA))^2\), we turn our attention to evaluating \(\E[(\vx^\intercal\mA\vx)^2]\).
Note that rearranging \Cref{lem:hutch-classical-variance} gives that
\begin{align}
	\E[(\vg^\intercal\mA\vg)^2] \leq (\tr(\mA))^2 + 2\norm{\bar\mA}_F^2,
	\label{eq:squared-expected-classic-hutch}
\end{align}
with equality if \(\vg\) is Gaussian.
Then, noting that \(\frac12(\pmrdm{:k-1}{\vx_{:k-1}}{\mA} + \pmrdm{:k-1}{\vx_{:k-1}}{\mA}^\intercal) = \pmrdm{:k-1}{\vx_{:k-1}}{\frac12(\mA+\mA^\intercal)}\), and taking the same approach as in the bias analysis, we see that
\begin{align*}
	\E[(\vx^\intercal\mA\vx)^2]
	&= \E[(\vx_k^\intercal \pmrdm{:k-1}{\vx_{:k-1}}{\mA} \vx_k)^2] \tag{\Cref{eq:introduce-pmrdm}}\\
	&= \E_{\vx_{:k-1}}\E_{\vx_k}[(\vx_k^\intercal \pmrdm{:k-1}{\vx_{:k-1}}{\mA} \vx_k)^2 \mid \vx_{:k-1}] \tag{Tower Rule} \\
	&= \E_{\vx_{:k-1}} [(\tr(\pmrdm{:k-1}{\vx_{:k-1}}{\mA}))^2 + 2 \norm{\pmrdm{:k-1}{\vx_{:k-1}}{\mB}}_F^2] \tag{\Cref{eq:squared-expected-classic-hutch}} \\
	&= \E_{\vx_{:k-1}} [(\vx_{:k-1}^\intercal \tr_{\{k\}}(\mA)\, \vx_{:k-1})^2] + 2 \E_{\vx_{:k-1}}[\norm{\pmrdm{:k-1}{\vx_{:k-1}}{\mB}}_F^2], \tag{\Cref{lem:pmrdm-to-partial-trace}}
\end{align*}
where we let \(\mB \defeq \frac12(\mA+\mA^\intercal)\).
The first term admits a recursive relationship just like in the proof of \Cref{thm:kron-hutch-unbiased}, so we can easily handle that.
The second term however requires more effort.
The key observation we use is that the squared Frobenius norm of the PMRDM is in fact a sum of the squares of many Kronecker-Hutchinson Estimators:
\begin{align*}
	\norm{\pmrdm{:i}{\vx_{:i}}{\mB}}_F^2
	&= \sum_{\hat i, \hat j = 1}^{d^{k-i}} [\pmrdm{:i}{\vx_{:i}}{\mB}]_{\hat i \hat j}^2 \\
	&= \sum_{\hat i, \hat j = 1}^{d^{k-i}} (\ve_{\hat i}^\intercal \pmrdm{:i}{\vx_{:i}}{\mB} \ve_{\hat j})^2 \\
	&= \sum_{\hat i, \hat j = 1}^{d^{k-i}} (\ve_{\hat i}^\intercal (\vx_{:i} \otimes \mI^{\otimes k-i})^\intercal \mB (\vx_{:i} \otimes \mI^{k-i}) \ve_{\hat j})^2 \\
	&= \sum_{\hat i, \hat j = 1}^{d^{k-i}} (\vx_{:i}^\intercal (\mI^{\otimes i} \otimes \ve_{\hat i})^\intercal \mB (\mI^{\otimes i} \otimes \ve_{\hat j}) \vx_{:i})^2 \tag{\Cref{lem:kron-extract-vec}}
\end{align*}
In particular, this last line is the sum of squares of many Kronecker-Hutchinson Estimators for a wide variety of matrices.
Despite these being trace estimates for different matrices, we can just carefully analyze the resulting terms using the Mixed-Product Property (\Cref{lem:mixed-product}), \Cref{lem:kron-extract-vec}, and \Cref{eq:squared-expected-classic-hutch}.
We get the following recurrence relation:

\begin{lemma}
\label{lem:frob-inductive-step}
Fix \(i\in[k]\), and let \(\mB \in \bbR^{d^j \times d^j}\) for some \(j \geq i\).
Then,
\[
	\E_{\vx_{:i}}[\norm{\pmrdm{:i}{\vx_{:i}}{\mB}}_F^2]
	\leq
	\E_{\vx_{:i-1}} \norm{\pmrdm{:i-1}{\vx_{:i-1}}{\tr_i(\mB)}}_F^2 + 2 \E_{\vx_{:i-1}} \norm{\pmrdm{:i-1}{\vx_{:i-1}}{\tsfrac12(\mB + \mB^{\ptran{\{i\}}})}}_F^2,
\]
with equality if \(\vx_i\) is a standard normal Gaussian vector.
\end{lemma}
\begin{proof}
We first extract the vector \(\vx_i\) via the same technique as in \Cref{eq:introduce-pmrdm}.
Recall that we always take \(\mI = \mI_d \in \bbR^{d \times d}\) to be the identity in \(d\) dimensions.
\begin{align*}
	&\E_{\vx_{:i}}[\norm{\pmrdm{:i}{\vx_{:i}}{\mB}}_F^2] \\
	&= \sum_{\hat i, \hat j} \E_{\vx_{:i}}[(\ve_{\hat i}^\intercal \pmrdm{:i}{\vx_{:i}}{\mB} \ve_{\hat j})^2] \\
	&= \sum_{\hat i, \hat j} \E_{\vx_{:i}}[(\ve_{\hat i}^\intercal (\vx_{:i} \otimes \mI^{\otimes j-i})^\intercal\mB(\vx_{:i} \otimes \mI^{\otimes j-i}) \ve_{\hat j})^2] \\
	&= \sum_{\hat i, \hat j} \E_{\vx_{:i}}[((\vx_{:i-1} \otimes \vx_i \otimes \ve_{\hat i})^\intercal\mB(\vx_{i-1} \otimes \vx_i \otimes \ve_{\hat j}))^2] \\
	&= \E_{\vx_{:i-1}} \sum_{\hat i, \hat j} \E_{\vx_{i}}[
		(\vx_{i}^\intercal
			(\vx_{:i-1} \otimes \mI \otimes \ve_{\hat i})^\intercal \mB (\vx_{:i-1} \otimes \mI \otimes \ve_{\hat j})
		\vx_{i})^2 \mid \vx_{:i-1}]
\end{align*}
We take a moment to examine the variance of this Hutchinson's sample.
Letting \(\mC \defeq \vx_{:i-1} \otimes \mI\), we get that the variance of the sample is at most
\[
	2\norm{\tsfrac12((\mC \otimes \ve_{\hat i})^\intercal \mB (\mC \otimes \ve_{\hat j}) + 
	(\mC \otimes \ve_{\hat j})^\intercal \mB (\mC \otimes \ve_{\hat i}))}_F^2.
\]
By \Cref{lem:reorder-parital-transpose}, we know the term on the right is equivalent to \((\mC \otimes \ve_{\hat i})^\intercal \mB^{\ptran{\{i\}}} (\mC \otimes \ve_{\hat j})\), so that the norm above becomes \(2\norm{(\mC \otimes \ve_{\hat i})^\intercal \frac12(\mB+\mB^{\ptran{\{i\}}}) (\mC \otimes \ve_{\hat j})}_F^2\).
\Cref{lem:kron-sum-of-frob} further tells us the the sum of this variance across all \(\hat i, \hat j \in [d^{j-i}]\) is \(2\norm{(\mC \otimes \mI)^\intercal \frac12(\mB+\mB^{\ptran{\{i\}}}) (\mC \otimes \mI)}_F^2\).
Returning to our overall analysis, we get
\begin{align*}
	&\E_{\vx_{:i}}[\norm{\pmrdm{:i}{\vx_{:i}}{\mB}}_F^2] \\
	&\leq \E_{\vx_{:i-1}} \sum_{\hat i, \hat j} \left[
			(\tr((\vx_{:i-1} \otimes \mI \otimes \ve_{\hat i})^\intercal \mB (\vx_{:i-1} \otimes \mI \otimes \ve_{\hat j})))^2
			+ 2 \norm{(\vx_{:i-1} \otimes \mI \otimes \ve_{\hat i})^\intercal \tsfrac12(\mB+\mB^{\ptran{\{i\}}}) (\vx_{:i-1} \otimes \mI \otimes \ve_{\hat j})}_F^2
		\right]\\
	&= \E_{\vx_{:i-1}} \sum_{\hat i, \hat j = 1} \left[
			(\tr((\vx_{:i-1} \otimes \mI \otimes \ve_{\hat i})^\intercal \mB (\vx_{:i-1} \otimes \mI \otimes \ve_{\hat j})))^2\right]
		+ 2 \E_{\vx_{:i-1}} \norm{(\vx_{:i-1} \otimes \mI^{\otimes 2})^\intercal \tsfrac12(\mB+\mB^{\ptran{\{i\}}}) (\vx_{:i-1} \otimes \mI^{\otimes 2})}_F^2
		\\
	&= \E_{\vx_{:i-1}} \sum_{\hat i, \hat j} \left[
			\left(\sum_{\ell=1}^d \ve_{\ell}^\intercal(\vx_{:i-1} \otimes \mI \otimes \ve_{\hat i})^\intercal \mB (\vx_{:i-1} \otimes \mI \otimes \ve_{\hat j})\ve_{\ell}\right)^2\right]
		+ 2 \E_{\vx_{:i-1}} \norm{\pmrdm{:i-1}{\vx_{:i-1}}{\tsfrac12(\mB+\mB^{\ptran{\{i\}}})}}_F^2
		\\
	&= \E_{\vx_{:i-1}} \sum_{\hat i, \hat j} \left[
			\left(\ve_{\hat i}^\intercal \sum_{\ell=1}^d (\vx_{:i-1} \otimes \ve_{\ell} \otimes \mI)^\intercal \mB (\vx_{:i-1} \otimes \ve_{\ell} \otimes \mI)\ve_{\hat j}\right)^2\right]
		+ 2 \E_{\vx_{:i-1}} \norm{\pmrdm{:i-1}{\vx_{:i-1}}{\tsfrac12(\mB+\mB^{\ptran{\{i\}}})}}_F^2
		\\
	&= \E_{\vx_{:i-1}} \sum_{\hat i, \hat j} \left[
			\left(\ve_{\hat i}^\intercal \pmrdm{:i-1}{\vx_{:i-1}}{\tr_{\{i\}}(\mB)}\ve_{\hat j}\right)^2\right]
		+ 2 \E_{\vx_{:i-1}} \norm{\pmrdm{:i-1}{\vx_{:i-1}}{\tsfrac12(\mB+\mB^{\ptran{\{i\}}})}}_F^2
		\\
	&= \E_{\vx_{:i-1}} \norm{\pmrdm{:i-1}{\vx_{:i-1}}{\tr_i(\mB)}}_F^2
		+ 2 \E_{\vx_{:i-1}} \norm{\pmrdm{:i-1}{\vx_{:i-1}}{\tsfrac12(\mB+\mB^{\ptran{\{i\}}})}}_F^2
  \end{align*}
Note that the only inequality applies \Cref{eq:squared-expected-classic-hutch}, which is an equality if \(\vx_i\) is a standard normal Gaussian.

\end{proof}
Notice that this lemma takes in an arbitrary matrix \mB, instead of the specific choice of \(\mB = \frac12(\mA+\mA^\intercal)\) from earlier.
This is useful, as we apply \Cref{lem:frob-inductive-step} repeatedly in the following proof:
\begin{lemma}
\label{lem:expected-pmrdm-frob}
Fix \(i\in[k]\), and let \(\mB \in \bbR^{d^j \times d^j}\) for some \(j \geq i\).
Let \(\bar\mB = \frac1{2^i}\sum_{\cV\subseteq[i]}\mB^{\ptran\cV}\).
Then,
\[
	\E_{\vx_{:i}}[\norm{\pmrdm{:i}{\vx_{:i}}{\mB}}_F^2]
	\leq \sum_{\cS \subseteq [i]} 2^{i-\abs{\cS}} \norm{\tr_{\cS}(\bar\mB)}_F^2,
\]
with equality if \(\vx_1,\ldots,\vx_i\) are iid standard normal Gaussian vectors.
\end{lemma}
\begin{proof}
We prove this by induction over \(i=1,\ldots,j\).
For notational clarity in this induction, we let \(f_i(\mB) \defeq \frac1{2^i}\sum_{\cV\subseteq[i]}\mB^{\ptran\cV}\).
Then, our goal is to prove that \(\E_{\vx_{:i}}[\norm{\pmrdm{:i}{\vx_{:i}}{\mB}}_F^2] \leq \sum_{\cS \subseteq [i]} 2^{i-\abs{\cS}} \norm{\tr_{\cS}(f_i(\mB))}_F^2\) for all \(i\) and \mB.
We start with the base case \(i=1\).
First note that \(\tr_{\{1\}}(\mB) = \tr_{\{1\}}(\frac12(\mB+\mB^{\ptran{\{1\}}})) = \tr_{\{1\}}(f_1(\mB))\) by \Cref{lem:partial-trace-of-partial-transpose}.
Then \Cref{lem:frob-inductive-step} gives us that
\begin{align*}
	\E_{\vx_{1}}[\norm{\pmrdm{\{1\}}{\vx_{1}}{\mB}}_F^2]
	&\leq \norm{\tr_{\{1\}}(\mB)}_F^2 + 2 \norm{\tsfrac12(\mB+\mB^{\ptran{\{1\}}})}_F^2 \\
	&= 2^{1-1} \norm{\tr_{\{1\}}(f_1(\mB))}_F^2 + 2^{1-0} \norm{f_1(\mB)}_F^2 \\
	&= \sum_{\cS \subseteq [1]} 2^{1 - \abs\cS} \norm{\tr_\cS(f_1(\mB))}_F^2.
\end{align*}
Next, for the induction case, we assume the lemma holds for \(i-1\).
Since \(f_i(\mB)\) is the average of the all partial transposes in the first \(i\) subsystems of \mB, we can expand \(f_i(\mB) = f_{i-1}(\tsfrac12(\mB+\mB^{\ptran{\{i\}}}))\).
So, for instance, by our induction hypothesis we have that
\[
	\E_{\vx_{:i-1}} \norm{\pmrdm{:i-1}{\vx_{:i-1}}{\tsfrac12(\mB+\mB^{\ptran{\{i\}}})}}_F^2
	\leq \sum_{\cS\subseteq[i-1]} 2^{i-1-\abs{\cS}} \norm{\tr_{\cS}(f_{i-1}(\tsfrac12(\mB+\mB^{\ptran{\{i\}}})))}_F^2
	= \sum_{\cS\subseteq[i-1]} 2^{i-1-\abs{\cS}} \norm{\tr_{\cS}(f_i(\mB))}_F^2.
\]
Therefore, by \Cref{lem:frob-inductive-step}, by our inductive hypothesis, and by the composition of partial traces, we have that
\begin{alignat*}{4}
	\E_{\vx_{:i}}[\norm{\pmrdm{:i}{\vx_{:i}}{\mB}}_F^2]
	&\leq
		\E_{\vx_{:i-1}} \norm{\pmrdm{:i-1}{\vx_{:i-1}}{\tr_{\{i\}}(\mB)}}_F^2
		&&+ 2 \E_{\vx_{:i-1}} \norm{\pmrdm{:i-1}{\vx_{:i-1}}{\tsfrac12(\mB+\mB^{\ptran{\{i\}}})}}_F^2 \\
	&=
		\E_{\vx_{:i-1}} \norm{\pmrdm{:i-1}{\vx_{:i-1}}{\tr_{\{i\}}(\tsfrac12(\mB+\mB^{\ptran{\{i\}}}))}}_F^2
		&&+ 2 \E_{\vx_{:i-1}} \norm{\pmrdm{:i-1}{\vx_{:i-1}}{\tsfrac12(\mB+\mB^{\ptran{\{i\}}})}}_F^2 \\
	&\leq
		\sum_{\cS\subseteq[i-1]} 2^{i-1-\abs{\cS}} \norm{\tr_\cS(\tr_{\{i\}}(f_i(\mB)))}_F^2
		&&+ 2 \sum_{\cS\subseteq[i-1]} 2^{i-1-\abs{\cS}} \norm{\tr_{\cS}(f_i(\mB))}_F^2 \\
	&=
		\sum_{\cS\subseteq[i-1]} 2^{i-1-\abs{\cS}} \norm{\tr_{\cS\cup\{i\}}(f_i(\mB))}_F^2
		&&+ 2 \sum_{\cS\subseteq[i-1]} 2^{i-1-\abs{\cS}} \norm{\tr_{\cS}(f_i(\mB))}_F^2 \\
	&=
            \sum_{\cS\subseteq[i-1]} 2^{i-(\abs{\cS}+1)} \norm{\tr_{\cS\cup\{i\}}(f_i(\mB))}_F^2
		&&+ \sum_{\cS\subseteq[i-1]} 2^{i-\abs{\cS}} \norm{\tr_{\cS}(f_i(\mB))}_F^2 \numberthis \label{eq:frob-combining-terms} \\
	&= \sum_{\cS' \subseteq [i]} 2^{i-\abs{\cS'}} \norm{\tr_{\cS'}(f_i(\mB))}_F^2 \numberthis \label{eq:frob-combined-terms}
\end{alignat*}
The second inequality applies the induction hypotheses to the matrix \(\tr_{\{i\}}(\frac12(\mB+\mB^{\ptran{\{i\}}}))\) on the left term, and applies the induction hypothesis to the matrix \(\frac12(\mB+\mB^{\ptran{\{i\}}})\) on the right term.

Further, \Cref{eq:frob-combined-terms} holds by comparing the first and second terms in \Cref{eq:frob-combining-terms}, where the left sum covers all \(\cS'\subseteq[i]\) such that \(i \in \cS'\), and where the right term covers all \(\cS'\subseteq[i]\) such that \(i \notin\cS'\).
Note that the only inequalities used are \Cref{eq:squared-expected-classic-hutch,lem:frob-inductive-step}, as well as the induction hypothesis, which are all equalities when considering iid standard normal Gaussian vectors.
\end{proof}
We now have all the tools needed to prove the bound on the variance of the Kronecker-Hutchinson Estimator.
\begin{reptheorem}{thm:real-partial-trace-transpose}
	Fix \(d,k\in\bbN\), and let \(\mA\in\bbR^{d^k \times d^k}\).
	Let \(\vx = \vx_1 \otimes \cdots \otimes \vx_k\), where each \(\vx_i\) is independent and is either a Rademacher vector or a \(\cN(\vec0,\mI)\) Gaussian vector.
	Let \(\bar\mA \defeq \frac1{2^k} \sum_{\cV \subseteq [k]} \mA^{\ptran\cV}\) be the average of all partial transposes of \mA.
	Then,
	\(\Var[\vx^\intercal\mA\vx] \leq \sum_{\cS\subset[k]} 2^{k - \abs{\cS}} \norm{\tr_{\cS}(\bar\mA)}_F^2\).
	Furthermore, if all \(\vx_i\) are \(\cN(\vec0,\mI)\), then this expression is an equality.
\end{reptheorem}
\begin{proof}
We start again with the observation in \Cref{eq:introduce-pmrdm}, which we use to prove the following claim by induction over \(i = 1, \ldots, k\):
\begin{align}
	\E_{\vx_{:i}}[(\vx_{:i}^\intercal\tr_{i+1:}(\mA)\vx_{:i})^2]
	\leq (\tr(\mA))^2 + 2\norm{\tr_{2:}(\tsfrac12(\mA+\mA^\intercal))}_F^2 + 2\sum_{j=1}^{i-1} \E_{\vx_{:j}} \norm{\pmrdm{:j}{\vx_{:j}}{\tr_{j+2:}(\tsfrac12(\mA+\mA^\intercal))}}_F^2,
	\label{eq:main-thm-induction}
\end{align}
where we interpret \(\tr_{k+1}(\mA) \defeq \mA\).
The base case for \Cref{eq:main-thm-induction} is when \(i=1\), where we immediately have
\begin{align*}
	\E_{\vx_1}[(\vx_1^\intercal\tr_{2:}(\mA)\vx_1)^2]
	\leq (\tr(\tr_{2:}(\mA)))^2 + 2 \norm{\tr_{2:}(\tsfrac12(\mA+\mA^\intercal))}_F^2
	= (\tr(\mA))^2 + 2 \norm{\tr_{2:}(\tsfrac12(\mA+\mA^\intercal))}_F^2
\end{align*}
which completes the base case.
Now we assume \Cref{eq:main-thm-induction} holds for \(i-1\), and we show that it holds for \(i\):
\begin{align*}
	&\E_{\vx_{:i}}[(\vx_{:i}^\intercal\tr_{i+1:}(\mA)\vx_{:i})^2] \\
	&= \E_{\vx_{:i-1}} \E_{\vx_i} [(\vx_i^\intercal\pmrdm{:i-1}{\vx_{i-1}}{\tr_{i+1:}(\mA)}\vx_i)^2 \mid \vx_{:i-1}] \\
	&\leq \E_{\vx_{:i-1}}\left[(\tr(\pmrdm{:i-1}{\vx_{i-1}}{\tr_{i+1:}(\mA)}))^2 + 2\norm{\pmrdm{:i-1}{\vx_{i-1}}{\tr_{i+1:}(\tsfrac12(\mA+\mA^\intercal))}}_F^2\right] \tag{\Cref{lem:hutch-classical-variance}} \\
	&= \E_{\vx_{:i-1}}[(\vx_{:i-1}^\intercal\tr_{i:}(\mA)\vx_{:i-1})^2] + \E_{\vx_{:i-1}}2\norm{\pmrdm{:i-1}{\vx_{:i-1}}{\tr_{i+1:}(\tsfrac12(\mA+\mA^\intercal))}}_F^2 \\
	&\leq (\tr(\mA))^2 + 2\norm{\tr_{2:}(\mA)}_F^2 + 2\sum_{j=1}^{i-1} \E_{\vx_{:j}} \norm{\pmrdm{:j}{\vx_{:j}}{\tr_{j+2:}(\tsfrac12(\mA+\mA^\intercal))}}_F^2
\end{align*}
where the last line applies the induction hypothesis.
This completes the proof by induction of \Cref{eq:main-thm-induction}.
We can then appeal to \Cref{lem:expected-pmrdm-frob} and merge terms together.
First, let \(\mC_{i} \defeq \frac1{2^i} \sum_{\cV\subseteq[i]} (\frac12(\mA+\mA^\intercal))^{\ptran\cV}\), and note that \(\mC_0 = \tsfrac12(\mA+\mA^\intercal)\).
Then,
\begin{align*}
	\E[(\vx^\intercal\mA\vx)^2]
	&\leq (\tr(\mA))^2 + 2\norm{\tr_{2:}(\tsfrac12(\mA+\mA^\intercal))}_F^2 + 2\sum_{i=1}^{k-1} \E_{\vx_{:i}} \norm{\pmrdm{:i}{\vx_{:i}}{\tr_{i+2:}(\tsfrac12(\mA+\mA^\intercal))}}_F^2 \\
	&\leq (\tr(\mA))^2 + 2\norm{\tr_{2:}(\mC_0)}_F^2 + 2\sum_{i=1}^{k-1} \sum_{\cS\subseteq[i]} 2^{i - \abs{\cS}} \norm{\tr_{\cS \cup \{i+2,\ldots,k\}} (\mC_{i})}_F^2 \\
	&= (\tr(\mA))^2 + \sum_{i=0}^{k-1} \sum_{\cS\subseteq[i]} 2^{i + 1 - \abs{\cS}} \norm{\tr_{\cS \cup \{i+2,\ldots,k\}} (\mC_{i})}_F^2 \\
	&= (\tr(\mA))^2 + \sum_{i=1}^{k} \sum_{\cS\subseteq[i-1]} 2^{i - \abs{\cS}} \norm{\tr_{\cS \cup \{i+1,\ldots,k\}} (\mC_{i-1})}_F^2
\end{align*}
We proceed by simplifying the partial trace \(\tr_{\cS \cup \{i+1,\ldots,k\}} (\mC_{i-1})\).
First, note that \(\mC_{i-1}\) is the average of all partial transposes of \(\frac12(\mA+\mA^\intercal)\) in the first \(i-1\) subsystems.
By \Cref{lem:partial-trace-of-partial-transpose}, we can extend this to be the average of all partial transposes of \(\frac12(\mA+\mA^\intercal)\) in across all subsystems \(\cV\subset[i-1] \cup \{i+1,\ldots,k\}\), so that
\begin{align*}
	\tr_{i+1:}(\mC_{i-1})
	&= \tr_{i+1:}\left(\frac1{2^{k-1}} \sum_{\cV \subseteq [k]\setminus\{i\}} (\tsfrac12(\mA+\mA^\intercal))^{\ptran\cV}\right) \\
	&= \tr_{i+1:}\left(\frac{\frac1{2^{k-1}} \sum_{\cV \subseteq [k]\setminus\{i\}} \mA^{\ptran\cV} + \frac1{2^{k-1}} \sum_{\cV \subseteq [k]\setminus\{i\}} (\mA^\intercal)^{\ptran\cV}}2\right).
\end{align*}
The left term above is the average of all partial transposes of \mA that do not include subsystem \(i\).
Since \((\mA^\intercal)^{\ptran\cV} = \mA^{\ptran{\bar\cV}}\) where \(\bar\cV \defeq [k] \setminus \cV\), the right term above is the average of all partial transposes of \mA that do include \(i\).
So, we find that \(\tr_{i+1:}(\mC_{i-1}) = \tr_{i+1:}(\frac{1}{2^k} \sum_{\cV\subseteq[k]} \mA^{\ptran\cV}) = \tr_{i+1:}(\bar\mA)\).
We are therefore left with the expression
\begin{align}
	\E[(\vx^\intercal\mA\vx)^2]
	\leq (\tr(\mA))^2 + \sum_{i=1}^{k} \sum_{\cS\subseteq[i-1]} 2^{i - \abs{\cS}} \norm{\tr_{\cS \cup \{i+1,\ldots,k\}} (\bar\mA)}_F^2. \label{eq:variance-annoying-sum}
\end{align}
We now focus on the double summation in \Cref{eq:variance-annoying-sum}.
We start by understanding the inner summation in terms of the index \(i\) from the outer iteration.
Let
\[
	F(i) \defeq \sum_{\cS\subseteq[i-1]} 2^{i - \abs{\cS}} \norm{\tr_{\cS \cup \{i+1,\ldots,k\}} (\bar\mA)}_F^2
\]
so that
\[
	\sum_{i=1}^{k} \sum_{\cS\subseteq[i-1]} 2^{i - \abs{\cS}} \norm{\tr_{\cS \cup \{i+1,\ldots,k\}} (\bar\mA)}_F^2
	= \sum_{i=1}^{k} F(i).
\]
We see that \(F(i)\) is a weighted sum of \(\norm{\tr_{\hat\cS}(\bar\mA)}_F^2\) across all \(\hat\cS \subseteq [k]\) of the form \(\hat\cS = \cS \cup \{i+1,\ldots,k\}\) where \(\cS\subseteq[i-1]\).
That is, we consider all \(\hat\cS\) that satisfies two properties:
\begin{enumerate}
	\item Indices \(i+1,\ldots,k\) are all included: \(\{i+1,\ldots,k\} \subseteq \cS\)
	\item Index \(i\) is not included: \(i \notin \hat\cS\)
\end{enumerate}
We can intuitively think of \(i\) as a ``Gap'' that precedes the ``Suffix'' \(\{i+1,\ldots,k\}\).
Crucially, if \(\norm{\tr_{\hat\cS}(\bar\mA)}_F^2\) appears in a term of \(F(i)\), then \(\norm{\tr_{\hat\cS}(\bar\mA)}_F^2\) \emph{does not} appear in a term of any other \(F(j)\).
This is because each \(\hat\cS\) has a unique gap: for any \(\hat\cS\), let \(i\in[k]\) be the smallest value such that \(\{i+1,\ldots,k\}\subseteq\hat\cS\).
Then \(i\) is the gap of \(\hat\cS\).
We note an edge-case, where if \(\hat\cS = \{1,\ldots,k\}\) then the gap is \(i=0\) and we recover the term \(\norm{\tr_{\hat\cS}(\bar\mA)}_F^2 = (\tr(\bar\mA))^2\) in \Cref{eq:variance-annoying-sum}.

Next, note that every \(\hat\cS\) is scaled by \(2^{i-\abs{\cS}}\).
Since \(|\hat\cS| = \abs{\cS \cup \{i+1,\ldots,k\}} = \abs{\cS} + k-i\), we know \(\abs{\cS} - i = |\hat\cS| - k\), so we can write \(2^{i-\abs{\cS}} = 2^{k - |\hat\cS|}\).
So, we have shown that every \(\hat\cS \subseteq[k]\) appears exactly once in \Cref{eq:variance-annoying-sum}, and each term \(\norm{\tr_{\hat\cS}(\bar\mA)}_F^2\) is scaled by \(2^{k-|\hat\cS|}\).
That is,
\begin{align*}
	\E[(\vx^\intercal\mA\vx)^2] &\leq \sum_{\hat\cS\subseteq[k]} 2^{k - |\hat\cS|} \norm{\tr_{\hat\cS}(\bar\mA)}_F^2,
\end{align*}
which completes the proof since
\[
	\Var[\vx^\intercal\mA\vx] = \E[(\vx^\intercal\mA\vx)^2]	- (\E[\vx^\intercal\mA\vx])^2
	= \sum_{\hat\cS\subset[k]} 2^{k - |\hat\cS|} \norm{\tr_{\hat\cS}(\bar\mA)}_F^2.
\]
Note that all inequalities in the proof become equalities if we know that \(\vx_1,\ldots,\vx_k\) are all iid standard normal Gaussians.
\end{proof}

\subsection{Proof of \Cref{thm:real-partial-trace}}
\label{sec:proof-of-partial-trace}

The variance expression in \Cref{thm:real-partial-trace-transpose} is extremely tight, being an equality in the case that our query vectors are Kroneckers of Gaussian vectors.
However, it is not immediately clear how we should interpret this expression.
We now prove \Cref{thm:real-partial-trace}, which is a more approachable upper bound of \Cref{thm:real-partial-trace-transpose}.
We start with a short lemma:
\begin{lemma}
\label{lem:partial-trace-frob-invariant}
Let \(\mA\in\bbR^{d^k \times d^k}\) and \(\cV\subseteq[k]\). Then, \(\norm{\mA^{\ptran\cV}}_F = \norm{\mA}_F\).
\end{lemma}
\begin{proof}
Consider any \(i \in \cV\).
Then, recall that each standard basis \(\ve_{\hat i} \in \bbR^{d^k}\) can be decomposed as \(\ve_{\hat i} = \ve_{\hat j} \otimes \ve_{\hat k} \otimes \ve_{\hat \ell}\) where \(\ve_{\hat j}\in\bbR^{d^{i-1}}\), \(\ve_{\hat k}\in\bbR^d\), and \(\ve_{\hat \ell}\in\bbR^{d^{k-i}}\).
Then, we can expand the norm of \(\mA\) as
\begin{align*}
	\norm{\mA^{\ptran{\{i\}}}}_F^2
	&= \sum_{\hat i_1, \hat i_2 = 1}^{d^k}
		(\ve_{\hat i_1}^\intercal \mA^{\ptran{\{i\}}} \ve_{\hat i_2})^2 \\
	&= \sum_{\hat j_1, \hat j_2 = 1}^{d^{i-1}} \sum_{\hat k_1, \hat k_2 = 1}^{d} \sum_{\hat \ell_1, \hat \ell_2 = 1}^{d^{k-1}}
		((\ve_{\hat j_1} \otimes \ve_{\hat k_1} \otimes \ve_{\hat \ell_1})^\intercal\mA^{\ptran{\{i\}}}(\ve_{\hat j_2} \otimes \ve_{\hat k_2} \otimes \ve_{\hat \ell_2}))^2 \\
	&= \sum_{\hat j_1, \hat j_2 = 1}^{d^{i-1}} \sum_{\hat k_1, \hat k_2 = 1}^{d} \sum_{\hat \ell_1, \hat \ell_2 = 1}^{d^{k-1}}
		\left(
			\sum_{\hat m,\hat n=1}^{d}
			(\ve_{\hat j_1}^\intercal \otimes \ve_{\hat k_1}^\intercal\ve_{\hat m}\ve_{\hat n}^\intercal \otimes \ve_{\hat \ell_1}^\intercal) \mA (\ve_{\hat j_2} \otimes \ve_{\hat m}\ve_{\hat n}^\intercal\ve_{\hat k_2} \otimes \ve_{\hat \ell_2})
		\right)^2 \\
	&= \sum_{\hat j_1, \hat j_2 = 1}^{d^{i-1}} \sum_{\hat k_1, \hat k_2 = 1}^{d} \sum_{\hat \ell_1, \hat \ell_2 = 1}^{d^{k-1}} 
		((\ve_{\hat j_1}^\intercal \otimes \ve_{\hat k_2}^\intercal \otimes \ve_{\hat \ell_1}^\intercal) \mA (\ve_{\hat j_2} \otimes \ve_{\hat k_1} \otimes \ve_{\hat \ell_2}))^2 \\
	&= \sum_{\hat i_1, \hat i_2 = 1}^{d^k} \ve_{\hat i_1}^\intercal \mA \ve_{\hat i_2} \\
	&= \norm{\mA}_F^2,
\end{align*}
where we remove the sum over \(\hat m\) and \(\hat n\) be noting that \(\ve_{\hat k_1}^\intercal\ve_{\hat m}\) and \(\ve_{\hat n}^\intercal \ve_{\hat k_2}\) are zero unless \(\hat m = \hat k_1\) and \(\hat n = \hat k_2\).
We then complete the lemma by decomposing \(\mA^{\ptran\cV} = (((\mA^{\ptran{\{i_1\}}})^{\ptran{\{i_2\}}})\cdots)^{\ptran{\{i_{\abs\cV}\}}}\), where \(\cV = \{i_1,\ldots,i_{\abs\cV}\}\).
\end{proof}
We now prove the theorem:
\begin{reptheorem}{thm:real-partial-trace}
Fix \(d,k\in\bbN\), and let \(\mA\in\bbR^{d^k \times d^k}\).
Let \(\vx = \vx_1 \otimes \cdots \otimes \vx_k\), where each \(\vx_i\) is independent and is either a Rademacher vector or a \(\cN(\vec0,\mI)\) Gaussian vector.
Then,
\(\Var[\vx^\intercal\mA\vx] \leq \sum_{\cS\subset[k]} 2^{k - \abs{\cS}} \norm{\tr_{\cS}(\mA)}_F^2\).
Furthermore, if all \(\vx_i\) are \(\cN(\vec0,\mI)\) and if \(\mA=\bar\mA\), then this expression is an equality.
\end{reptheorem}
\begin{proof}
We prove that \(\norm{\tr_\cS(\bar\mA)}_F \leq \norm{\tr_\cS(\mA)}_F\), which implies the theorem.
First, recall that the partial trace is linear, so that by the triangle inequality we can write
\[
	\norm{\tr_\cS(\bar\mA)}_F
	= \norm{\tsfrac1{2^k} {\textstyle\sum_{\cV \subseteq [k]}} \tr_\cS(\mA^{\ptran\cV})}_F
	\leq \tsfrac1{2^k} {\textstyle\sum_{\cV \subseteq [k]}} \norm{\tr_\cS(\mA^{\ptran\cV})}_F.
\]
Next, we examine the norm \(\norm{\tr_\cS(\mA^{\ptran\cV})}\).
First, by \Cref{lem:partial-trace-of-partial-transpose}, we can assume without loss of generality that \(\cV \cap \cS = \emptyset\).
Then, for notationally simplicity, suppose that \(\cV = \{1,\ldots,\abs\cV\}\).
We can then write \(\tr_\cS{\mA^{\ptran\cV}} = (\tr_\cS(\mA))^{\ptran\cV}\).
To see this, fix any \(i \in \cV\) and \(j \in \cS\).
By our without-loss-of-generality claim, we know that \(i < j\), and can therefore expand
\begin{align*}
	\tr_{\{j\}}(\mA^{\ptran{\{i\}}})
	&= \sum_{\hat i = 1}^d
		(\mI^{\otimes j-1} \otimes \ve_{\hat i} \otimes \mI^{\otimes k-j})^\intercal
		\mA^{\ptran{\{i\}}}
		(\mI^{\otimes j-1} \otimes \ve_{\hat i} \otimes \mI^{\otimes k-j}) \\
	&= \sum_{\hat i, \hat j, \hat \ell = 1}^d
		(\mI^{\otimes i-1} \otimes \ve_{\hat j}\ve_{\hat \ell}^\intercal \otimes \mI^{j-i-1} \otimes \ve_{\hat i}^\intercal \otimes \mI^{\otimes k-j})
		\mA
		(\mI^{\otimes i-1} \otimes \ve_{\hat j}\ve_{\hat \ell}^\intercal \otimes \mI^{j-i-1} \otimes \ve_{\hat i} \otimes \mI^{\otimes k-j}) \\
	&= \sum_{\hat i, \hat j, \hat \ell = 1}^d
		(\mI^{\otimes i-1} \otimes \ve_{\hat j}\ve_{\hat \ell}^\intercal \otimes \mI^{k-i})
		\tr_{\{j\}}(\mA)
		(\mI^{\otimes i-1} \otimes \ve_{\hat j}\ve_{\hat \ell}^\intercal \otimes \mI^{k-i}) \\
	&= (\tr_{\{j\}}(\mA))^{\ptran{\{i\}}}.
\end{align*}
So, we can pull all of the partial transposes from inside of the partial trace to outside the partial trace.
Then, by \Cref{lem:partial-trace-frob-invariant}, we can conclude that
\begin{align*}
	\norm{\tr_\cS(\bar\mA)}_F
	&\leq \tsfrac1{2^k} {\textstyle\sum_{\cV \subseteq [k]}} \norm{\tr_\cS(\mA^{\ptran\cV})}_F \\
	&= \tsfrac1{2^k} {\textstyle\sum_{\cV \subseteq [k]}} \norm{(\tr_\cS(\mA))^{\ptran\cV}}_F \\
	&= \tsfrac1{2^k} {\textstyle\sum_{\cV \subseteq [k]}} \norm{\tr_\cS(\mA)}_F \\
	&= \norm{\tr_\cS(\mA)}_F.
\end{align*}
\end{proof}

While \Cref{thm:real-partial-trace} is more approachable than \Cref{thm:real-partial-trace-transpose}, it is still not immediately clear how to think of this rate.
In order to get a better understanding of this variance, we turn to the worst-case performance of this estimator for the class of PSD matrices.
First however, we take a moment to consider what happens when we use uniformly random vectors on the sphere in place of Gaussian vectors.

\subsection{
	Exact Variance with Uniformly Random Vectors on the Sphere
}
\label{sec:exact-var-unit-vec}

In this section we consider using \(\vu = \vu_1 \otimes \cdots \vu_k\) where \(\vu_i\) is a uniformly random vector such that \(\norm{\vu_i}_2 =\sqrt d\).
These vectors are isotropic (i.e. \(\E[\vu=\vec0]\) and \(\E[\vu\vu^\intercal]=\mI\)), so we know the Kronecker-Hutchinson estimator is unbiased when using \(\vu\).
In the non-Kronecker setting this vector is commonly used in practice, as it converges more quickly than using the Gaussian vectors \cite{girard1989fast,epperlyStochastic}.

We can analyze the variance of the Kronecker-Hutchinson estimator when using \(\vu\) by recalling that the product of a random unit vector in \(\bbR^d\) and a \(\chi_d^2\) random variable is distributed as a \(\cN(\vec0,\mI)\) Gaussian vector.
Equivalently, we can decompose \(\cN(\vec0,\mI)\) vector \(\vx_i = \sqrt{\frac{\gamma_i}{d}} \vu_i\) where \(\gamma_i \sim \chi_d^2\) and \(\norm{\vu_i}=\sqrt d\).
Then, by the law of total variation, we can prove \Cref{corol:random-unit-vec-variance}.
\begin{repcorollary}{corol:random-unit-vec-variance}
Fix \(d,k\in\bbN\) and let \(\mA\in\bbR^{d^k \times d^k}\).
Let \(\vu = \vu_1 \otimes \cdots \otimes \vu_k\) where each \(\vu_i\) is uniformly at random drawn from the sphere such that \(\norm{\vu_i}_2 = \sqrt d\).
Let \(\vx = \vx_1 \otimes \cdots \otimes \vx_k\) where each \(\vx_i\) is \(\cN(\vec0,\mI)\).
Then, \(\Var[\vu^\intercal\mA\vu] = \frac{\Var[\vx^\intercal\mA\vx] - ((1+\frac2d)^k-1)(\tr(\mA))^2}{(1+\frac2d)^k} \leq \frac1{(1+\frac2d)^k} \Var[\vx^\intercal\mA\vx]\).
\end{repcorollary}
\begin{proof}
We decompose each \(\vx_i = \sqrt{\frac{\gamma_i}{d}} \vu_i\) as mentioned above.
Then, we can write \(\vx = \sqrt{\frac{\gamma}{d^k}} \vu\) where \(\gamma = \prod_{i=1}^k \gamma_i\).
Note that \(\E[\gamma_i^2] = d^2(1+\frac2d)\), which implies that \(\E[\gamma^2] = d^{2k}(1+\frac2d)^k\) and \(\Var[\gamma] = d^{2k}(1+\frac2d)^k - d^{2k}\).
Then, by the law of total variation, we can write
\begin{align*}
	\Var[\vx^\intercal\mA\vx]
	&= \Var[\tsfrac{\gamma}{d^k} \vu^\intercal\mA\vu] \\
	&= \Var[\E[\tsfrac{\gamma}{d^k} \vu^\intercal\mA\vx \mid \gamma]] + \E[\Var[\tsfrac{\gamma}{d^k} \vu^\intercal\mA\vu \mid \gamma]] \\ 
	&= \Var[\tsfrac{\gamma}{d^k} \tr(\mA)] + \E[\tsfrac{\gamma^2}{d^{2k}} \Var[\vu^\intercal\mA\vu]] \\ 
	&= \frac{\Var[\gamma]}{d^{2k}}(\tr(\mA))^2 + \frac{\E[\gamma^2]}{d^{2k}} \Var[\vu^\intercal\mA\vu] \\
	&= ((1+\frac2d)^k-1)(\tr(\mA))^2 + (1+\frac2d)^k \Var[\vu^\intercal\mA\vu]
\end{align*}
which implies that
\begin{align*}
	\Var[\vu^\intercal\mA\vu]
	&= \frac{\Var[\vx^\intercal\mA\vx] - ((1+\frac2d)^k-1)(\tr(\mA))^2}{(1+\frac2d)^k},
\end{align*}
completing the proof.
\end{proof}

%% file: psd-trace-estimation.tex
% Clean

\section{Cost of the Kronecker-Hutchinson Estimator for PSD Matrices}
\label{sec:real-psd}

Trace estimation is often studied under the assumption that the input matrix \mA is PSD.
This assumption is very common and reasonable in practice, and without such an assumption, it is impossible for an algorithm such as Hutchinson's Estimator to achieve a \((1\pm\eps)\) relative error approximation to the trace of \mA.
This is because an indefinite matrix \mA can have \(\tr(\mA) = 0\) while having large Frobenius norm, meaning that Hutchinson's Estimator would have to perfectly recover the trace of \mA despite having a nonzero variance.
This would incur infinite sample complexity.
However, if we assume that \mA is PSD, then \(\norm{\mA}_F \leq \tr(\mA)\), and so we know that the standard deviation of the classical Hutchinson's Estimator is \(\sqrt{\frac{2}{\ell}}\norm{\mA}_F \leq \sqrt{\frac2\ell}\tr(\mA)\).
That is, \(\ell = O(\frac1{\eps^2})\) samples suffice to achieve \((1\pm\eps)\) relative error with at least constant probability.

We now make the same assumption and study the Kronecker-Hutchinson Estimator, and show that \(\ell = O(\frac{3^k}{\eps^2})\) samples suffices to estimate the trace of a PSD matrix to \((1\pm\eps)\) relative error.
We prove this in two different ways, one that follows from \Cref{thm:real-partial-trace} which suffices to get the worst-case complexity from Gaussian vectors and one that uses Minkowski's integral inequality to a tighter complexity across more choices of base distribution \(\cD'\).

To start, we need a basic property of the partial trace:
\begin{lemma}
\label{lem:partial-trace-is-psd}
Let \(\mA\in\bbR^{d^k \times d^k}\) be a PSD matrix, and let \(\cS\subseteq[k]\).
Then, \(\tr_\cS(\mA)\) is also PSD.
\end{lemma}
\begin{proof}
Recall that if \mA is PSD, then \(\mB^\intercal\mA\mB\) is also PSD for any matrix \mB of compatible dimension.
Further recall that if \mA and \mC are both PSD, then \(\mA + \mC\) is also PSD.
Since \(\tr_\cS(\mA)\) is a composition of such operations on \mA, we get that \(\tr_\cS(\mA)\) is also PSD.
\end{proof}
\begin{lemma}[Folklore]
\label{lem:psd-l2-l1-norm}
Let \(\mB\) be a PSD matrix. Then, \(\norm{\mB}_F \leq \tr(\mB)\).
\end{lemma}
We now prove our next core theorem:
\begin{reptheorem}{thm:real-psd}
Fix \(d,k\in\bbN\), and let \(\mA\in\bbR^{d^k \times d^k}\) be a PSD matrix.
Let \(\vx = \vx_1 \otimes \cdots \otimes \vx_k\), where each \(\vx_i\) is either a Rademacher vector or a \(\cN(\vec0,\mI)\) Gaussian vector.
Then,
\(\Var[\vx^\intercal\mA\vx] \leq 3^k (\tr(\mA))^2\).
\end{reptheorem}
\begin{proof}
We first apply \Cref{thm:real-partial-trace,lem:psd-l2-l1-norm,lem:trace-of-partial-trace} to get
\[
	\E[(\vx^\intercal\mA\vx)^2]
	\leq \sum_{\cS \subseteq [k]} 2^{k - \abs{\cS}} \norm{\tr_\cS(\mA)}_F^2
	\leq \sum_{\cS \subseteq [k]} 2^{k - \abs{\cS}} (\tr(\tr_\cS(\mA)))^2
	= (\tr(\mA))^2 \sum_{\cS \subseteq [k]} 2^{k - \abs{\cS}}.
\]
So we turn our attention to bounding this sum.
If we let \(\sigma_i \defeq \mathbbm{1}_{[i \notin \cS]}\), then we can write \(k - \abs{\cS} = \sum_i \sigma_i\), and so
\begin{align*}
	\sum_{\cS \subseteq [k]} 2^{k - \abs{\cS}}
	= \sum_{\vsigma\in\{0,1\}^k} 2^{\sum_i \sigma_i}
	= \sum_{\sigma_1\in\{0,1\}} 2^{\sigma_1} \left(\cdots \left( \sum_{\sigma_{k-1}\in\{0,1\}} 2^{\sigma_{k-1}} \left(\sum_{\sigma_{k}\in\{0,1\}} 2^{\sigma_{k}}\right)\right)\right)
	= 3^k.
\end{align*}
We then conclude that \(\Var[\vx^\intercal\mA\vx] \leq \E[(\vx^\intercal\mA\vx)^2] \leq 3^k (\tr(\mA))^2\).
\end{proof}
In particular, if we take \(\ell\) samples in the Kronecker-Hutchinson Estimator, we get a standard deviation of \(\sqrt{\Var[H_\ell(\mA)]} \leq \sqrt{\frac{3^k}{\ell}} (\tr(\mA))\), which is less than \(\eps\tr(\mA)\) for \(\ell \geq \frac{3^k}{\eps^2}\).
More formally, by Chebyshev's inequality, we find that  \(\ell = O(\frac{3^k}{\eps^2})\) samples suffices to get
\[
	(1-\eps)\tr(\mA) \leq H_\ell(\mA) \leq (1+\eps)\tr(\mA)
\]
with constant probability.

Notably, if the base dimension \(d\) is small, then we can show that the Kronecker product of Rademacher vectors or vectors uniformly drawn from the sphere achieve a much faster rate of convergence in the worst case for PSD matrices.
This is because, for PSD matrices, these distributions have a lower variance than Gaussian vectors in the classical (\(k=1\)) case.
In particular, they achieve different values of \(C\) in the following theorem:
\begin{reptheorem}{thm:coarse-psd-rate}
    Let \(\cD'\) be a distribution over vectors \(\vx_i\in\bbR^d\) such that \(\E[\vx_i]=\vec0\), \(\E[\vx_i\vx_i^\intercal] = \mI\), and \(\Var[\vx_i^\intercal\mB\vx_i] \leq C(\tr(\mB))^2\) for all PSD \mB.
    Then, letting \(\vx = \vx_1 \otimes \cdots \otimes \vx_k\), we have \(\Var[\vx^\intercal\mA\vx] \leq (1+C)^k (\tr(\mA))^2\).
    In particular, taking \(\ell = O(\frac{(1+C)^k}{\eps^2})\) suffices for \(H_\ell(\mA)\) to have standard deviation at most \(\eps\tr(\mA)\).
\end{reptheorem}
We will prove \Cref{thm:coarse-psd-rate} momentarily, but first we show how different choices of base distribution achieve different values of \(C\).
In particular, we will prove that when \(\cD'\) is Rademacher then \(C = 2-\frac2d\), which implies a complexity of \(\ell = O(\frac{1}{\eps^2}(3-\frac2d)^k)\).
Similary, we will also how that when \(\cD'\) is uniformly distributed on the sphere, \(\ell = O(\frac{1}{\eps^2}(3-\frac6{d+2})^k)\) samples suffice.
These bounds on \(C\) let us finish proving \Cref{thm:real-psd}:
\begin{reptheorem}{thm:real-psd}
Fix \(d,k\in\bbN\), and let \(\mA\in\bbR^{d^k \times d^k}\) be a PSD matrix.
Let \(\vx = \vx_1 \otimes \cdots \otimes \vx_k\) where each \(\vx_i\) is drawn iid from a distribution \(\cD'\).
If \(\cD'\) is Gaussian, then \(\Var[\vx^\intercal\mA\vx] \leq 3^k (\tr(\mA))^2\).
If \(\cD'\) is Rademacher, then \(\Var[\vx^\intercal\mA\vx] \leq (3-\frac2d)^k (\tr(\mA))^2\).
If \(\cD'\) is uniformly distributed over vectors with \(\norm{\vx_i}_2=\sqrt{d}\), then \(\Var[\vx^\intercal\mA\vx] \leq (3-\frac6{d+2})^k (\tr(\mA))^2\).
\end{reptheorem}
Now we prove the result of \Cref{thm:coarse-psd-rate}:
\begin{proof}
We prove that \(\Var[\vx^\intercal\mA\vx] \leq \E[(\vx^\intercal\mA\vx)^2] \leq (1+C)^k (\tr(\mA))^2\) via induction over \(k\).
In the base case where \(k=1\), we have \(\vx=\vv\) and use our assumptions on \cD to say that
\[
	\E[(\vx^\intercal\mA\vx)^2]
	= \Var[\vv^\intercal\mA\vv] + (\E[\vv^\intercal\mA\vv])^2
	\leq (1+C)(\tr(\mA))^2.
\]
We now tackle the inductive case.
Let \(\vx_{:k-1} \defeq \otimes_{i=1}^{k-1}\vx_i\).
Then, we inductively assume that for any PSD matrix \(\mB \in\bbR^{d^{k-1} \times d^{k-1}}\), we know that \(\E[(\vx_{:k-1}^\intercal\mB\vx_{:k-1})^2]\leq (1+C)^{k-1}(\E[\vx_{:k-1}^\intercal\mB\vx_{:k-1}])^2\).
We will then extend the claim to hold for any \(\mA\in\bbR^{d^k \times d^k}\).
By \Cref{lem:kron-extract-vec}, we know that
\[
	\vx^\intercal\mA\vx
	=(\vx_{:k-1}\otimes\vx_k)^\intercal\mA(\vx_{:k-1}\otimes\vx_k)
	=\vx_k^\intercal(\vx_{:k-1}\otimes\mI)^\intercal\mA(\vx_{:k-1}\otimes\mI)\vx_k.
\]
So, by taking \(\mB = (\vx_{:k-1}\otimes\mI)^\intercal\mA(\vx_{:k-1}\otimes\mI) \in \bbR^{d \times d}\), we know that
\begin{align}
	\E_{\vx_k}[(\vx^\intercal\mA\vx)^2]
	= \E_{\vx_k}[(\vx_k^\intercal\mB\vx_k)^2]
	\leq (1+C) (\E_{\vx_k} [\vx_k^\intercal\mB\vx_k])^2
	= (1+C) (\E_{\vx_k} [\vx^\intercal\mA\vx])^2
	\label{eq:squared-expect-bound-rademacher}
\end{align}
Then, by taking Minkowski's Integral Inequality with \(F(\vx_{:k-1},\vx_k) = \vx^\intercal\mA\vx\) and \(r=2\), we know that
\[
	\textstyle{\E_{\vx_{:k-1}}[(\E_{\vx_k} [\vx^\intercal\mA\vx])^2] \leq \left(\E_{\vx_k}\left[ \sqrt{\E_{\vx_{:k-1}} (\vx^\intercal\mA\vx)^2} \right]\right)^2}.
\]
Then, letting \(\mC = (\mI \otimes \vx_k)^\intercal\mA(\mI\otimes\vx_k) \in\bbR^{d^{k-1} \times d^{k-1}}\), we find that
\begin{align*}
	\E[(\vx^\intercal\mA\vx)^2]
	&= \E_{\vx_{:k-1}}\left[ \E_{\vx_k} (\vx^\intercal\mA\vx)^2 \right] \\
	&\leq (1+C) \, \E_{\vx_{:k-1}}\left[ \left(\E_{\vx_k} \vx^\intercal\mA\vx\right)^2 \right] \tag{\Cref{eq:squared-expect-bound-rademacher}} \\
	&\leq (1+C) \left(\E_{\vx_k}\left[ \sqrt{\E_{\vx_{:k-1}} (\vx^\intercal\mA\vx)^2} \right]\right)^2 \tag{Minkowski} \\
	&= (1+C) \left(\E_{\vx_k}\left[ \sqrt{\E_{\vx_{:k-1}} (\vx_{:k-1}^\intercal\mC\vx_{:k-1})^2} \right]\right)^2 \tag{\(\vx^\intercal\mA\vx = \vx_{:k-1}^\intercal\mC\vx_{:k-1}\)} \\
	&\leq (1+C) \left(\E_{\vx_k}\left[ \sqrt{(1+C)^{k-1} (\E_{\vx_{:k-1}}[\vx_{:k-1}^\intercal\mC\vx_{:k-1}])^2} \right]\right)^2 \tag{Induction} \\
	&= (1+C) \left(\E_{\vx_k}\left[ \sqrt{(1+C)^{k-1} (\E_{\vx_{:k-1}}[\vx^\intercal\mA\vx])^2} \right]\right)^2 \tag{\(\vx^\intercal\mA\vx = \vx_{:k-1}^\intercal\mC\vx_{:k-1}\)} \\
	&= (1+C)^k \left(\E_{\vx_k} [\E_{\vx_{:k-1}}[\vx^\intercal\mA\vx]] \right)^2\\
	&= (1+C)^k (\E[\vx^\intercal\mA\vx])^2 \\
	&= (1+C)^k (\tr(\mA))^2,
\end{align*}
which completes the proof.
\end{proof}

All that remains to prove \Cref{thm:real-psd} is that \(C=2-\frac2d\) when \(\cD'\) is Rademacher and that \(C = 2-\frac{6}{d+2}\) when \(\cD'\) is uniformly distributed on the sphere.
We start with the following exact expression for the non-Kronecker variance of Hutchinson's Estimator with Rademacher vectors.
\begin{lemma}[\cite{girard1987algorithme,hutchinson1989stochastic,AT11}]
\label{lem:hutch-classical-variance-rademacher}
	Let \(\mB \in \bbR^{D \times D}\) and let \(\vv \in \bbR^D\) be a vector of iid Rademacher entries.
	Let \(\hat\mB = \frac12(\mB + \mB^\intercal)\).
	Then, \(\E[\vv^\intercal\mB\vv] = \tr(\mB)\) and \(\Var[\vv^\intercal\mB\vv] = 2(\norm{\hat\mB}_F^2 - \norm{\diag{\hat \mB}}_2^2)\), where \(\diag(\hat\mB)=\diag(\mB)\in\bbR^D\) is the diagonal of \(\hat\mB\).
\end{lemma}
We then show a bound on this variance when \(\mB\) is PSD:
\begin{corollary}
	\label{corol:rademacher-variance-bound}
	Let \(\mB \in \bbR^{D \times D}\) be a PSD matrix and let \(\vv \in \bbR^D\) be a vector of iid Rademacher entries.
	Then, \(\E[\vv^\intercal\mB\vv] = \tr(\mB)\) and \(\Var[\vv^\intercal\mB\vv] \leq (2-\frac2d)(\tr(\mB))^2\).
\end{corollary}
\begin{proof}
	Owing to \Cref{lem:hutch-classical-variance-rademacher}, it suffices to bound \(\norm{\mB}_F^2 - \norm{\diag{\mB}}_2^2 \leq (1-\frac1d)\tr(\mB)\).
	First, for any \(i,j\in[d]\) with \(i\neq j\), consider the submatrix \(\sbmat{[\mB]_{ii} & [\mB]_{ij} \\ [\mB]_{ji} & [\mB]_{jj}}\) of \mB.
	Since \mB is PSD, so is the submatrix.
	So, the determinant of the submatrix is PSD, and therefore \([\mB]_{ij}^2 \leq [\mB]_{ii}[\mB]_{jj}\).
	Separately, by the AMGM inequality, note that
	\[
		\sum_{i=1}^d \sum_{j \neq i} [\mB]_{ii} [\mB]_{jj}
		\leq \frac12 \sum_{i = 1}^d \sum_{j \neq i} ([\mB]_{ii}^2 + [\mB]_{jj}^2)
		= \sum_{i = 1}^d \sum_{j \neq i} [\mB]_{ii}^2
		= (d-1) \sum_{i=1}^d [\mB]_{ii}^2.
	\]
	Combining this determinant bound with this AMGM bound, we find that
	\begin{align*}
		\norm{\mB}_F^2 - \norm{\diag{\mB}}_2^2
		&= \sum_{i = 1}^d \sum_{j \neq i} [\mB]_{ij}^2 \\
		&\leq \sum_{i = 1}^d \sum_{j \neq i} [\mB]_{ii}[\mB]_{jj} \\
		&= \frac1d \sum_{i = 1}^d \sum_{j \neq i} [\mB]_{ii}[\mB]_{jj} + (1-\frac1d) \sum_{i = 1}^d \sum_{j \neq i} [\mB]_{ii}[\mB]_{jj} \\
		&\leq \frac1d(d-1) \sum_{i = 1}^d [\mB]_{ii}^2 + (1-\frac1d) \sum_{i = 1}^d \sum_{j \neq i} [\mB]_{ii}[\mB]_{jj} \\
		&= (1-\frac1d) \sum_{i = 1}^d\sum_{j=1}^d [\mB]_{ii}[\mB]_{jj} \\
		&= (1-\frac1d) (\tr(\mB))^2
	\end{align*}
\end{proof}

We can also show the rate when \(\cD'\) is randomly drawn from the sphere:
\begin{lemma}[\cite{girard1989fast,epperlyStochastic}]
\label{lem:hutch-classical-variance-unitvec}
	Let \(\mB \in \bbR^{D \times D}\) and let \(\vv \in \bbR^D\) be a uniformly random vector such that \(\norm\vv_2 = \sqrt D\).
	Let \(\hat\mB = \frac12(\mB + \mB^\intercal)\).
	Then, \(\E[\vv^\intercal\mB\vv] = \tr(\mB)\) and \(\Var[\vv^\intercal\mB\vv] = 2\frac{d}{d+2}(\norm{\hat\mB}_F^2 - \frac1d(\tr(\hat\mB))^2)\).
\end{lemma}
By bounding \(\norm{\mB}_F \leq \tr(\mB)\), we immediately get the following result:
\begin{corollary}
	\label{corol:unitvec-variance-bound}
	Let \(\mB \in \bbR^{D \times D}\) be a PSD matrix and let \(\vv \in \bbR^D\) be a uniformly random vector such that \(\norm\vv_2=\sqrt{D}\).
	Then, \(\E[\vv^\intercal\mB\vv] = \tr(\mB)\) and \(\Var[\vv^\intercal\mB\vv] \leq (2-\frac6{d+2})(\tr(\mB))^2\).
\end{corollary}
These results then suffices to get the results from \Cref{table:exp-rates} in the real-valued columns.
For small \(d\), this is an enormous speedup over using Gaussian vectors in the worst case, though we fail to find an exact expression for the variance of this estimator.

\subsection{A Matching Lower Bound}
\label{sec:real-psd-lower-bound}
We might hope that the painful rate of \(\ell \geq (3-\frac2d)^k\frac{1}{\eps^2}\) results from a loose analysis of the Kronecker-Hutchinson Estimator.
Or, perhaps, that some more clever choice of random vector distribution can avoid this painful exponential dependence.
We show that this is not the case, at least not among vectors with real iid mean-zero unit-variance components:
\begin{reptheorem}{thm:real-psd-lower-bound}
	Fix \(d,k\in\bbN\).
	Consider Hutchinson's Estimator run with independent vectors \(\vx = \vx_1 \otimes \cdots \otimes \vx_k\) where \(\vx_i\) are vectors of real iid zero-mean unit-variance entries.
	Then, there exists a PSD matrix \mA such that Hutchinson's Estimator needs \(\ell = \Omega(\frac{(3-\frac2d)^k}{\eps^2})\) samples in order to have standard deviation \(\leq \eps \tr(\mA)\).
\end{reptheorem}
\begin{proof}
We let \(\mA \in \bbR^{d^k \times d^k}\) be the all-ones matrix.
If we let \(\ve\in\bbR^d\) denote the all-ones vector, then we can write \(\mA = \otimes_{i=1}^k \ve\ve^\intercal\).
This is helpful because we can then use the Mixed-Product property (\Cref{lem:mixed-product}) to show that
\[
	\textstyle{
	\vx^\intercal\mA\vx
	= \otimes_{i=1}^k \vx_i^\intercal\ve\ve^\intercal\vx_i
	= \prod_{i=1}^k (\ve^\intercal\vx_i)^2
	},
\]
where we can replace the Kronecker product with a scalar product because the Kronecker product of scalars is just the scalar product.
Notably, we see that \(\vx^\intercal\mA\vx\) is a product of \(k\) iid terms.
Since we are interested in the variance of the estimator, we actually want to examine
\[
	\E[(\vx^\intercal\mA\vx)^2]
	= \prod_{i=1}^k \E[(\ve^\intercal\vx_i)^4]
	= \left(\E[(\ve^\intercal\vx_1)^4]\right)^k.
\]
So, letting \(\vy \defeq \vx_1\) in order to simplify notation, we turn our attention to lower bounding the expectation of \((\ve^\intercal\vy)^4 = (\sum_{i=1}^d y_i)^4\) where \(\vy\) is a vector of iid zero-mean unit-variance terms.
We expand this sum as
\[
	\textstyle \E[(\sum_{i=1}^d y_i)^4] = \sum_{i,j,m,n=1}^d \E[y_i y_j y_m y_n].
\]
This expectation is zero if any random variable appears exactly once, like how \(\E[y_1y_2^3] = \E[y_1]\E[y_2^3] = 0\) by independence.
So, we can reduce the sum above to terms that either look like \(\E[y_i^4]\) or like \(\E[y_i^2y_j^2]\).
Terms like \(\E[y_i^4]\) occur exactly \(d\) times.
Terms like \(\E[y_i^2y_j^2]\) occur exactly \(3d^2-3d\) times\footnote{This can happen 3 ways: if \(i = j \neq m = n\) or if \(i = m \neq j = n\) or if \(i = n \neq j = m\). Each such pattern occurs \(d^2-d\) times, so the overall pattern occurs \(3d^2-3d\) times.}.
We can then bound \(\E[y_i^4] \geq (\E[y_i^2])^2 = 1\) by Jensen's Inequality, and \(\E[y_i^2y_j^2] = \E[y_i^2]\E[y_j^2] = 1\) by independence.
So, we get that 
\[
	\E[(\vx^\intercal\mA\vx)^2]
	= \left(\E[(\ve^\intercal\vy)^4]\right)^k
	\geq \left(d + 3d^2 - 3d\right)^k
	= (3d^2 - 2d)^k.
\]
Which, noting that \(\tr(\mA)=d^k\), in turn gets a variance lower bound of
\[
	\Var[\vx^\intercal\mA\vx]
	= \E[(\vx^\intercal\mA\vx)^2] - (\tr(\mA))^2
	\geq (3d^2 - 2d)^k - d^{2k}.
\]
Then, since \(\Var[H_\ell(\mA)] = \frac1\ell \Var[\vx^\intercal\mA\vx]\), in order for the Kronecker-Hutchinson's Estimator to have standard deviation \(\eps\tr(\mA)\), we need to take
\[
    \ell
    \geq \frac{\Var[\vx^\intercal\mA\vx]}{\eps^2(\tr(\mA))^2}
    \geq \frac{(3d^2 - 2d)^k - d^{2k}}{\eps^2 d^{2k}}
    = \frac{(3-\frac2d)^k - 1}{\eps^2}
    = \Omega\left(\frac{(3-\frac2d)^k}{\eps^2}\right).
\]
\end{proof}
We note that this lower bound uses a rank-one matrix, which is interesting because we can exactly compute the trace of any rank-one matrix with a single Kronecker-matrix-vector product.
In particular, if \(\mA = \vu\vu^\intercal\) for some \(\vu\in\bbR^d\), then for any vector \vx not orthogonal to \vu, we have that
\[
    \frac{\norm{\mA\vx}_2^2}{\vx^\intercal\mA\vx}
    = \frac{\vx^\intercal\vu \cdot \vu^\intercal\vu \cdot \vu^\intercal\vx}{\vx^\intercal\vu \cdot \vu^\intercal\vx}
    = \vu^\intercal\vu
    = \tr(\mA).
\]
Since a Kronecker of Gaussian vectors is not orthogonal to any fixed vector \vu with probability 1, we can exactly compute this trace with probability 1.
That is to say, the hard instance used in this lower bound is truly only hard for the Kronecker-Hutchinson Estimator, and is not hard in general for the Kronecker-matrix-vector oracle model.

The lower bound is also structured in a second way -- the matrix \(\mA = \otimes_{i=1}^k \ve\ve^\intercal\) is not only rank-one but also has Kronecker structure.
In \Cref{sec:kron-recovery} we prove that we can exactly compute the trace of any matrix \mA with Kronecker structure by using \(kd+1\) Kronecker-matrix-vector products.
That is to say, the hard instance used in this lower bound is in a second sense only hard for the Kronecker-Hutchinson's Estimator, and its class of hard input matrices is not hard in general for the Kronecker-matrix-vector oracle model.
That said, we do not have a clear sense of what algorithms may be able to bridge this gap and admit much faster trace estimation.

\subsection{Estimating the Trace of a Random Rank-One Matrix}
\label{sec:real-rank-one-estimation}
To round off our discussion of PSD matrix trace estimation, we show that estimating the trace of a random rank-1 matrix is much more efficient than estimating the trace of the worst-case matrix.
When the base distribution \(\cD'\) is Gaussian, we achieve a much better exponential dependence on \(k\) compared to \Cref{thm:real-psd}.
When \(\cD'\) is either Rademacher or uniformly drawn from the sphere, we avoid any dependence on \(k\) at all.

This proof tells several stories.
First, we can consider the context of the lower bound \Cref{thm:real-psd-lower-bound}, which uses the all-ones matrix as the worst-case input \mA.
The all-ones matrix is both rank-1 and has Kronecker structure.
We might then wonder if \mA having low rank but not having Kronecker structure forces us to have a sample complexity that is exponential in \(k\).
We answer this in the negative, so long as we avoid using Gaussian vectors for the base distribution \(\cD'\).
This suggests that the Kronecker-Hutchinson Estimator only has sample complexity that is exponential in \(k\) when \(\mA\) has Kronecker structure.
However, we have no rigorous result that formalizes this intuition.
Second, we show that the upper bound on the variance of the Gaussian estimator from \Cref{thm:real-partial-trace} (i.e. using \(\mA\) intead of \(\bar\mA\)) works to give an upper bound on the sample complexity of using Gaussian vectors here, however that upper bound is exponentially loose.
In particular this show that while \(\mA\) is much easier to work with compared to \(\bar\mA\), it risks being an extremely loose bound on the actual variance on the Kronecker-Hutchinson Estimator.

\begin{reptheorem}{thm:real-rank-one-estimation}
Fix \(d,n\in\bbN\).
Let \(\mA=\vg\vg^\intercal\) where \(\vg\in\bbR^{d^k}\) is a vector of iid standard normal Gaussian.
Then if \(\cD'\) is either Rademacher or uniformly distributed on the sphere, on average across the choice \(\vg\), it suffices to take \(\ell=\frac{2}{\eps^2}\) for \(H_\ell(\mA)\) to have standard deviation at most \(\eps\tr(\mA)\).
That is, \(\E_{\vg}[\Var[H_\ell(\mA) \mid \vg]] \leq \eps^2 (\E_{\vg}[\tr(\mA)])^2\).
However, if \(\cD'\) is Gaussian, then \(\ell = \Omega(\frac1{\eps^2}(1+\frac2d)^k)\) is needed to achieve the same guarantee.
\end{reptheorem}

\begin{proof}
The proof works by flipping the roles of \vg and \vx via the tower rule:
\begin{align*}
    \E_{\vg}\left[\Var\left[H_\ell(\mA) \mid \vg\right]\right]
    &= \E_{\vg}\left[\Var_{\vx^{(j)}}\left[\frac1\ell \sum_{j=1}^\ell {\vx^{(j)}}^\intercal \mA \vx^{(j)} \mid \vg\right]\right] \\
    &= \frac1\ell \E_{\vg}\left[\Var_{\vx}\left[ \vx^\intercal \mA \vx \mid \vg\right]\right] \tag{Linearity of Variance}\\
    &= \frac1\ell \E_{\vg}\left[\E_{\vx}\left[\left( \vx^\intercal \mA \vx\right)^2 \mid \vg\right] - \left(\E_{\vx}\left[ \vx^\intercal \mA \vx \mid \vg\right]\right)^2 \right] \\
    &= \frac1\ell \E_{\vg}\left[\E_{\vx}\left[\left( \vx^\intercal\vg\right)^4 \mid \vg\right] - \left(\tr(\vg\vg^\intercal)\right)^2 \right] \\
    &= \frac1\ell \left(\E_{\vx}\left[\E_{\vg}\left[\left( \vx^\intercal\vg\right)^4 \mid \vx\right]\right] - \E_{\vg}\left[\norm{\vg}_2^4\right]\right) \tag{Tower Rule}\\
    &= \frac1\ell \left(\E_{\vx}\left[3\norm{\vx}_2^4\right] - d^{2k}\left(1+\frac{2}{d^k}\right)\right) \tag{\(\vx^\intercal\vg\sim\cN(0,\norm\vx_2^2), \norm{\vg}_2^2 \sim \chi_{d^k}^2\)}\\
    &= \frac1\ell \left( 3\left(\E_{\vx_1}\left[\norm{\vx_1}_2^4\right]\right)^k - d^{2k}\left(1+\frac{2}{d^k}\right)\right)
\end{align*}
We want this expression to be at most \(\eps^2 (\E[\tr(\mA)])^2\).
Since \(\tr(\mA)=\vg^\intercal\vg\) is a chi-squared random variable with parameter \(d^k\), we know that \(\E[\tr(\mA)] = d^k\).
So, rearranging the requirement that \(\E_{\vg}[\Var[H_\ell(\mA) \mid \vg]] \leq \eps^2 (\E[\tr(\mA)])^2\), we get that
\[
	\ell \geq \frac{1}{\eps^2}\left(3\left(\frac{\mu}{d^2}\right)^k - 1 - \frac2{d^k}\right)
\]
If \(\vx_1\) is a Rademacher vector or is uniformly drawn from the sphere, then we know that \(\norm{\vx_1}_2^2 = d\) deterministically, so that \(\mu = d^2\).
So, we get a complexity of
\[
	\ell \geq \frac{1}{\eps^2}\left(2 - \frac2{d^k}\right),
\]
so we see that \(\ell = \frac{2}{\eps^2}\) suffices.

However, if instead \vx is the Kronecker product of Gaussian vectors, then we know that \(\norm{\vx_1}_2^2\) is a $\chi^2$-distributed random variable with parameter \(d\).
So, \(\mu\) is the expected square of the chi-squared distribution, which is \(\mu \defeq \E[(\norm{\vx}_2^2)^2] = d^2(1+\frac{2}{d})\).
We then find that
\[
	\ell \geq \frac{1}{\eps^2}\left(3\left(1+\frac2d\right)^k - 1 - \frac2{d^k}\right)
	= \Omega(\frac{(1+\frac2d)^k}{\eps^2}),
\]
completing the proof.
\end{proof}

With this result in mind, where we get a rate of convergence of \((1+\frac2d)^k\) with a tight analysis, we can see how loose \Cref{thm:real-partial-trace} is with respect to \Cref{thm:real-partial-trace-transpose}.
That is, how much information do we loose when we bound \(\norm{\bar\mA}_F^2 \leq \norm{\mA}_F^2\).
For this random rank-1 case, this bound looses us an exponential factor.
From \Cref{thm:real-rank-one-estimation}, we see that \(\E_{\vg}[\Var[\vx^\intercal\mA\vx \mid \vg]] = d^{2k}(3(1+\frac2d)^k - 1 - \frac2{d^k})\).
In contrast, the upper bound from \Cref{thm:real-partial-trace} incurs a spurious factor of \(2^k\):
\begin{reptheorem}{thm:random-rank-one-real-var-overestimate}
Fix \(d,k\in\bbN\).
Let \(\mA=\vg\vg^\intercal\) where \(\vg\in\bbR^{d^k}\) is a vector of iid standard normal Gaussian.
Let \(\vx\) be the kronecker product of \(k\) standard normal Gaussian vectors.
Then,
\[
	\E_{\vg}\left[\sum_{\cS\subseteq[k]} 2^{k - \abs\cS} \norm{\tr_\cS(\mA)}_F^2\right]
	\geq 2^k d^{2k} (1+\tsfrac2d)^k.
\]
This inequality is tight up to a multiplicative factor of 3.
\end{reptheorem}
\begin{proof}
We start by analyzing the expected squared Frobenius norm of the partial trace of \mA with respect to a set of subsystems \(\cS \subseteq [k]\).
Let \(i \defeq \abs\cS\), and note that without loss of generality we can assume that \(\cS = \{1,\ldots,i\}\) when trying to compute \(\E[\norm{\tr_\cS(\mA)}_F^2]\).
We then decompose the vector \(\vg = \sbmat{\vg_1 \\ \vdots \\ \vg_{d^i}}\) where \(\vg_i \in \bbR^{d^{k-i}}\) are standard normal Gaussian vectors.
Then, note that we can decompose \(\vg = (\ve_1 \otimes \vg_1) + \ldots + (\ve_{d^i} \otimes \vg_{d^i})\).
This lets us write out
\[
	\vg^\intercal(\ve_j \otimes \mI^{\otimes k-i})
	= \sum_{\hat j = 1}^{d^i} (\ve_{\hat j} \otimes \vg_{\hat j})^\intercal (\ve_j \otimes \mI^{\otimes k-i})
	= \sum_{\hat j = 1}^{d^i} (\ve_{\hat j}^\intercal\ve_j \otimes \vg_{\hat j}^\intercal)
	= \vg_j^\intercal.
\]
Then, if we define \(\mG = \bmat{\vg_1 \, \cdots ~ \vg_{d^i}} \in \bbR^{d^{k-i} \times d^i}\), then we see that
\[
	\tr_{\cS}(\vg\vg^\intercal)
	= \sum_{j=1}^{d^i} (\ve_j \otimes \mI^{\otimes k-i})^\intercal\vg\vg^\intercal(\ve_j \otimes \mI^{\otimes k-i})
	= \sum_{j=1}^{d^i} \vg_j\vg_j^\intercal
	= \mG\mG^\intercal.
\]
Noticing that \mG is just a matrix full of iid standard normal Gaussians, we can then bound \(\norm{\tr_{:i}(\mA)}_F^2 = \norm{\mG\mG^\intercal}_F^2 = \norm{\mG}_4^4\) in expectation by just using known properties of the norms of Gaussians, where \(\norm{\cdot}_4\) denotes the Schatten 4-norm.
For instance, Lemma A.1 from \cite{tropp2023randomized} implies that a standard normal Gaussian matrix \(\mG\in\bbR^{N \times P}\) has \(\E[\norm{\mG}_4^4] = NP(N+P+1)\).
So,
\[
	\E[\norm{\tr_{\cS}(\mA)}_F^2]
	= \E[\norm{\mG}_4^4]
	= d^k (d^{k-\abs\cS} + d^{\abs\cS} + 1).
\]
This in turn implies that
\begin{align*}
	\E_{\vg}\left[\sum_{\cS\subseteq[k]} 2^{k - \abs\cS} \norm{\tr_\cS(\mA)}_F^2\right]
	&= \sum_{\cS\subseteq[k]} 2^{k - \abs\cS} d^k(d^{k-\abs\cS} + d^{\abs\cS} + 1) \\
	&= \sum_{i=0}^k \binom{k}{i} 2^{k - i} d^k(d^{k-i} + d^{i} + 1).
\end{align*}
The last line above recognizes that we only depended on \(\abs\cS\), so we could just sum over all values of \(i=\abs\cS\), scaled by the \(\binom{k}{i}\) times that \(\abs\cS\) appears.
Rearranging terms, and recalling that \(\sum_{i=0}^k \binom{k}{i} c^i = (1+c)^k\), we get
\begin{align*}
	\E_{\vg}\left[\sum_{\cS\subseteq[k]} 2^{k - \abs\cS} \norm{\tr_\cS(\mA)}_F^2\right]
	&= 2^kd^k \sum_{i=0}^k \binom{k}{i} 2^{-i} (d^{k-i} + d^{i} + 1) \\
	&= 2^kd^k \sum_{i=0}^k \binom{k}{i} \left(d^k (\tsfrac{1}{2d})^{i} + (\tsfrac{d}{2})^{i} + (\tsfrac{1}{2})^i\right) \\
	&= 2^kd^k \left( d^k (1+\tsfrac1{2d})^k + (1+\tsfrac{d}{2})^k + (\tsfrac32)^k \right) \\
	&= 2^kd^{2k} \left( (1+\tsfrac1{2d})^k + (\tsfrac{1}{2}+\tsfrac1d)^k + (\tsfrac3{2d})^k \right) \\
	&\geq 2^k d^{2k} (1+\tsfrac1{2d})^k
\end{align*}
To see that this bound is tight up to a multiplicative factor of 3, we can take the only inequality in our analysis and instead upper bound
\[
	\left( (1+\tsfrac1{2d})^k + (\tsfrac{1}{2}+\tsfrac1d)^k + (\tsfrac3{2d})^k \right)
	\leq 3(1+\tsfrac1{2d})^k.
\]
\end{proof}

%% file: complex-oracle.tex
% Clean
\section{The Complex-Valued Kronecker-Matrix-Vector Oracle}
\label{sec:complex-matvec}

Until now, we have only considered using a \emph{real-valued} Kronecker-matrix-vector oracle, where we are only allowed to compute \(\mA\vx\) for vectors \(\vx = \vx_1 \otimes \cdots \otimes \vx_k\) where \(\vx_i \in \bbR^{d}\).
In this section, we consider how our results change when we instead allow \(\vx_i \in \bbC^d\).
If we did not have any Kronecker constraint on our oracle, then this change of oracles would be banal.
If some algorithm computes \(q\) complex-valued non-Kronecker matrix-vector products with \mA, then we can find another algorithm which produces the exact same result, but instead uses \(2q\) real-valued non-Kronecker matrix-vector products with \mA.
Each complex-values query gets decomposed as \(\mA\vx = \mA\mathfrak{Re}[\vx] + i\mA\mathfrak{Im}[\vx]\), which can be evaluated with just two real-valued oracle queries.

However, this picture changes with a Kronecker constraint.
If \(\vx = \vx_1 \otimes \cdots \otimes \vx_k\) for \(\vx_i \in \bbC^d\), then the real part of the vector \(\mathfrak{Re}[\vx]\) does not necessarily have Kronecker structure as well, and therefore we cannot reproduce the strategy above.
We can see this even with very small vectors:
\[
	\bmat{1 \\ i} \otimes \bmat{1 \\ i} = \bmat{1 \\ i \\ i \\ -1}
	\hspace{0.7cm} \text{which gives} \hspace{0.7cm}
	\mathfrak{Re}\left(\bmat{1 \\ i \\ i \\ -1}\right) = \bmat{1 \\ 0 \\ 0 \\ -1}
	\hspace{0.5cm} \text{and} \hspace{0.5cm}
	\mathfrak{Im}\left(\bmat{1 \\ i \\ i \\ -1}\right) = \bmat{0 \\ 1 \\ 1 \\ 0}.
\]
The real and imaginary component vectors do not have Kronecker structure.
We can still simulate a complex-valued oracle with a real-valued oracle, but this requires far more real-valued oracle queries in the worst case.
Specifically, we expand the real and imaginary parts of each vectors \(\vx_i\).
To keep the notation simple, let \(\vx_i = \vr_i + i\vm_i\) decompose the real and imaginary parts of \(\vx_i \in \bbC^d\), so that \(\vr_i,\vm_i\in\bbR^d\).
Then,
\begin{align*}
	\vx_1 \otimes \vx_2 \otimes \cdots \otimes \vx_k
	&= (\vr_1 + i\vm_1) \otimes \vx_2 \otimes \cdots \vx_k \\
	&= (\vr_1 \otimes \vx_2 \otimes \cdots \otimes \vx_k) + i(\vm_1 \otimes \vx_2 \otimes \cdots \otimes \vx_k) \\
	&= (\vr_1 \otimes \vr_2 \otimes \cdots \otimes \vx_k) + i(\vr_1 \otimes \vm_2 \otimes \cdots \otimes \vx_k) \\
	&\hspace{1cm} + i(\vm_1 \otimes \vr_2 \otimes \cdots \otimes \vx_k) - (\vm_1 \otimes \vm_2 \otimes \cdots \otimes \vx_k) \\
	&= \ldots
\end{align*}
where we continue to expand this out into a sum of \(2^k\) real-valued Kronecker-structured vectors.
So, if some algorithm solves a problem using \(q\) complex-valued Kronecker-structure matrix-vector products with \mA, then we can only immediately guarantee the existence of an algorithm that computes \(2^k q\) real-valued Kronecker-structured products with \mA that also solves the problem.

Keeping this exponential gap between the real-valued and complex-valued oracles in mind, we now turn our attention back to the Kronecker-Hutchinson Estimator.
In the non-Kronecker case, it is well known that the Hutchinson Estimator performs well with complex queries:
\begin{lemma}[\cite{martinsson2020randomized,epperlyStochastic}]
\label{lem:hutch-complex-variance}
	Let \(\mA \in \bbR^{D \times D}\) and let \(\vg \in \bbC^D\) be a vector of either iid complex Rademacher or complex standard normal Gaussian entries.
	That is, we have \(\vg = \frac1{\sqrt2}\vr + \frac{i}{\sqrt2}\vm\) where \vr and \vm are iid real-valued vectors of either Rademacher or standard normal Gaussian entries.
	Then, \(\E[\vg^\herm\mA\vg] = \tr(\mA)\) and \(\Var[\vg^\herm\mA\vg] \leq \norm{\frac{\mA+\mA^\intercal}{2}}_F^2 \leq \norm{\mA}_F^2\), where the first inequality is tight if \vg is complex Gaussian, and the second is tight if \mA is symmetric.
\end{lemma}
Note that the difference between \Cref{lem:hutch-classical-variance} and \Cref{lem:hutch-complex-variance} is just a factor of 2 in the variance term.
However, this factor of 2 has been extremely influential in all of our proofs so far.
Improving this constant to be 1 and following the exact same proofs give the following results:
\begin{reptheorem}{thm:complex-upper-bounds}
	Fix \(d,k\in\bbN\), and let \(\mA\in\bbR^{d^k \times d^k}\).
	Let \(\vx = \vx_1 \otimes \cdots \otimes \vx_k\), where each \(\vx_i\) is either a complex Rademacher vector or a complex Gaussian vector.
	That is, either we have \(\vx_i = \frac1{\sqrt2}\vr_i + \frac{i}{\sqrt2}\vm_i\) where \(\vr_i\) and \(\vm_i\) are \(\cN(\vec0,\mI)\) vectors, or each entry of \(\vx_i\) is drawn uniformly iid from \(\{\pm1,\pm i\}\).
	Let \(\bar\mA \defeq \frac1{2^k} \sum_{\cV \subseteq [k]} \mA^{\ptran\cV}\) be the average of all partial transposes of \mA.
	Then,
	\(\Var[\vx^\herm\mA\vx] \leq \sum_{\cS\subset[k]} \norm{\tr_{\cS}(\bar\mA)}_F^2 \leq \sum_{\cS\subset[k]} \norm{\tr_{\cS}(\mA)}_F^2\).
	If \(\vx_i\) are complex Gaussians, then the first expression is an equality.
	If we further know that \(\mA = \bar\mA\), then both expression are equalities.
	Further, if \(\vu = \vu_1 \otimes \cdots \otimes \vu_k\) where \(\vu_i\in\bbC^d\) are uniformly random vectors with \(\norm{\vu_i}_2=\sqrt d\), then \(\Var[\vu^\herm\mA\vu] = \frac{\Var[\vx^\herm\mA\vx]-((1+\frac1d)^k-1)(\tr(\mA))^2}{(1+\frac1d)^k} \leq \frac{1}{(1+\frac1d)^k} \Var[\vx^\herm\mA\vx]\).
\end{reptheorem}
\begin{proof}
	Retrace the proofs of \Cref{lem:frob-inductive-step,lem:expected-pmrdm-frob} as well as \Cref{thm:real-partial-trace-transpose,thm:real-psd,corol:random-unit-vec-variance}.
	Anywhere that a 2 appears not in an exponent, replace it with a 1.
\end{proof}
We can directly see from \Cref{lem:hutch-complex-variance} that we can take \(C=1\) in \Cref{thm:coarse-psd-rate} when \(\cD'\) is complex Gaussian.
Further, using the fact that the variance of a non-Kronecker Hutchinson sample using a complex Rademacher is \(\Var[\vx_i^\herm\mB\vx_i] = \norm{\mB}_F^2 - \norm{\text{diag}(\mB)}_2^2\), we can use the proof technique in \Cref{corol:rademacher-variance-bound} to take \(C = 1-\frac1d\) in \Cref{thm:coarse-psd-rate}.
Similarly, using that complex unit vectors have \(\Var[\vx_i^\herm\mB\vu_i] = \frac{d}{d+1}(\norm{\mB}_F^2 - (\tr(\mB))^2)\), we get \(C = 1-\frac{2}{d+1}\) in \Cref{thm:coarse-psd-rate}.
We conclude the following rates:
\begin{lemma}
If \(\cD'\) is complex Gaussian, then \(\Var[\vx^\intercal\mA\vx] \leq 2^k (\tr(\mA))^2\).
If \(\cD'\) is complex Rademacher, then \(\Var[\vx^\intercal\mA\vx] \leq (2-\frac1d)^k (\tr(\mA))^2\).
If \(\cD'\) is uniformly from the complex sphere, then \(\Var[\vx^\intercal\mA\vx] \leq (2-\frac{2}{d+1})^k (\tr(\mA))^2\).
\end{lemma}
We can similarly show that this dependence of \((2-\frac1d)^k\) is tight amongst isotropic query vectors with iid entries:
\begin{reptheorem}{thm:complex-psd-lower-bound}
	Fix \(d,k\in\bbN\).
	Consider Hutchinson's Estimator run with vectors \(\vx = \vx_1 \otimes \cdots \otimes \vx_k\) where \(\vx_i\) are vectors of iid complex zero-mean unit-variance entries.
	Then, there exists a PSD matrix \mA such that Hutchinson's Estimator needs \(\ell = \Omega(\frac{(2-\tsfrac1d)^k}{\eps^2})\) samples in order to have standard deviation \(\leq \eps \tr(\mA)\).
\end{reptheorem}
\begin{proof}
	Like in the proof of \Cref{thm:real-psd-lower-bound}, we take \(\mA = \otimes_{i=1}^k \ve\ve^\intercal\) where \(\ve\in\bbR^d\) is the all-ones vector.
	In fact, we follow the proof of \Cref{thm:real-psd-lower-bound} rather closely, but we recreate much of it to handle the complex conjugation correctly.
	First, we expand the product
	\[
		\E[(\vx^\herm\mA\vx)^2]
		= \prod_{i=1}^k \E[(\vx_i^\herm\ve\ve^\herm\vx_i)^2]
		= \left(\E[(\ve^\herm\vx_1)^2\overline{(\ve^\herm\vx_1)^2}]\right)^k.
	\]
	We let \(\vy \defeq \vx_1\) to simplify the notation.
	Then, we expand this inner expectation:
	\[
		\textstyle
		\E[(\ve^\herm\vy)^2\overline{(\ve^\herm\vy)^2}]
		= \E[(\sum_{i,j=1}^d y_iy_j)(\sum_{m,n=1}^d \overline{y_my_n})]
		= \sum_{i,j,m,n=1}^d \E[y_iy_j\overline{y_my_n}].
	\]
	This expectation is zero if any random variable appears exactly once, like how \(\E[y_1y_2\overline{y_2^2}] = \E[y_1]\E[y_2\overline{y_2^2}] = 0\) by independence.
	So, we only need to consider terms in this sum where we do not have any one term appear alone.
	This can happen in three different ways.

	First, if \(i=j=m=n\), then by Jensen's inequality, we get \(\E[y_iy_j\overline{y_my_n}] = \E[|y_i|^4] \geq (\E[|y_i|^2])^2 = 1\).
	This case occurs \(d\) times in the sum.
	Second, if \(i=j\neq m=n\), then by the iid distribution of the entry of \vy, we get \(\E[y_iy_j\overline{y_my_n}] = \E[y_i^2]\overline{\E[y_m^2]} = |\E[y_i^2]|^2 \geq 0\).
	Lastly, if \(i=m \neq j=n\) or \(i=n \neq j=m\), then \(\E[y_iy_j\overline{y_my_n}] = \E[|y_i|^2]\E[|y_j|^2] = 1\).
	This case occurs \(2(d^2-d)\) times in the sum.
	So, we get that
	\[
		\E[(\vx^\herm\mA\vx)^2] \geq \left( 2d^2 - 2d + d \right)^k = (2d^2 - d)^k.
	\]
	Which, noting that \(\tr(\mA)=d^k\), in turn gets a variance lower bound of
	\[
		\Var[\vx^\herm\mA\vx]
		= \E[(\vx^\herm\mA\vx)^2] - (\tr(\mA))^2
		\geq (2d^2 - d)^k - d^{2k}.
	\]
	Then, since \(\Var[H_\ell(\mA)] = \frac1\ell \Var[\vx^\herm\mA\vx]\), in order for the complex Kronecker-Hutchinson's Estimator to have standard deviation \(\eps\tr(\mA)\), we need to take
	\[
	    \ell
	    \geq \frac{\Var[\vx^\herm\mA\vx]}{\eps^2(\tr(\mA))^2}
	    \geq \frac{(2d^2 - d)^k - d^{2k}}{\eps^2 d^{2k}}
	    = \frac{(2-\frac1d)^k - 1}{\eps^2}
	    = \Omega\left(\frac{(2-\tsfrac1d)^k}{\eps^2}\right).
	\]
\end{proof}
Next, we analyze the exact variance of the Kronecker-Hutchinson estimator on random rank-1 matrices:
\begin{reptheorem}{thm:complex-rank-one-estimation}
Fix \(d,n\in\bbN\).
	Let \(\mA=\vg\vg^\intercal\) where \(\vg\in\bbR^{d^k}\) is a vector of iid standard normal Gaussian.
	Then if \(\cD'\) is either complex Rademacher or uniformly distributed on the complex sphere, on average across the choice \(\vg\), it suffices to take \(\ell=\frac{2}{\eps^2}\) for \(H_\ell(\mA)\) to have standard deviation at most \(\eps\tr(\mA)\).
	That is, \(\E_{\vg}[\Var[H_\ell(\mA) \mid \vg]] \leq \eps^2 (\E_{\vg}[\tr(\mA)])^2\).
	However, if \(\cD'\) is complex Gaussian, then \(\ell = \Omega(\frac1{\eps^2}(1+\frac1d)^k)\) is needed to achieve the same guarantee.
\end{reptheorem}
\begin{proof}
Follow the proof of \Cref{thm:real-rank-one-estimation}, but take the value of \(\mu=d\) for complex Rademacher and uniform sample from the complex sphere, and take \(\mu = d^2(1+\frac1d)\) for complex Gaussian vectors.
\end{proof}

Lastly, we note that the \(\norm{\bar\mA}_F^2 \leq \norm{\mA}_F^2\) bound produces a much tighter bound on the sample complexity of the Kronecker-Hutchinson's Estimator when working with random rank-one matrices:
\begin{reptheorem}{thm:random-rank-one-complex-var-overestimate}
Fix \(d,k\in\bbN\).
Let \(\mA=\vg\vg^\intercal\) where \(\vg\in\bbR^{d^k}\) is a vector of iid standard normal Gaussian.
Let \(\vx\) be the kronecker product of \(k\) complex standard normal Gaussian vectors.
Then,
\[
	\E_{\vg}\left[\sum_{\cS\subseteq[k]} \norm{\tr_\cS(\mA)}_F^2\right]
	\leq 3 d^{2k} (1+\tsfrac1d)^k.
\]
This inequality is tight up to a multiplicative factor of 3.
\end{reptheorem}
\begin{proof}
Recall from the proof of \Cref{thm:random-rank-one-real-var-overestimate} that the expected squared Frobenius norm of the partial trace of \mA is
\[
	\E[\norm{\tr_{\cS}(\mA)}_F^2]
	= \E[\norm{\mG}_4^4]
	= d^k (d^{k-\abs\cS} + d^{\abs\cS} + 1).
\]
Then, we can directly bound
\begin{align*}
	\sum_{\cS\subseteq[k]} \E[\norm{\tr_\cS(\mA)}_F^2]
	&= \sum_{\cS\subseteq[k]} d^k(d^{k-\abs\cS} + d^{\abs\cS} + 1) \\
	&= \sum_{i=0}^k \binom{k}{i} d^k(d^{k-i} + d^{i} + 1).
\end{align*}
Where the last line recognizes that we only depended on \(\abs\cS\), so we could just sum over all values of \(i=\abs\cS\), scaled by the \(\binom{k}{i}\) times that \(\abs\cS\) appears.
Rearranging terms, and recalling that \(\sum_{i=0}^k \binom{k}{i} c^i = (1+c)^k\), we get
\begin{align*}
	\sum_{\cS\subseteq[k]} \E[\norm{\tr_\cS(\mA)}_F^2]
	&= d^k \sum_{i=0}^k \binom{k}{i} (d^{k-i} + d^{i} + 1) \\
	&= d^k \sum_{i=0}^k \binom{k}{i} \left(d^k (\tsfrac{1}{d})^{i} + d^{i} + 1^i\right) \\
	&= d^k \left( d^k (1+\tsfrac1d)^k + (1+d)^k + 2^k \right) \\
	&= d^{2k} \left( 2(1+\tsfrac1d)^k + (\tsfrac2d)^k \right) \\
	&\leq 3 d^{2k} (1+\tsfrac1d)^k.
\end{align*}
To see that this bound is tight up to a multiplicative factor of 3, we can take the only inequality in our analysis and instead lower bound
\[
	\left( 2(1+\tsfrac1d)^k + (\tsfrac2d)^k \right)
	\geq (1+\tsfrac1{d})^k.
\]
\end{proof}
The proof technique of \Cref{thm:real-rank-one-estimation} shows that the true variance indeed scales as \(d^{2k}(1+\frac2d)^k\) in the real Gaussian case, and \(d^{2k}(1+\frac1d)^k\) in the complex Gaussian case.
The variance upper bound using \(\norm{\bar\mA}_F^2 \leq \norm{\mA}_F^2\) in the real-valued case was \(2^k d^{2k}(1+\frac2d)^k\), which was loose by a factor of \(2^k\).
In contrast, that same variance upper bound here in the complex-valued case is tight, achieving \(d^{2k}(1+\frac1d)^k\).
Overall, it remains unclear exactly when taking the approximation \(\norm{\bar\mA}_F^2 \leq \norm{\mA}_F^2\) is lossy.
Here, taking the exact same choice of \mA matrix but changing the query vectors from being real Gaussian to complex Gaussian allowed this bound to become asympotically tight.

%% file: matvec-lower-bound.tex
% Clean

\section{Looking Beyond the Kronecker-Hutchinson Estimator}
\label{sec:matvec-lower-bound}

We have shown that the Kronecker-Hutchinson Estimator has a painful exponential dependence on the \(k\) in the worst case.
However, we are broadly interested in the class of all algorithms that operate in the Kronecker-matrix-vector oracle model.
We show that all trace estimation algorithms in this model must have query complexity that grows with \(k\):
\begin{theorem}
\label{thm:matvec-lower-bound}
	Any algorithm that accesses a PSD matrix \mA via Kronecker-matrix-vector products \(\mA\vx^{(1)},\ldots,\mA\vx^{(q)}\), where \(\vx^{(1)},\ldots,\vx^{(q)}\) are (possibly adaptively) chosen real-valued vectors, requires \(q = \Omega(\frac{\sqrt k}{\eps})\) in order to output an estimate \(\hat t\) such that \(\sqrt{\E[(\hat t - \tr(\mA))^2]} \leq \eps \tr(\mA)\).
\end{theorem}
If we constrain ourselves to estimators \(\hat t\) that are unbiased (i.e. \(\E[\hat t] = \tr(\mA)\)), then this lower bound says that \(q = \Omega(\frac{\sqrt k}{\eps})\) Kronecker-matrix-vector products are required to achieve standard deviation \(\leq \eps\tr(\mA)\).
Our analysis builds off of tools in \cite{braverman2020gradient,jiang2021optimal}, which studies matrix-vector lower bounds without the Kronecker constraint.
In particular, we use the following theorem:

\begin{definition}
	Let \(\mG \in \bbR^{d \times d}\) be a matrix of iid standard normal Gaussian entries.
	Then, the matrix \(\mW \defeq \mG^\intercal\mG \in \bbR^{d \times d}\) is distributed as a \emph{\(\text{Wishart}(d)\)} random matrix.
\end{definition}

\begin{importedtheorem}[Lemma 13 from \cite{braverman2020gradient}]
\label{impthm:hidden-wishart}
	Let \(\mW \sim \emph{\text{Wishart(d)}}\).
	Then, for any sequence of (possibly adaptively chosen) matrix-vector queries \(\vx^{(1)},\ldots,\vx^{(q)}\) and responses \(\vz^{(1)},\ldots,\vz^{(q)}\) (so that \(\vz^{(j)} = \mW\vx^{(j)}\)), there exists a rotation matrix \(\mV \in \bbR^{d \times d}\) constructed as a function of \(\vx^{(1)},\ldots,\vx^{(q)}\) such that the matrix \(\mV\mW\mV^\intercal\) can be written as
	\[
		\mV\mW\mV^\intercal = \mDelta + \bmat{\mat0 & \mat0 \\ \mat0 & \widetilde\mW},
	\]
	where \(\widetilde\mW \in \bbR^{d-q \times d-q}\) conditioned on the values of \(\vx^{(1)},\ldots,\vx^{(q)},\vz^{(1)},\ldots,\vz^{(q)}\) is distributed as Wishart matrix.
	That is, \(\widetilde\mW ~|~ (\vx^{(1)},\ldots,\vx^{(q)},\vz^{(1)},\ldots,\vz^{(q)}) \sim \emph{\text{Wishart}}(d-q)\).
	Further, \mDelta is a PSD matrix that is a deterministic function of \(\vx^{(1)},\ldots,\vx^{(q)},\vz^{(1)},\ldots,\vz^{(q)}\).
\end{importedtheorem}
This theorem is powerful because it says that any algorithm that has accessed a Wishart matrix via \(q\) matrix-vector products must have a submatrix which the algorithm has absolutely no knowledge of.
We note that in \cite{braverman2020gradient}, they do not state that \(\mDelta\) is a deterministic function of \(\vx^{(1)},\ldots,\vx^{(q)},\vz^{(1)},\ldots,\vz^{(q)}\), but this fact is easily verified by examining by the proof of Lemma 13 in \cite{braverman2020gradient}.

We start by extending \Cref{impthm:hidden-wishart} to the Kronecker-matrix-vector oracle model.
We construct our hard instance input matrix as \(\mA = \otimes_{i=1}^k \mW_i \in \bbR^{d^k \times d^k}\), where \(\mW_1,\ldots,\mW_k \in \bbR^{d \times d}\) are iid Wishart matrices:

\begin{corollary}
\label{corol:hidden-kron-wishart}
	Let \(\mW_1,\ldots,\mW_k\) be iid \(\emph{\text{Wishart}}(d)\) matrices, and let \(\mA \defeq \otimes_{i=1}^k \mW_i \in \bbR^{d^k \times d^k}\).
	Then, for any sequence of (possibly adaptively chosen) Kronecker-matrix-vector queries \(\vx^{(1)},\ldots,\vx^{(q)}\) and responses \(\vz^{(1)},\ldots,\vz^{(q)}\) (so that \(\vz^{(j)} = \mA\vx^{(j)}\)), there exists orthogonal matrices \(\mV_1,\ldots,\mV_k \in \bbR^{d \times d}\) such that \(\mV = \otimes_{i=1}^k \mV_i\) has
	\[
		\mV\mA\mV^\intercal = \otimes_{i=1}^k \left( \mDelta_i + \bmat{\mat 0 & \mat 0 \\ \mat 0 & \widetilde\mW_i}\right).
	\]
	Let \(\vx^{(j)} = \otimes_{i=1}^k \vx_i^{(j)}\) and \(\vz_i^{(j)} = \mW_i\vx_i^{(j)}\) be decompositions of the query and response vectors.
	Then, each \(\widetilde\mW_i \in \bbR^{d-q \times d-q}\) conditioned on the values of \(\vx_i^{(1)},\ldots,\vx_i^{(q)},\vz_i^{(1)},\ldots,\vz_i^{(q)}\) is distributed as a Wishart matrix.
	That is, \(\widetilde\mW_i ~|~ (\vx_i^{(1)},\ldots,\vx_i^{(q)},\vz_i^{(1)},\ldots,\vz_i^{(q)}) \sim \emph{\text{Wishart}}(d-q)\).
	Further, each \(\mDelta_i\) is a PSD matrix that is a deterministic function of \(\vx_i^{(1)},\ldots,\vx_i^{(q)},\vz_i^{(1)},\ldots,\vz_i^{(q)}\).
\end{corollary}
With this setup, we can now form a lower bound against any trace estimating algorithm.
We first prove in \Cref{lem:lower-bound-use-mmse} that, in terms of minimizing mean squared error, no estimator can outperform the conditional expectation \(\tilde t = \E[\tr(\mA) ~|~ \mV_1,\ldots,\mV_k,\mDelta_1,\ldots,\mDelta_k]\).
Next, in \Cref{lem:lower-bound-exact-mse}, we compute the exact variance of this estimator.
Lastly, in the final proof of \Cref{thm:matvec-lower-bound}, we compare this variance to the trace of \mA, to conclude that \(q \geq \Omega(\frac{\sqrt k}{\eps})\).

To simplify notation, we let \(\cW = (\widetilde\mW_1,\ldots,\widetilde\mW_k)\) denote the list of hidden Wishart matrices, and let \(\cV \defeq (\mV_1,\ldots,\mV_k,\mDelta_1,\ldots,\mDelta_k)\) be the variables that \(\tilde t\) is conditioned on.
Note that the algorithm only observes the response vectors \(\vz^{(1)},\ldots,\vz^{(q)}\), and that the algorithm does not exactly know the vectors \(\vz_i^{(j)} = \mW_i \vx_i^{(j)}\).
Therefore, the algorithm might not observe all of \cV, but \cV definitely covers everything that the algorithm does see.

\begin{lemma}
	\label{lem:lower-bound-use-mmse}
	Consider an algorithm that interacts with \mA via \(q\) Kronecker-matrix-vector products as in the setting of \Cref{corol:hidden-kron-wishart}.
	Let \(\hat t\) be the output of that algorithm, and let \(\tilde t \defeq \E[\tr(\mA) ~|~ \cV]\).
	Then, \(\tilde t\) has better means squared error than \(\hat t\):
	\[
		\E[(\hat t - \tr(\mA))^2] \geq \E[(\tilde t - \tr(\mA))^2].
	\]
\end{lemma}
\begin{proof}
Note that the algorithm only observes \(\vx^{(1)},\ldots,\vx^{(q)},\vz^{(1)},\ldots\vz^{(q)}\), so the estimator \(\hat t\) has to be a (possibly randomized) function of \cV.
We then expand the mean squared error of \(\hat t\) by using the tower rule:
\begin{align*}
	\E[(\hat t - \tr(\mA))^2]
	&= \E[ \E_{\cW} [ (\hat t - \tr(\mA))^2 ~|~ \cV ] ] \\
	&\geq \E[ \E_{\cW} [ (\tilde t - \tr(\mA))^2 ~|~ \cV ] ] \\
	&= \E[(\tilde t - \tr(\mA))^2],
\end{align*}
where the inequality comes from \(\tilde t\) being the minimum mean square error estimator for that particular conditional expectation.
\end{proof}

For any specific instantiation of \cV, we can ask what the expected mean squared error of \(\tilde t\) is:
\[
	\E_{\cW}[(\tr(\mA) - \tilde t)^2 ~|~ \cV]
	= \E_{\cW}[(\tr(\mA) - \E[\tr(\mA)|\cV])^2 ~|~ \cV]
	= \Var_{\cW}[\tr(\mA) ~|~ \cV].
\]
We do not really care about this value for a specific choice of \cV, but instead care about it on average across all values of \cV.
That is, we care about the squared error \(\E[(\tilde t - \tr(\mA))^2] = \E_{\cV}[\Var[\tr(\mA) ~|~ \cV]]\).
We then want to ensure that the error is less than \(\eps\E[\tr(\mA)]\).
So, we want to find out what value of \(q\) ensures that
\[
	\sqrt{\E[\Var[\tr(\mA)~|~\cV]]} \leq \eps\E[\tr(\mA)].
\]
We start by computing the variance inside the square root on the left.

\begin{lemma}
\label{lem:lower-bound-exact-mse}
The mean squared error of \(\tilde t\) is \(\E[\Var[\tr(\mA) ~|~ \cV]] = (d^4 + 2d^2)^k - (d^4 + 2d^2 - 2(d-q)^2)^k\).
\end{lemma}
\begin{proof}
Let \(X_i \defeq \tr(\mW_i)\) and \(X \defeq \tr(\mA)\), so that \(X = \prod_{i=1}^k X_i\) is a product of \(k\) independent random variables.
Then, since these terms are independence, we know that
\[
	\Var[{\textstyle \prod_{i=1}^k X_i}]
	= \prod_{i=1}^k \E[X_i^2] - \prod_{i=1}^k (\E[X_i])^2.
\]
Further, since \(X_i | \cV\) is independent of \(X_j | \cV\) for \(i\neq j\), we know that
\begin{align*}
	\E[\Var[{\textstyle \prod_{i=1}^k X_i}~|~\cV]]
	&= \E\left[\prod_{i=1}^k \E[X_i^2|\cV] - \prod_{i=1}^k (\E[X_i|\cV])^2\right] \\
	&= \prod_{i=1}^k \E[X_i^2] - \prod_{i=1}^k \E[(\E[X_i|\cV])^2].
\end{align*}
The last expectation on the right can be simplified a bit further.
Since \(\Var[X_i ~|~\cV] = \E[X_i^2 ~|~ \cV] - (\E[X_i ~|~ \cV])^2\), we can write
\[
	\E[(\E[X_i|\cV])^2]
	= \E[ \E[X_i^2|\cV] - \Var[X_i|\cV] ]
	= \E[X_i^2] - \E[\Var[X_i|\cV]],
\]
which in turn means we have
\begin{align}
	\E[\Var[{\textstyle \prod_{i=1}^k X_i}~|~\cV]]
	= \prod_{i=1}^k \E[X_i^2] - \prod_{i=1}^k \left(\E[X_i^2] - \E[\Var[X_i|\cV]]\right).
	\label{eq:lower-bound-expand-variances}
\end{align}
Now we can directly analyze these means and expectations.
First, note that \(X_i = \tr(\mW_i) = \tr(\mG_i^\intercal\mG_i) = \norm{\mG_i}_F^2 \sim \chi_{d^2}^2\) is a chi-squared random variable.
Therefore, \(\E[X_i^2] = \Var[X_i] + (\E[X_i^2])^2 = 2d^2 + d^4\).
Next, we focus on \(\E[\Var[X_i | \cV]]\).
By the linearity and cyclic property of the trace, and since \(\mV_i\mV_i^\intercal=\mI\), we have that
\[
	X_i
	= \tr(\mW_i)
	= \tr(\mW_i \mV_i\mV_i^\intercal)
	= \tr(\mV_i^\intercal \mW_i \mV_i)
	= \tr(\mDelta_i) + \tr(\widetilde\mW_i).
\]
Therefore the only term in \cV that \(X_i\) depends on is \(\mDelta_i\), and so we can expand
\begin{align*}
	\E[\Var[X_i | \cV]]
	= \E[ \Var[\tr(\mDelta_i) + \tr(\widetilde\mW_i) ~|~ \mDelta_i] ]
	= \E[ \Var[\tr(\widetilde\mW_i)] ]
	= \Var[\tr(\widetilde\mW_i)] = 2(d-q)^2,
\end{align*}
where the last equality notes that \(\tr(\widetilde\mW_i) = \tr(\widetilde\mG_i^\intercal\widetilde\mG_i) = \norm{\widetilde\mG_i}_F^2 \sim \chi_{(d-q)^2}\) is also a chi-squared random variable.
Plugging our expectations back into \Cref{eq:lower-bound-expand-variances}, we find that
\begin{align*}
	\E[\Var[{\textstyle \prod_{i=1}^k X_i}~|~\cV]]
	&= \left(d^4 + 2d^2\right)^k - \left(d^4 + 2d^2 - 2(d-q)^2\right)^k,
\end{align*}
which completes the proof.
\end{proof}

We now have built up the tools we need to prove the central lower bound:
\begin{reptheorem}{thm:matvec-lower-bound}
	Any algorithm that accesses a PSD matrix \mA via Kronecker-matrix-vector products \(\mA\vx^{(1)},\ldots,\mA\vx^{(q)}\), where \(\vx^{(1)},\ldots,\vx^{(q)}\) are (possibly adaptively) chosen real-valued vectors, requires \(q = \Omega(\frac{\sqrt k}{\eps})\) in order to output an estimate \(\hat t\) such that \(\sqrt{\E[(\hat t - \tr(\mA))^2]} \leq \eps \tr(\mA)\).
\end{reptheorem}
\begin{proof}
	We consider \mA as in \Cref{corol:hidden-kron-wishart}.
	By \Cref{lem:lower-bound-use-mmse,lem:lower-bound-exact-mse}, we know that
	\[
		\E[(\hat t - \tr(\mA))^2] \geq \E[(\tilde t - \tr(\mA))^2] = (d^4 + 2d^2)^k - (d^4 + 2d^2 - 2(d-q)^2)^k.
	\]
	We then want to show that this term on the right hand side is at most \(\eps^2(\E[\tr(\mA)])^2\).
	First, we quickly note that \(\E[\tr(\mA)] = \prod_{i=1}^k \E[\tr(\mW_i)] = d^{2k}\), since \(\tr(\mW_i) = \tr(\mG_i^\intercal\mG_i) = \norm{\mG_i}_F^2 \sim \chi_{d^2}^2\) is a chi-squared random variable.
	So, we want to find what values of \(q\) ensure that
	\[
		(d^4 + 2d^2)^k - (d^4 + 2d^2 - 2(d-q)^2)^k \leq \eps^2 d^{4k}.
	\]
	In fact, it will be mathematically convenient to demand slightly less accuracy, which will still give a meaningful lower bound:
	\[
		(d^4 + 2d^2)^k - (d^4 + 2d^2 - 2(d-q)^2)^k \leq \eps^2 (d^4+2d^2)^k.
	\]
	We can rearrange this inequality to produce
	\[
		(d-q)^2 \leq \tsfrac12(d^4 + 2d^2) (1-(1-\eps^2)^{1/k}).
	\]
	Note that for \(\eps\in(0,\frac12)\) and \(k \geq 1\), we have that \((1-\frac{2\eps^2}{k})^k = (1-\frac{2\eps^2}{k})^{\frac{k}{2\eps^2}\cdot2\eps^2} \leq e^{-2\eps^2} \leq 1-\eps^2\), so that \((1-\eps^2)^{1/k} \geq 1-\frac{2\eps^2}{k}\).
	So, by additionally bounding \(d^4 + 2d^2 \leq 2d^4\) for \(d\geq2\), we can simplify the above bound to be
	\[
		(d-q)^2 \leq \tsfrac12(d^4 + 2d^2) (1-(1-\eps^2)^{1/k}) \leq \tsfrac{2d^4\eps^2}{k}.
	\]
	Solving this inequality for \(q\) gives \(q \geq d - \frac{\sqrt{2}\eps}{\sqrt k} d^2\).
	Optimizing over \(d\) gives \(q \geq \frac{4 \sqrt k}{\sqrt 2 \eps} = \Omega(\frac{\sqrt k}{\eps})\), completing the proof.
\end{proof}

%% file: kronecker-matrix-recovery.tex
% Clean

\section{
	\texorpdfstring{Exactly Computing the Trace of \(\mA = \otimes_{i=1}^k\mA_i\) in \(kd+1\) Queries}{Exactly Computing the Trace of Kronecker Matrices with kd+1 Queries}
}
\label{sec:kron-recovery}

In this section, we are promised that \(\mA = \mA_1 \otimes \cdots \otimes \mA_k\) for some unknown matrices \(\mA_1,\ldots,\mA_k\in\bbR^{d \times d}\).
We show that by using exactly \(kd+1\) Kronecker-matrix-vector products, we can recover a set of matrices \(\mB_1,\ldots,\mB_k \in \bbR^{d \times d}\) such that \(\otimes_{i=1}^k \mB_i = \mA\).
Note that we cannot exactly recover the matrices \(\mA_1,\ldots,\mA_k\) because the Kronecker product does not admit a unique decomposition: we have \(\alpha\mA \otimes \mB = \mA \otimes \alpha\mB\) for all scalars \(\alpha\) and matrices \mA, \mB.
However, we are able to recover the trace exactly since  \(\tr(\mA_1 \otimes \mA_2) = \tr(\mA_1) \tr(\mA_2)\), we can then compute \(\tr(\mA) = \prod_{i=1}^k \tr(\mB_i)\).

\begin{algorithm}[ht]
	\caption[Exact Kronecker Recovery]{Exact Kronecker Recovery}
	\label{alg:kron-recovery}
	{\bfseries input}: Kronecker-matrix-vector oracle access to \(\mA\in\bbR^{d^k \times d^k}\). \\
	{\bfseries output}: Factor matrices \(\mB_1,\ldots,\mB_k\in\bbR^{d \times d}\).\\
	\vspace{-1em}
	\begin{algorithmic}[1]
		\STATE Sample \(\cN(\vec0,\mI)\) vectors \(\vg_1,\ldots,\vg_k \in \bbR^{d}\).
		\STATE Let \(\bar\vx = \otimes_{i=1}^k \vg_i\).
		\STATE Compute \(\gamma \defeq \bar\vx^\intercal\mA\bar\vx\).
		\STATE If \(\gamma = 0\), then {\bfseries return} \(\mB_1,\ldots,\mB_k = \mat 0\)
		\STATE For all \(i\in[k], m\in[d]\), let \(\vx_i^{(m)} = (\otimes_{j=1}^{i-1} \vg_i) \otimes \ve_m \otimes (\otimes_{j=i+1}^k \vg_i)\).
		\STATE For all \(i\in[k], m\in[d]\), compute \(\mA\vx_{i}^{(m)}\).
		\STATE For all \(i\in[k], m\in[d], n\in[d]\), store \([\mB_i]_{m, n} = \gamma^{-(1-\frac1k)} {\vx_i^{(m)}}^\intercal\mA\vx_i^{(n)}\).
		\STATE {\bfseries return} \(\mB_1,\ldots,\mB_k\).
	\end{algorithmic}
\end{algorithm}

\begin{theorem}
Let \(\mA = \mA_1 \otimes \cdots \otimes \mA_k \in \bbR^{d^k \times d^k}\).
Then, with probability 1, \Cref{alg:kron-recovery} computes exactly \(kd+1\) Kronecker-matrix-vector products with \mA and returns matrices \(\mB_1,\ldots,\mB_k\) such that \(\mA = \mB_1 \otimes \cdots \otimes \mB_k\).
\end{theorem}
\begin{proof}
At a high level, the quadratic form \({\vx_i^{(m)}}^\intercal \mA\vx_i^{(n)}\) is trying to recover the \((m,n)\) entry of \(\mA_i\).
To see why, notice that 
\[
	{\vx_i^{(m)}}^\intercal \mA \vx_{i}^{(n)}
	= (\otimes_{j=1}^{i-1} \vg_j^\intercal \mA_j \vg_j) \otimes \ve_m^\intercal\mA_i\ve_n \otimes (\otimes_{j=i+1}^{k} \vg_j^\intercal \mA_j \vg_j).
\]
If we let \(\gamma_i \defeq \vg_i^\intercal\mA_i\vg_i\) denote the various components of \(\gamma = \otimes_{i=1}^k \vg_i^\intercal\mA_i\vg_i = \prod_{i=1}^k \gamma_i\), then we can write the above expression as
\[
	{\vx_i^{(m)}}^\intercal \mA \vx_{i}^{(n)}
	= ({\textstyle \prod_{j \neq i} \gamma_i}) \ve_m^\intercal\mA_i\ve_n
	= \tsfrac{\gamma}{\gamma_i} [\mA_i]_{m,n}
\]
So, we get that the matrix \(\mB_i\) has \([\mB_i]_{m,n} = \frac{\gamma^{\frac1k}}{\gamma_i} [\mA_i]_{m,n}\).
We therefore conclude that
\[
	\otimes_{i=1}^k \mB_i
	= \otimes_{i=1}^k \big(\tsfrac{\gamma^{\frac1k}}{\gamma_i} \mA_i\big)
	= \big({\textstyle \prod_{i=1}^k} \tsfrac{\gamma^{\frac1k}}{\gamma_i} \big) \mA
	= \tsfrac{\gamma}{\gamma}\mA
	= \mA.
\]
This analysis implicitly assumes that \(\gamma_i = \vg_i^\intercal\mA\vg_i \neq 0\) for all \(i\in[k]\).
Since \(\vg_i\) are random Gaussian vectors, we know that \(\vg_i\) does not lie in the nullspace of \(\mA_i\) with probability 1 so long as \(\mA_i\) is not the all-zeros matrix.
So, \Cref{alg:kron-recovery} is correct so long as no factor \(\mA_i\) is the all-zeros matrix.
If some \(\mA_i\) is the all-zeros matrix, then we get that \(\mA\) is also the all-zeros matrix, so that \(\gamma = 0\), and therefore our output will have \(\mB_1,\ldots,\mB_k\) all being all-zeros matrices, which is correct.

To count the number of Kronecker-matrix-vector products computed, note that \Cref{alg:kron-recovery} only computes such products on lines 3 and 6.
Line 3 computes exactly 1 product, and line 6 computes exactly \(kd\) products.
\end{proof}

%% file: conclusion.tex
% Clean

\section{Conclusion}
In this paper, we produced new and tight bounds on the rate of convergence for the Kronecker-Hutchinson Estimator.
We show that the estimator generally has an exponential dependence on \(k\), and we showed an expression for the exact variance of the estimator when using the Kronecker of Gaussian vectors.

There are several natural future directions for this research.
It is worth exploring the variance of more classes of matrices under \Cref{thm:real-psd}.
In particular, is there a natural property of matrices which exactly characterizes the input matrices for which the Kronecker-Hutchinson Estimator has a polynomial rate of convergence?
If so, this would be a significant step towards a sub-exponential runtime for trace estimation in the Kronecker-matrix-vector oracle.

It is also worth looking at algorithms in the Kronecker-matrix-vector oracle model beyond the Kronecker-Hutchinson Estimator.
For instance, when \mA has a quickly decaying spectrum, could a low-rank approximation method like ones used in the non-Kronecker case provide good trace estimation guarantees \cite{li2021randomized,saibaba2017randomized}?

There is also potential for newer and stronger lower bounds.
We provide a lower bound against possibly adaptive algorithms, which is very general, but is also perhaps too general to be descriptive of most algorithms people are currently using.
Is there perhaps a lower bound against all sketching-based methods which is exponential in \(k\)?
Or, even more narrowly, is there a lower bound against all quadratic-form based algorithms, in the vain of \cite{wimmer2014optimal}?

Lastly, there is also plenty of potential to study more problems in the Kronecker-matrix-vector oracle model.
Is it possible to accurately and efficiently compute the top eigenvalue or eigenvector of a matrix from Kronecker-matrix-vector products?
How about solving a linear system?
These problems all need to be explored in great detail.

%% file: acks.tex
% Clean

\section*{Acknowledgements}
Raphael Meyer was partially supported by a GAANN fellowship from the US Department of Education.
Haim Avron was partially supported by the US-Israel Binational Science Foundation (Grant no. 2017698). 
We thank Moshe Goldstein and Noa Feldman for discussions on quantum entanglement estimation, which inspired this work.
We thank Tyler Chen for the design of \Cref{alg:kron-recovery}, Feyza Duman Keles for help with tight analysis of Rademacher vectors, and William Swartworth for tightening the analysis of the random rank-1 case.
We also thank Christopher Musco for pointing us to relevant results in \cite{ahle2020oblivious}.